\documentclass[10pt]{article}

\usepackage{amsmath,amstext,amsfonts,amssymb,amsthm,bbold} % For maths bbold for the indicator
\usepackage{graphicx,subfigure,color} % For graphics and color
\usepackage[latin1]{inputenc}   % UNIX, codage ISO 8859-1 (include accents)
\usepackage[english]{babel} % Language

\def \A {\mathbf{A}}

\def \B {\mathbf{B}}

\def \d {\mathbf{d}}
\def \E {\mathbf{E}}
\def \e {\mathbf{e}}

\def \I {\mathbf{I}}

\def \M {\mathbf{M}}
\def \m {\mathbf{m}}

\def \Q {\mathbf{Q}}

\def \R {\mathbf{R}}
\def \s {\mathbf{s}}
\def \S {\mathbf{S}}
\def \T {\mathbf{T}}

\def \u {\mathbf{u}}
\def \v {\mathbf{v}}
\def \V {\mathbf{V}}
\def \W {\mathbf{W}}

\def \x {\mathbf{x}}

\def \Y {\mathbf{Y}}
\def \y {\mathbf{y}}

\def \z {\mathbf{z}}

\def \Ccal {\mathcal{C}}
\def \Dcal {\mathcal{D}}
\def \Ecal {\mathcal{E}}

\def \Kcal {\mathcal{K}}

\def \Ncal {\mathcal{N}}
\def \Ocal {\mathcal{O}}

\def \Rcal {\mathcal{R}}

\def \Cbb {\mathbb{C}}
\def \Ebb {\mathbb{E}}
\def \Nbb {\mathbb{N}}
\def \Pbb {\mathbb{P}}
\def \Rbb {\mathbb{R}}
\def \Vbb {\mathbb{V}}

\def \drm {\mathrm{d}}
\def \erm {\mathrm{e}}
\def \irm {\mathrm{i}}
\def \Prm {\mathrm{P}}
\def \Qrm {\mathrm{Q}}

\def \xibs {\boldsymbol{\xi}}

\def \Deltabs {\boldsymbol{\Delta}}
\def \Gammabs {\boldsymbol{\Gamma}}
\def \Lambdabs {\boldsymbol{\Lambda}}
\def \Omegabs {\boldsymbol{\Omega}}
\def \Pibs {\boldsymbol{\Pi}}
\def \Sigmabs {\boldsymbol{\Sigma}}

\def \Xibs {\boldsymbol{\Xi}}

\def \det {\mathrm{det}}

\def \Tr {\mathrm{tr}\ }

\def \Res {\mathrm{Res}}
\def \diag{\mathrm{diag}}
\def \Vec{\mathrm{Vec}}

\DeclareMathOperator{\supp}{supp}

\renewcommand{\Im}{\mathrm{Im}}
\renewcommand{\Re}{\mathrm{Re}}

\def\restriction#1#2
{
	\mathchoice{\setbox1\hbox{${\displaystyle #1}_{\scriptstyle #2}$}\restrictionaux{#1}{#2}}
	{\setbox1\hbox{${\textstyle #1}_{\scriptstyle #2}$}\restrictionaux{#1}{#2}}
	{\setbox1\hbox{${\scriptstyle #1}_{\scriptscriptstyle #2}$}\restrictionaux{#1}{#2}}
	{\setbox1\hbox{${\scriptscriptstyle #1}_{\scriptscriptstyle #2}$}\restrictionaux{#1}{#2}}
}
\def\restrictionaux#1#2{{#1\,\smash{\vrule height .8\ht1 depth .85\dp1}}_{\,#2}}

\newtheorem{corollary}{Corollary}
\newtheorem{theorem}{Theorem}

\newtheorem{lemma}{Lemma}
\newtheorem{remark}{Remark}

\newtheorem{proposition}{Proposition}

\newcounter{countassum}
\newenvironment{assumption}
{
	\refstepcounter{countassum}
	\begin{flushleft}
	\noindent\textbf{Assumption A-\thecountassum}:
	\it
}
{
	\end{flushleft}
	
}

\makeatletter
\def\blfootnote{\xdef\@thefnmark{}\@footnotetext}
\makeatother

\title{A CLT for an improved subspace estimator with observations of increasing dimensions}
\author{P. Vallet$^{(1)}$, X. Mestre$^{(2)}$, P. Loubaton$^{(3)}$}
\date{}
\begin{document}
\maketitle
\blfootnote
{
	$^{(1)}$ Laboratoire de l'Int\'egration du Mat\'eriau au Syst\`eme (CNRS, Univ. Bordeaux, Bordeaux INP), 
	351, Cours de la Lib\'eration 33405 Talence (France), 
	\textsf{pascal.vallet@bordeaux-inp.fr}	
}
\blfootnote
{
	$^{(2)}$ Centre Tecnol\`{o}gic de Telecomunicacions de Catalunya (CTTC), 
	Av. Carl Friedrich Gauss 08860 Castelldefels, Barcelona (Spain), 
	\textsf{xavier.mestre@cttc.cat}
}
\blfootnote
{
	$^{(3)}$ Universit\'e Paris-Est / Laboratoire d'Informatique Gaspard Monge (CNRS UMR 8049), 5 Bd. Descartes 77454 Marne-la-Vallée (France),
	\textsf{loubaton@univ-mlv.fr}
}
\blfootnote
{
	This work was partially supported by the French programs GDR ISIS/GRETSI 
	"Jeunes Chercheurs" and Project ANR-12-MONU-OOO3 DIONISOS Project.
}

\begin{abstract}
	This paper deals with subspace estimation in the small sample size regime, where the number of samples is comparable in magnitude with the observation dimension.
	The traditional estimators, mostly based on the sample correlation matrix, are known to perform well as long as the number of available samples is much larger
	than the observation dimension. However, in the small sample size regime, the performance degrades. 
	Recently, based on random matrix theory results, a new subspace estimator was introduced, which was shown to be consistent 
	in the asymptotic regime where the number of samples and the observation dimension converge to infinity at the same rate. 
	In practice, this estimator outperforms the traditional ones even for certain scenarios where the observation dimension is small and of the same order of magnitude as
	the number of samples.
	In this paper, we address a performance analysis of this recent estimator, by proving a central limit theorem in the above asymptotic regime.
	We propose an accurate approximation of the mean square error, which can be evaluated numerically. 
\end{abstract}
\tableofcontents

\section{Introduction}
\label{section:introduction}
	
	\subsection{Motivation}

The problem of subspace estimation, i.e. estimating the eigenspaces of the correlation matrix of a certain multivariate time series of dimension $M$, 
available from a set of $N$ noisy observations, is an important problem in statistical signal processing, and covers several topics such as 
Direction of Arrival (DoA) estimation \cite{schmidt1986multiple}, multiuser detection in Code Division Multiple Access (CDMA) \cite{liu1996subspace}, 
chirp parameter estimation \cite{volcker2001chirp} or beamforming \cite{citron1984improved}.
Let us consider an complex $M$-variate time series $(\y_n)_{n \geq 1}$, following a "signal plus noise" model
\begin{align}
	\y_n = \s_n + \v_n,
	\notag
\end{align}
where $\s_n$ corresponds to a signal part and $\v_n$ to a noise part, and assume that $N$ observations $\y_1,\ldots,\y_N$ are collected and stacked in the $M \times N$ matrix
\begin{align}
	\Y_N = [\y_1,\ldots,\y_N] = \S_N + \V_N,
	\notag
\end{align}
with $\S_N=[\s_1,\ldots,\s_N]$ and $\V_N = [\v_1,\ldots,\v_N]$.
In many applications, the signals $(\s_n)_{n \geq 0}$ are moreover constrained to a subspace of dimension $K$ less than $M$ and the matrix $\S_N$ is full rank $K$.
The subspace estimation problem consists in estimating the column space of $\S_N$ called the "signal subspace", of dimension $K$ 
(or equivalently its orthogonal complement called the "noise subspace" of dimension $M-K$) from the observation matrix $\Y_N$.

The usual way of estimating the signal or noise subspaces consists in estimating their orthogonal projection matrices. 
The estimation is performed most of the time by using the so-called sample correlation matrix (SCM) of the observations 
\begin{align}
	\frac{\Y_N\Y_N^*}{N} = \frac{1}{N} \sum_{n=1}^N \y_n\y_n^*,
	\notag
\end{align} 
and these projections are directly estimated by considering 
their sample estimates, i.e. by considering the corresponding orthogonal projection matrix of the SCM. 
For example, the noise subspace projection matrix $\Pibs_N$, i.e.
the orthogonal projection matrix onto the kernel of $\S_N\S_N^*$, is traditionally estimated by $\hat{\Pibs}_N$, the orthogonal projection matrix onto the eigenspace associated with the $M-K$ smallest eigenvalues of $\frac{\Y_N\Y_N^*}{N}$.

These sample estimators are known to perform well when the number of available samples $N$ is much larger than the observation dimension $M$, in particular because
the SCM is a good estimator of the true correlation matrix of the observations. Indeed, in the asymptotic regime where $M$ is constant and
$N$ converges to infinity, under some technical conditions, the law of large numbers ensures that
\begin{align}
	\left\|\hat{\Pibs}_N - \Pibs_N\right\| \xrightarrow[]{} 0 \notag
\end{align}
almost surely as $N \to \infty$, i.e. the sample projection matrices are consistent estimators of the true ones. These sample estimators have been also characterized in terms of Central Limit Theorems (CLT) 
in the previous asymptotic regime, and several accurate approximations of the Mean Square Error (MSE) have been obtained, 
see e.g. Anderson \cite{anderson1958introduction}, Stoica \cite{stoica1989music}, and the references therein.
However, it may exist some situations where obtaining such an amount of samples is not conceivable, for example in situations where the signals are stationnary only for
a short period of time, or simply if the observation dimension is large.
As a consequence, in the low sample size regime where $M$ and $N$ are of the same order of magnitude, the performance of the sample subspace estimators severely degrades,
essentially because the SCM does not estimate properly the true correlation matrix.

In this context, based on recent results in random matrix theory, a new subspace estimator was proposed by Mestre \cite{mestre2008improved}, in the case where $(\s_n)_{n\geq 0}$ 
and $(\v_n)_{n \geq 0}$ are modeled as two independent zero-mean Gaussian stationnary processes, temporally uncorrelated, with the signal correlation matrix $\R_s = \Ebb[\s_n\s_n^*]$ being rank $K$ and the noise covariance being equal to $\sigma^2 \I$, where $\sigma > 0$ and $\I$ is the $M \times M$ identity matrix, i.e. $(\y_n)_{n \geq 0}$ can be modeled equivalently as
\begin{align}
	\y_n = \left(\R_s + \sigma^2 \I\right)^{1/2} \x_n,
\end{align}
with $(\x_n)_{n \geq 0}$ a standard spatially and temporally white Gaussian process.
Later Vallet et al. \cite{vallet2012improved} obtained, using the same approach, a different estimator in the more general situation where the signals $(\s_n)_{n \geq 0}$ are considered as unknown 
deterministic. 
The estimators of  \cite{mestre2008improved} and \cite{vallet2012improved} were shown to be consistent in the asymptotic regime where both the observation dimension $M$ and the number of samples 
$N$ converge to infinity in such a way that the ratio $\frac{M}{N}$ converge to a positive constant.
Moreover, these estimators do not assume any particular assumption on the behaviour of the rank $K$, which may also converge to infinity with $M,N$.
In practice, these estimators outperform the traditional ones, when $M,N$ are of the same order of magnitude.
Based on these results, an application to DoA estimation of $K$ source signals impinging on an array of $M$ sensors was proposed, and an improved subspace DoA estimator called
G-MUSIC (Generalized MUltiple Signal Classification) was built, which was shown to numerically outperform the traditionnal MUSIC estimator, for realistic values of $M,N$.
This DoA estimator was also shown to be consistent in Hachem et al. \cite{hachem2012large}.

Recently, Hachem et al. \cite{hachem2012subspace} proposed an analysis of the subspace estimator \cite{vallet2012improved}, in terms of a Central Limit Theorem (CLT), in the previous asymptotic
regime where $M,N$ converge to infinity at the same rate, and by assuming that the rank $K$ is fixed. 
In practice, these results are accurate as long as the rank $K$ remains small compared to $M,N$. 
However, when the rank $K$ is of the same order of magnitude than the dimension $M$ and $N$, the corresponding results do not predict anymore the behaviour of 
the subspace estimator \cite{vallet2012improved} and the results of  \cite{hachem2012subspace} are not very accurate.

In this paper, 
\footnote{The material of this paper was party presented in the conference paper \cite{mestre2011asymptotic}.}
we propose to extend the analysis of \cite{hachem2012subspace} regarless the behaviour of the rank $K$, which may increase with $M,N$. For that purpose, we use a different approach and prove a Central Limit Theorem (CLT) in the previous asymptotic regime. We also provide an explicit expression for the Mean Square Error (MSE) which can be easily evaluated numerically. Numerical examples confirm the validity of the results.

The paper is organized as follows.
In the remainder of section \ref{section:introduction}, we introduce formally the model of signals used in the paper, and recall some basic results from random matrix theory, necessary for the next sections. In section \ref{section:subspace_estimation}, we introduce the subspace estimator of \cite{vallet2012improved} and provide the main result of the paper, namely a CLT for this estimator, as well as numerical illustrations.
Sections \ref{section:proof_clt} and \ref{section:appendix} contain the proofs of the results.

	\subsection{Notations}
	
We introduce here the main notations used throughout the paper.

The sets $\Rbb$, $\Rbb^+$ and $\Nbb$ (resp. $\Nbb^*$) will respectively represent the real numbers, the non-negative numbers and the non-negative integers (resp. the positive integers). 
$\Cbb$ will be the set of complex numbers, and for $z \in \Cbb$, $\Re(z)$,  $\Im(z)$ and $z^*$ will stand 
for the real part, the imaginary part and the complex conjuguate. $\irm$ will be the imaginary unit and we will also use the set $\Cbb^+ = \{z \in \Cbb: \Im(z)>0\}$.
The indicator of a set $\Ecal \subset \Rbb$ is denoted $\mathbb{1}_{\Ecal}$, $\partial \Ecal$ and $\mathrm{Int}(\Ecal)$ will denote the boundary and interior of $\Ecal$.

For a real-valued function $\varphi$ defined on $\Rbb$, $\supp(\varphi)$ will represent the support of $\varphi$, and $\Ccal_c^{\infty}(\Rbb,\Ecal)$ will the set of smooth compactly supported 
functions defined on $\Rbb$, taking values in some set $\Ecal \subset \Rbb$.

Matrices (respectively vectors) are denoted by bolfaced capital (respectively boldfaced lower case) letters. For a complex matrix $\A$, we denote by $\A^T, \A^*$ its transpose and 
its conjuguate transpose, and by $\Tr(\A)$ and $\|\A\|$ its trace and spectral norm. The identity matrix will be $\I$. $\e_n$ will refer to a vector having 
all its components equal to $0$ except the $n$-th equals to $1$.

The real normal distribution with mean $m$ and variance $\sigma^2$ is denoted $\Ncal_{\Rbb}(\alpha,\sigma^2)$ and the multivariate normal distribution in $\Rbb^k$, with mean $\m$ 
and covariance $\Gammabs$ is denoted in the same way $\Ncal_{\Rbb^k}(\m,\Gammabs)$. We will say that a complex random variable $Z = X +\irm Y$ follow the distribution
$\Ncal_{\Cbb}(\alpha+\irm\beta,\sigma^2)$ if $X$ and $Y$ are independent with respective distributions $\Ncal_{\Rbb}(\alpha, \frac{\sigma^2}{2})$ and $\Ncal_{\Rbb}(\beta, \frac{\sigma^2}{2})$.
The expectation and variance of a complex random variable $Z$ will be denoted $\Ebb[Z]$ and $\Vbb[Z]$.
The support of a probability measure $\mu$ will be denoted $\supp(\mu)$. For a sequence of random variables $(X_n)_{n \in \Nbb}$ and a random variable $X$, we write
\begin{align}
	X_n \xrightarrow[n\to\infty]{a.s.} X \text{ and } X_n \xrightarrow[n\to\infty]{\Dcal} X
	\notag
\end{align}
when $X_n$ converges respectively with probability one and in distribution to $X$. Finally, $X_n = o_{\Pbb}(1)$ will stand for the convergence of $X_n$ to $0$ in probability, and
$X_n = \Ocal_{\Pbb}(1)$ will stand for boundedness in probability (tightness).

Some other special notations may be used at some very localized parts in the paper, and will be introduced in the text.

	\section{Asymptotic behaviour of the sample eigenvalues}
	\label{section:review_rmt}
	
In this section, we present some basic results from random matrix theory, describing the behaviour of the eigenvalues of the SCM $\frac{\Y_N\Y_N^*}{N}$, in the asymptotic regime where $M,N$ converge to infinity such that $\frac{M}{N} \to c > 0$. 
These results will be required to introduce the improved subspace estimator of \cite{vallet2012improved}.
To that end, we will work with the following random matrix model, refered to as "Information plus Noise" in the literature.
We consider $M,N,K \in \Nbb^*$ such that $K < M < N$ and $M=M(N)$, $K=K(N)$ are functions of $N$ satisfying $c_N = \frac{M}{N} \to c \in (0,1)$ as $N \to \infty$.
For each $N \in \Nbb^*$, we consider the $M \times N$ random matrix $\Sigmabs_N$, defined by
\begin{align}
	\Sigmabs_N = \B_N + \W_N,
	\label{eq:INmodel}
\end{align}
with
\begin{itemize}
 \item $\B_N$ a rank $K$ deterministic matrix satisfying $\sup_{N} \|\B_N\| < \infty$,
 \item $\W_N$ having i.i.d. entries $W_{i,j} \sim \Ncal_{\Cbb}\left(0,\frac{\sigma^2}{N}\right)$.
\end{itemize}
We denote by $\lambda_{1,N}> \ldots > \lambda_{K,N} > \lambda_{K+1,N} = \ldots = \lambda_{M,N} = 0$ the eigenvalues of $\B_N\B_N^*$ (the non-zero eigenvalues are assumed 
to have multiplicity one for simplicity), and by $\u_{1,N},\ldots,\u_{M,N}$ the respective unit norm eigenvectors.
Equivalently, $\hat{\lambda}_{1,N}\geq\ldots\geq\hat{\lambda}_{M,N}$ are the eigenvalues of the matrix $\Sigmabs_N\Sigmabs_N^*$ and $\hat{\u}_{1,N},\ldots,\hat{\u}_{M,N}$ 
the respective unit norm eigenvectors.

		\subsection{The asymptotic spectral distribution}
		\label{section:asymptotic_spectral_distribution}
		
Let $\hat{\mu}_N = \frac{1}{M} \sum_{k=1}^M \delta_{\hat{\lambda}_{k,N}}$ the empirical spectral measure of the matrix $\Sigmabs_N\Sigmabs_N^*$, with $\delta_x$ the Dirac
measure at point $x$.
From Dozier \& Silverstein \cite{dozier2007empirical} \cite{dozier2007analysis}, there exists a deterministic probability measure $\mu_N$, 
with support $\supp(\mu_N) \subset \Rbb^+$, such that w.p.1.,
\begin{align}
	\hat{\mu}_N - \mu_N \xrightarrow[N \to \infty]{w} 0,
	\notag
\end{align}
where "$w$" stands for the weak convergence.
Equivalently the Stieltjes transform $\hat{m}_N(z)$ of $\hat{\mu}_N$, defined by
\begin{align}
	\hat{m}_N(z) = \int_{\Rbb} \frac{\drm \hat{\mu}_N(\lambda)}{\lambda - z} = \frac{1}{M} \Tr \Q_N(z)
	\notag
\end{align}
where $\Q_N(z) = \left(\Sigmabs_N\Sigmabs_N^* - z \I\right)^{-1}$ satisfies for all $z \in \Cbb \backslash \Rbb$
\begin{align}
	\hat{m}_N(z) - m_N(z) \xrightarrow[N \to \infty]{a.s.} 0,
	\notag
\end{align}
where $m_N(z) = \int_{\Rbb} \frac{\drm \mu_N(\lambda)}{\lambda - z}$ is the Stieltjes transform of $\mu_N$, which satisfies the equation
\begin{align}
	m_N(z) = \frac{1}{M} \Tr \T_N(z),
	\label{def:m}
\end{align}
for all $z \in \Cbb\backslash\Rbb$, where the matrix $\T_N(z)$ is defined by
\begin{align}
	\T_N(z) = \left(\frac{\B_N\B_N^*}{1+\sigma^2 c_N m_N(z)} - z (1+\sigma^2 c_N m_N(z))\I + \sigma^2 (1-c_N)\I \right)^{-1}.
	\notag
\end{align}
Moreover, $m_N(z)$ can be further continuously extended to the real axis when $z \in \Cbb^+ \to x \in \Rbb$, and we denote the limit $m_N(x)$. 
Defined in this way, $x \mapsto m_N(x)$ is continuous on $\Rbb$, continuously differentiable on $\Rbb\backslash \partial \supp(\mu_N)$ and still satisfies
the equation \eqref{def:m} for $x \in \Rbb\backslash \partial \supp(\mu_N)$. 

We now recall the characterization of the support of $\mu_N$ provided in \cite{vallet2012improved}.
Define the function $w_N(z)$ by
\begin{align}
	w_N(z) = z \left(1+\sigma^2 c_N m_N(z)\right)^2 - \sigma^2 (1-c_N)\left(1+\sigma^2 c_N m_N(z)\right),
	\label{def:w}
\end{align}
The main equation \eqref{def:m} can be expressed in terms of an equation in $w_N(z)$, i.e.
\begin{align}
	z = \phi_N(w_N(z)),
	\label{eq:phi_w_z}
\end{align}
where
\begin{align}
	\phi_N(w)=w(1-\sigma^2 c_N f_N(w))^2 + \sigma^2(1-c_N)(1-\sigma^2 c_N f_N(w))
	\label{def:phi}
\end{align}
and 
\begin{align}
	f_N(w) = \frac{1}{M}\Tr\left(\B_N\B_N^* - w \I\right)^{-1}. 
	\label{def:f}
\end{align}
Starting from the properties that $w_N$ is real and increasing on $\Rbb\backslash\supp(\mu_N)$ and $w_N(x) \in \Cbb^+$ for $x \in \supp(\mu_N)$, \cite{vallet2012improved} 
characterized $w_N(x)$ among the set of all solutions of the polynomial equation $\phi_N(w)=x$ (which has degree $2K+2$), for $x \in \Rbb$, 
and showed that $\phi_N$ admits $2Q$ ($1 \leq Q \leq K+1$) positive local extrema
\footnote{Note that $Q=Q(N)$ is a function of $N$} $0 < x_{1,N}^- < x_{1,N}^+ < \ldots < x_{Q,N}^- < x_{Q,N}^+$  whose preimages are
\begin{align}
	w_{N}(x_{1,N}^-) < 0 < w_N(x_{1,N}^+) < \ldots < w_N(x_{Q,N}^-) < w_N(x_{Q,N}^+).
	\label{eq:ordering_wq}
\end{align}
Moreover, we always have $w_N(x_{Q,N}^+) > \lambda_{1,N}$, and if $Q > 1$, it turns out that for each $q =1 ,\ldots, Q-1$, there exists $k \in \{0,\ldots,K\}$ such that
\begin{align}
	w_N(x_{q,N}^+), w_N(x_{q+1,N}^-) \in \left(\lambda_{k,N}, \lambda_{k+1,N}\right).
	\notag
\end{align}
By differentiating \eqref{eq:phi_w_z} on both sides, we find $\phi'_N(w_N(x)) > 0$ for all $x \in \Rbb \backslash \supp(\mu_N)$.
Finally, by showing that $\Im\left(w_N(x)\right)=0$ for $x \in \Rbb \backslash \bigcup_{q=1}^Q [x_{q,N}^-,x_{q,N}^+]$ and $\Im\left(w_N(x)\right) > 0$
for $x \in \bigcup_{q=1}^Q [x_{q,N}^-,x_{q,N}^+]$, \cite{vallet2012improved} concludes that the support of $\mu_N$ is given by the union
\begin{align}
	\supp(\mu_N) = \bigcup_{q=1}^Q \left[x_{q,N}^-, x_{q,N}^+\right],
	\label{def:supp_mu}
\end{align}
where the intervals $\left[x_{q,N}^-, x_{q,N}^+\right]$ are called "clusters".

A typical illustration of function $\phi_N(w)$ for $w\in \Rbb$ is given in figure \ref{figure:phi1}.
\begin{center}
	\begin{figure}[h]
		\centering
		\includegraphics[scale=1]{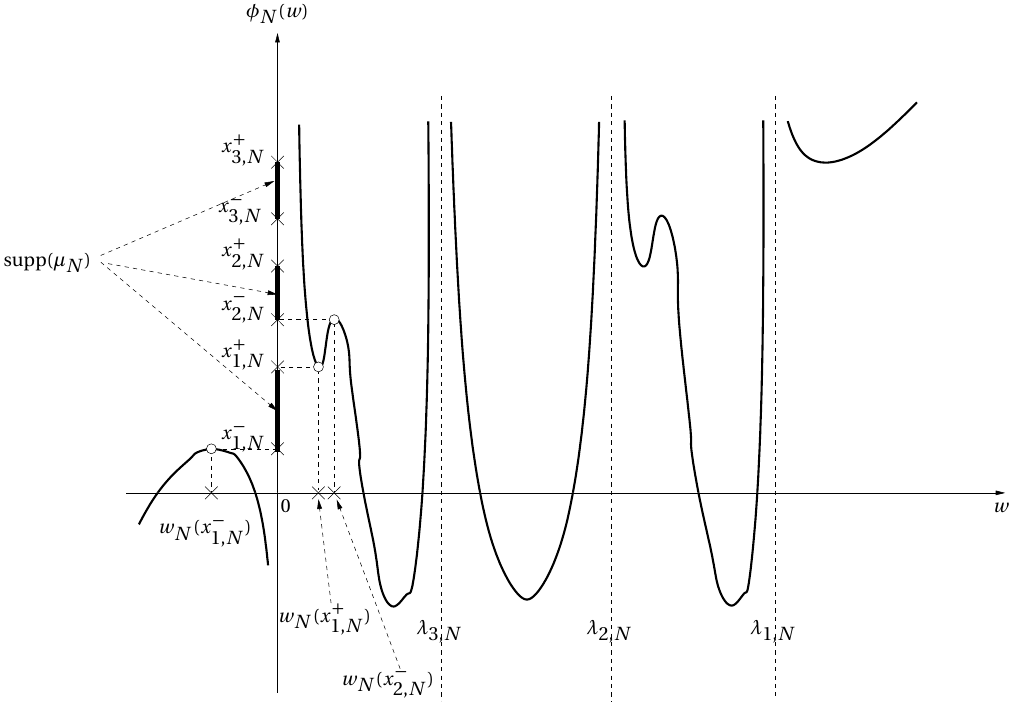}
		\caption
		{
			Typical example of the function $\phi_N$, for $K=3$.	 For $x \in (x_{q,N}^+, x_{q+1,N}^-)$ with $q=1,\ldots,Q-1$, 
			the equation $\phi_N(w)=x$ admits $2K+2$ real solutions and $w_N(x)$ is the unique solution in the interval $(w_N(x_{q,N}^+), w_N(x_{q+1,N}^-))$, 
			and for $x < x_{1,N}^-$ (resp. $x > x_{Q,N}^+$), $w_N(x)$ is the unique solution in $(- \infty, w_N(x_{1,N}^-))$ (resp. $(w_N(x_{Q,N}^+),\infty)$). 
			Conversely, for $x \in \supp(\mu_N)$, the equation $\phi_N(w)=x$ admits $2K$ real solutions plus two complex
			conjugate solutions, and $w_N(x)$ coincides with the solution having positive imaginary part.
		}
		\label{figure:phi1}
	\end{figure}
\end{center}

		\subsection{Useful quantities and bounds}
		
We now introduce a few bounds which will be of constant use for the derivation of the main results of the paper.
Let us define
\begin{align}
	\tilde{m}_N(z) = c_N m_N(z) - \frac{1-c_N}{z},
	\label{def:mtilde}
\end{align}
which corresponds to the Stieltjes transform of the probability measure $c_N \mu_N + (1-c_N) \delta_0$. 
It can be shown that $\tilde{m}_N(z) = \frac{1}{N} \Tr \tilde{\T}_N(z)$, with
\begin{align}
	\tilde{\T}_N(z) = 
	\left(
		\frac{\B_N^*\B_N}{1+\sigma^2 \tilde{m}_N(z)} - z (1+\sigma^2 c_N m_N(z)) \I
	\right)^{-1},
	\label{def:Ttilde}
\end{align}
and note that $w_N(z)$ defined in \eqref{def:w} can be written as
\begin{align}
	w_N(z) = z \left(1+\sigma^2 c_N m_N(z)\right)\left(1+\sigma^2 \tilde{m}_N(z)\right).
	\notag
\end{align}
The proof of the following bounds can be found in \cite{loubaton2011almost}, \cite{vallet2012improved} and \cite{hachem2012large}: matrices $\T_N(z)$ and $\tilde{\T}_N(z)$ satisfy
\begin{align}
	\left\|\T_N(z)\right\| \leq \frac{C}{\drm\left(z_,\supp(\mu_N)\right)}
	\quad\text{and}\quad
	\left\|\tilde{\T}_N(z)\right\| \leq \frac{\tilde{C}}{\drm\left(z_,\supp(\mu_N)\cup \{0\}\right)},
	\label{eq:bound_norm_T_Ttilde}
\end{align}
where $\drm\left(z, \Ecal\right)$ is the distance of $z$ to a set $\Ecal$, and $C, \tilde{C}$ are two positive constants independent of $N,z$. 
We also have, for all $z \in \Cbb$,
\begin{align}
	\left|1+\sigma^2 c_N m_N(z)\right|^{-1} \leq 2,
	\label{eq:bound_m_half}
\end{align}
and
\begin{align}
	\min_{k=1,\ldots,M} \left|\lambda_{k,N} - w_N(z)\right| \geq \frac{\drm\left(z,\supp(\mu_N)\right)}{2}.
	\label{eq:bound_dist_w_lambda}
\end{align}
Note finally the two useful identities
\begin{align}
	1+\sigma^2 c_N m_N(z) = \frac{1}{1-\sigma^2 c_N f_N(w_N(z))} \text{ and } 1+\sigma^2 c_N \tilde{m}_N(z) = \frac{1}{1-\sigma^2 c_N \tilde{f}_N(w_N(z))},
	\label{eq:link_m_f}
\end{align}
where $f_N$ is defined by \eqref{def:f} and $\tilde{f}_N(w) = \frac{1}{M} \Tr\left(\B^*_N\B_N - w \I\right)^{-1}$.

To conclude this section, we introduce some quantities which will appear during the computations of the CLT.
We define
\begin{align}
	u_N(z_1,z_2) = \frac{\sigma^2}{N} \Tr \frac{\T_N(z_1)\B_N\B_N^*\T_N(z_2)}{\left(1+\sigma^2 c_N m_N(z_1)\right)\left(1+\sigma^2 c_N m_N(z_2)\right)}
	\label{def:u}
\end{align}
as well as
\begin{align}
	v_N(z_1,z_2) = \frac{\sigma^2}{N} \Tr \T_N(z_1) \T_N(z_2)
	\quad\text{and}\quad
	\tilde{v}_N(z_1,z_2) = \frac{\sigma^2}{N} \Tr \tilde{\T}_N(z_1) \tilde{\T}_N(z_2).
	\label{def:v_vtilde}	
\end{align}
Finally, we define
\begin{align}
	\Delta_N(z_1,z_2) = \left(1-u_N(z_1,z_2)\right)^2 - z_1 z_2 v_N(z_1,z_2)\tilde{v}_N(z_1,z_2).
	\label{def:determinant}
\end{align}
The last quantity $\Delta_N(z_1,z_2)$ satisfies moreover the following bounds.
\begin{lemma}
	\label{lemma:properties_determinant}
	For all $z_1,z_2 \not \in \supp(\mu_N)$ such that $z_1 \neq z_2$, $\Delta_N(z_1,z_2)$ can be represented as
	\begin{align}
		\Delta_N(z_1,z_2) = \frac{z_1-z_2}{w_N(z_1)-w_N(z_2)}.
		\label{eq:link_Delta_w}
	\end{align}
	Moreover, if there exists a closed set $\Ecal$ independent of $N$ s.t. $\supp(\mu_N) \subset \Ecal$ for all large $N$, and 
	if $\Kcal$ is a compact set s.t. $\Kcal \subset \Cbb\backslash\left(\{0\} \cup \Ecal\right)$, then
	\begin{align}
		\limsup_{N\to \infty} \sup_{z_1,z_2 \in \Kcal} \left|u_N(z_1,z_2)\right| < 1
		\label{eq:bound_u}
	\end{align}
	and	
	\begin{align}
		0 < \liminf_{N \to \infty} \inf_{z_1,z_2 \in \Kcal} \left|\Delta_N(z_1,z_2)\right|
		\leq \limsup_{N\to \infty} \sup_{z_1,z_2 \in \Kcal} \left|\Delta_N(z_1,z_2)\right| < \infty.
		\label{eq:bound_det}
	\end{align}
	Finally, we also have
	\begin{align}
		\left|\frac{\Delta_N(z_1,z_2)}{\left(1-u_N(z_1,z_2)\right)^2} - 1\right| < 1,
		\label{eq:bound_det_serie}
	\end{align}
	for all $z_1,z_2 \in \Kcal$.
\end{lemma}
Lemma \ref{lemma:properties_determinant} is proved in appendix \ref{appendix:determinant}.

		\subsection{Separation of the sample eigenvalues}

In this section, we review some existing results concerning the location of the sample eigenvalues.
		
The following terminology will be used in the remainder: an eigenvalue $\lambda_{k,N}$ of $\B_N\B_N^*$ is \textit{associated with the interval} $[x_{q,N}^-,x_{q,N}^+]$ of the support of $\mu_N$ if $w_N(x_{q,N}^-) <\lambda_{k,N} < w_N(x_{q,N}^+)$. 
It turns out that the "noise eigenvalue" $0$ is always associated with the first interval $[x_{1,N}^-,x_{1,N}^+]$ since $w_N(x_{1,N}^-) < 0 < w_N(x_{1,N}^+)$ 
(see \eqref{eq:ordering_wq}), which is thus called in this context "noise cluster". 
Moreover, each "signal eigenvalue" $\lambda_{1,N},\ldots,\lambda_{K,N}$ is associated with a unique interval $[x_{q,N}^-,x_{q,N}^+]$ for $q=1,\ldots,Q$ ; 
in particular, a signal eigenvalue may be associated with the "noise cluster" while two signal eigenvalues may be associated with the same interval.
We now introduce the two following additional assumptions, which informally ensure that the $K$ signal eigenvalues $\lambda_{1,N},\ldots,\lambda_{K,N}$
will not be associated with the noise cluster, that is, will be separated from the "noise eigenvalue" $0$ ($\lambda_{K+1,N},\ldots,\lambda_{M,N}$) for large $N$.
This assumption  will be necessary to guaranty the consistency of the subspace estimator introduced in the forthcoming sections.
\begin{assumption}
	\label{assumption:subspace1}
	There exists $t_1^-,t_1^+,t_2^-,t_2^+$ s.t.
	\begin{align}
		0 < t_1^- < \liminf_{N \to \infty} x_{1,N}^- < \limsup_{N \to \infty} x_{1,N}^+ < t_1^+ 
		< t_2^- < \liminf_{N\to \infty} x_{2,N}^- < \limsup_{N \to \infty} x_{Q,N}^+ < t_2^+.
		\notag
	\end{align}
\end{assumption}
Assumption {\bf A-\ref{assumption:subspace1}} thus ensures that the noise cluster remains asymptotically separated from the the other intervals in the support of $\mu_N$, as $N \to \infty$.
From \eqref{eq:ordering_wq} and the fact that $w_N$ is increasing on $\Rbb \backslash \mu_N$, we have $w_N(t_1^-) < 0 < w_N(t_1^+) < w_N(t_2^-)$ for all large $N$.
The second assumption {\bf A-\ref{assumption:subspace2}} is related to the behaviour of the signal eigenvalues. 
\begin{assumption}
	\label{assumption:subspace2}
	For all large $N$, $0$ is the unique eigenvalue of $\B_N\B_N^*$ associated with the noise cluster, i.e.
	\begin{align}
		 w_N(t_2^-) < \lambda_{K,N}.
		\notag
	\end{align}
\end{assumption}
Note that this assumption implies that $\liminf_{N\to\infty} \lambda_{K,N} > 0$, thus ensuring that noise and signal eigenvalues are asymptotically separated 
(see lemma \ref{lemma:conseq_sep} below).
These separation conditions have a direct consequence on the localization of the eigenvalues of the matrix $\Sigmabs_N\Sigmabs_N^*$.
Indeed, it was shown in \cite{vallet2012improved} that under assumptions {\bf A-\ref{assumption:subspace1}} and {\bf A-\ref{assumption:subspace2}},
\begin{align}
	\hat{\lambda}_{1,N},\ldots,\hat{\lambda}_{M-K,N} \in [t_1^-,t_1^+]
	\quad\text{and}\quad
	\hat{\lambda}_{M-K+1,N},\ldots,\hat{\lambda}_{M,N} \in [t_2^-,t_2^+],
	\label{eq:separation_lambdahat}
\end{align}
with probability one, for $N$ large, i.e. the "noise sample eigenvalues" split from the "signal sample eigenvalues".
An illustration of the density of $\mu_N$ and the localization of the sample eigenvalues \eqref{eq:separation_lambdahat} is given in figure \ref{figure:loceig}.

\begin{figure}[!h]
	\centering
	\includegraphics[scale=0.4]{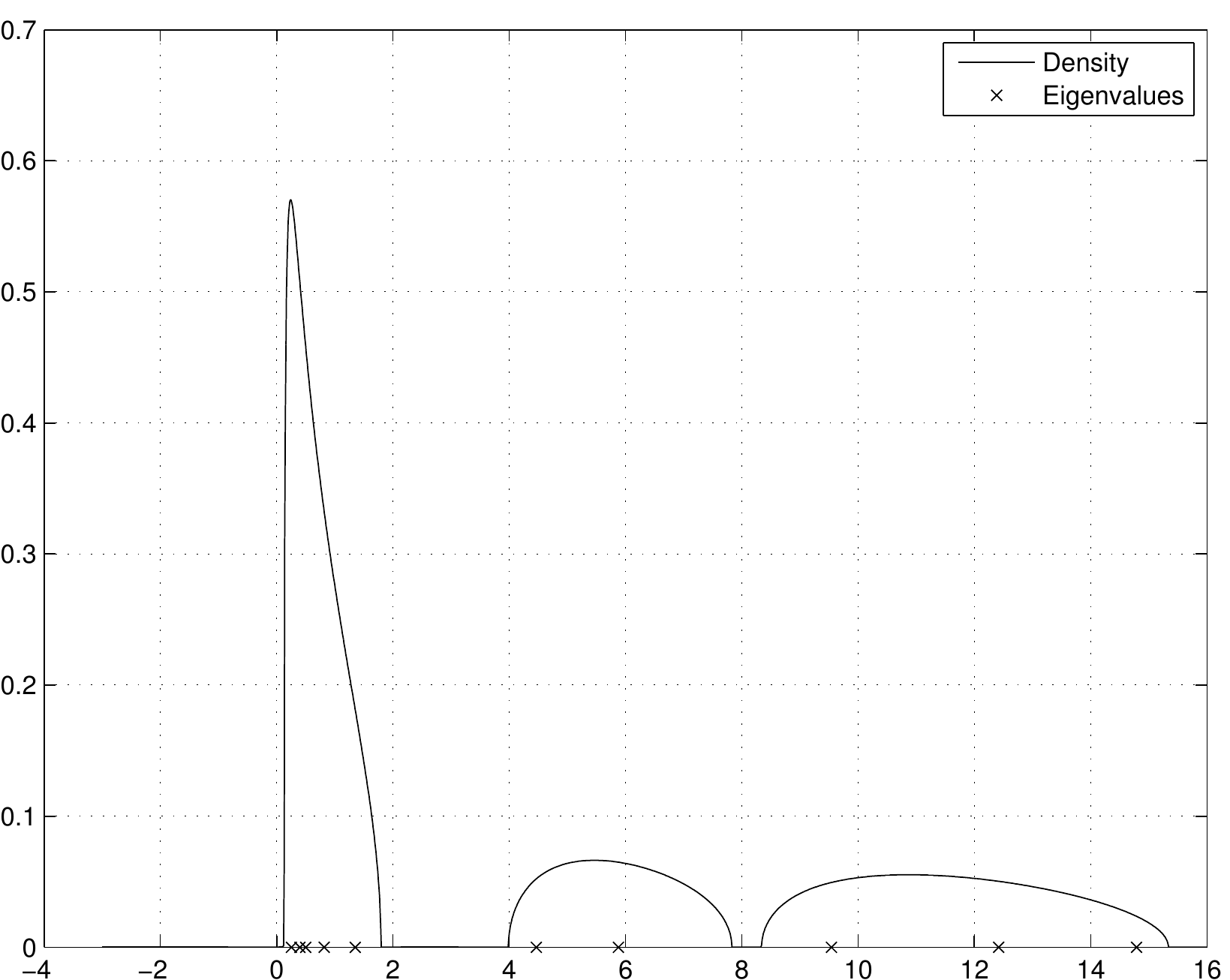}
	\caption
	{
		Density of $\mu_N$ and sample eigenvalues of $\Sigmabs_N\Sigmabs_N^*$ for one trial. The parameters are $M=10$, $N=20$, $\sigma=1$ and the eigenvalues
		of $\B_N\B_N^*$ are $0$ (with multiplicity $5$), $5$ (with multiplicity $2$) and $10$ (with multiplicity $3$)
	}
	\label{figure:loceig}
\end{figure}

Functions $\phi_N$ represents in some sense a link between the support of $\mu_N$ and the eigenvalues of $\B_N\B_N^*$.
% As discussed in section \ref{section:model_signals}, the separation assumption {\bf A-\ref{assumption:subspace2}} ensures that $0$ is the unique eigenvalue associated 
% with the first cluster $[x_{1,N}^-,x_{1,N}^+]$, i.e. that $\lambda_{M-K+1,N} > w_N(x_{2,N}^-) > w_N(t_2^-)$.
Figure \ref{figure:phi2} shows the consequence of assumption {\bf A-\ref{assumption:subspace2}} on the behaviour of $\phi_N(w)$ near $w=0$.
\begin{center}
	\begin{figure}[h]
		\centering
		\subfigure[Separation]
		{
			\includegraphics[scale=1.2]{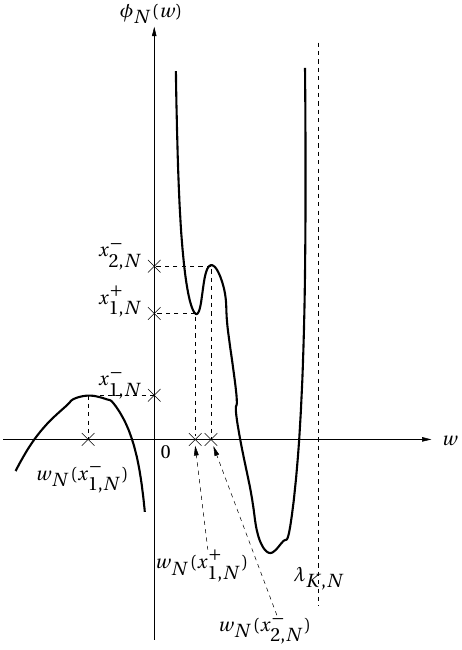}
			\label{subfigure:sep}
		}
		\qquad\qquad
		\subfigure[No separation]
		{
			\includegraphics[scale=1.2]{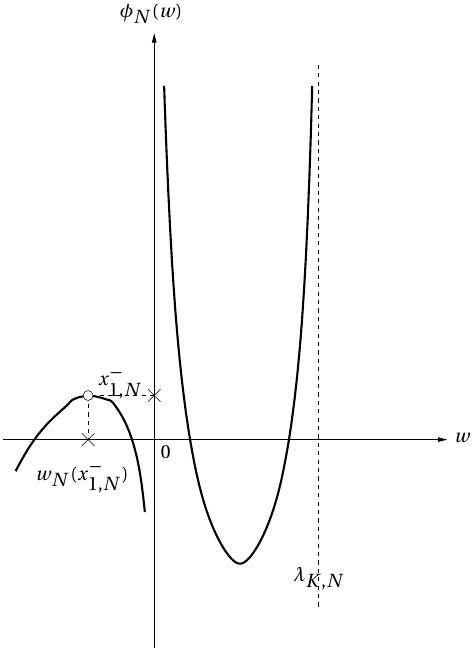}
			\label{subfigure:nosep}
		}
		\caption
		{
			Typical example of the behaviour of $\phi_N$ near $0$, when assumption {\bf A-\ref{assumption:subspace2}} is satified \subref{subfigure:sep}, 
			or not \subref{subfigure:nosep}.
		}
		\label{figure:phi2}
	\end{figure}
\end{center}
To conclude this part, we show that the separation conditions {\bf A-\ref{assumption:subspace1}} and {\bf A-\ref{assumption:subspace2}}
ensure the effective separation between signal and noise eigenvalues of $\B_N\B_N^*$.
\begin{lemma}
	\label{lemma:conseq_sep}
	Assume the separation conditions {\bf A-\ref{assumption:subspace1}} and {\bf A-\ref{assumption:subspace2}} hold. Then,
	\begin{align}
		\liminf_{N \to \infty} \lambda_{K,N} > 0.
		\notag
	\end{align}
\end{lemma}
\begin{proof}
	Assume the converse. Then there exists a subsequence $\varphi(N)$ such that $\lambda_{K,\varphi(N)} \to_N 0$.
	From the condition {\bf A-\ref{assumption:subspace2}}, we have for all large $N$
	\begin{align}
		w_N(t_1^-) < 0 < w_N(t_1^+) < w_N(t_2^-) < \lambda_{K,N},
		\notag
	\end{align}
	and the condition {\bf A-\ref{assumption:subspace1}} ensures the existence of $x,y \in (t_1^+,t_2^-)$, with $x < y$, such that
	\begin{align}
		w_{\varphi(N)}(y) - w_{\varphi(N)}(x) \xrightarrow[N \to \infty]{} 0.
		\label{eq:diff_w}
	\end{align}
	But using \eqref{eq:link_Delta_w} and \eqref{eq:bound_det} in lemma \ref{lemma:properties_determinant} contradicts \eqref{eq:diff_w}.
\end{proof}

	\subsection{The spiked model case: fixed rank}
	\label{section:spiked_model}
	
When $K$ is constant with respect to $N$, the results of the previous sections can be simplified. Indeed, in this case, we have for all $z \in \Cbb\backslash\Rbb^+$,
\begin{align}
	m_N(z) - m(z) \xrightarrow[N \to \infty]{} 0,
	\label{eq:conv_MP}
\end{align}
where $m(z)$ satisfies the equation
\begin{align}
	m(z) = \frac{1}{-z\left(1+\sigma^2 c m(z)\right) + \sigma^2 (1-c)},
	\label{eq:def_m_check_N}
\end{align}
and is the Stieltjes transform of the Marchenko-Pastur distribution \cite{Marchenko1967distribution}, with support $[\sigma^2(1-\sqrt{c})^2,\sigma^2(1+\sqrt{c})^2]$.
An illustration of the Marchenko-Pastur distribution is given in figure \ref{figure:densityMP}.
\begin{center}
	\begin{figure}[!h]
		\centering
		\includegraphics[scale=0.4]{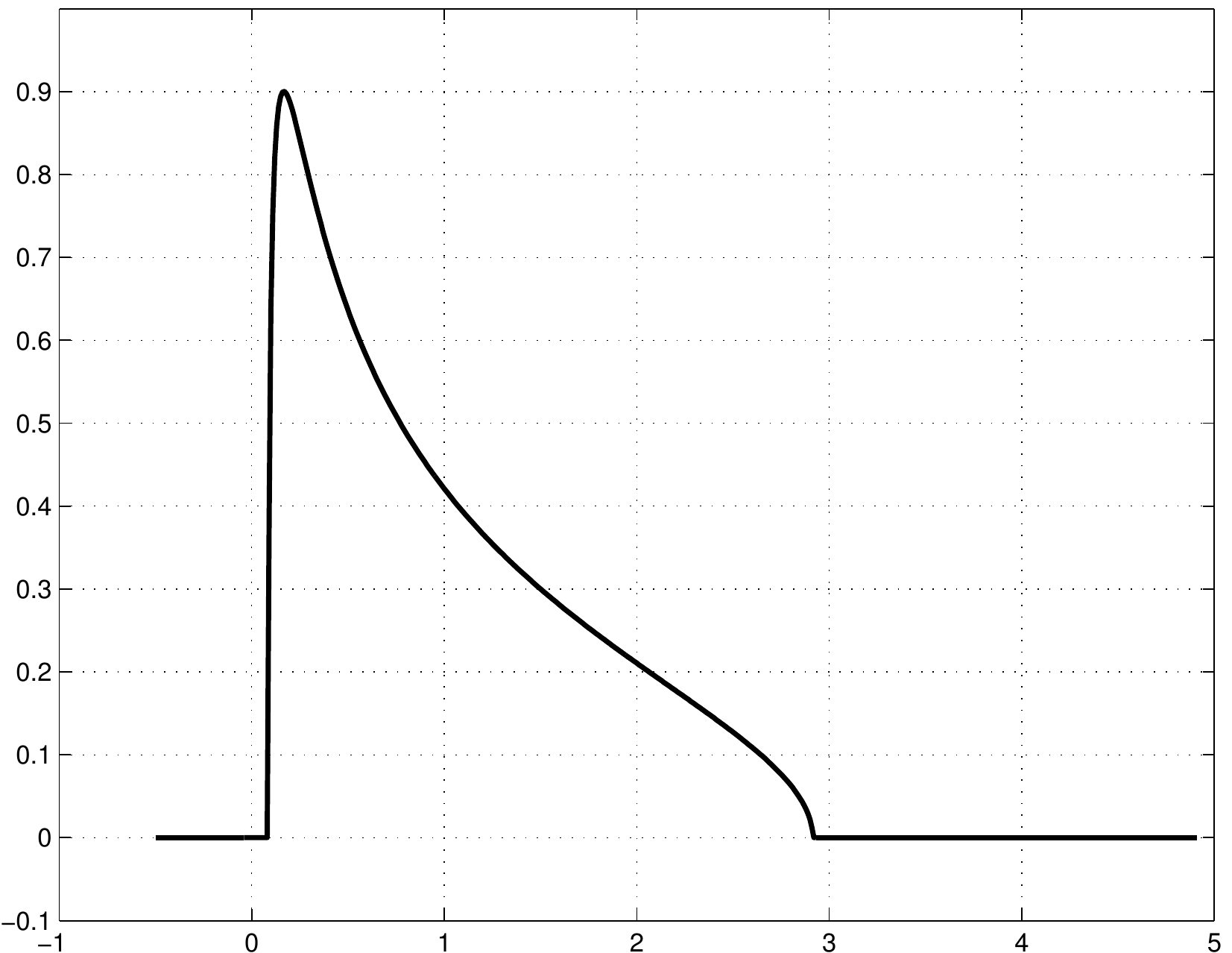}
		\caption
		{
			Density of the Marchenko-Pastur distribution 
			in the case where $c = 0.5$, $\sigma=1$.
		}
		\label{figure:densityMP}
	\end{figure}
\end{center}
For any compact $\Kcal \in \Cbb \backslash \left(\left[t_1^-,t_1^+\right]\cup\left[t_2^-,t_2^+\right]\right)$, \eqref{eq:conv_MP} can be strengthened with
\begin{align}
	\sup_{z \in \Kcal} \left|m_N(z) - m(z)\right| \xrightarrow[N \to \infty]{} 0,
	\label{eq:conv_MP2}
\end{align}
Simple algebra allows to rewrite the usual quantities in a simpler way. Indeed, we will have (in the same way as for \eqref{eq:link_m_f})
\begin{align}
	1+\sigma^2 c m(z) = \frac{w(z)}{w(z) + \sigma^2 c},
	\notag
\end{align}
with $w(z)$ given by
\begin{align}
	w(z)=z\left(1+\sigma^2 c m(z)\right)^2 + \sigma^2 (1-c)\left(1+\sigma^2 c m(z)\right).
	\label{eq:def_w_check_N}
\end{align}
As for \eqref{eq:phi_w_z}, equation \eqref{eq:def_m_check_N} can be rewritten as
\begin{align}
	\phi(w(z))=z,
	\label{eq:canonical_eq_mp}
\end{align}
where 
\begin{align}
	\phi(w) = \frac{\left(w+\sigma^2 c\right)\left(w+\sigma^2\right)}{w}.
	\label{eq:phi_MP}
\end{align}
Of course, the boundary points of the support of the Marchenko-Pastur distribution, namely $\sigma^2(1-\sqrt{c})^2$ and $\sigma^2(1+\sqrt{c})^2$, 
coincides with the local extrema of $\phi$, and with respective preimages 
\begin{align}
	w\left(\sigma^2(1-\sqrt{c})^2\right)=-\sigma^2 \sqrt{c} \text{ and } w\left(\sigma^2(1+\sqrt{c})^2\right)=\sigma^2 \sqrt{c}. 
	\notag
\end{align}
The function $w$ is continuous on $\Rbb$, real and increasing on $\Rbb \backslash [\sigma^2 (1-\sqrt{c})^2,\sigma^2(1+\sqrt{c})^2]$.
Figure \eqref{fig:spikesupport} provides an illustration of the behaviour of the density of $\mu_N$ when $K$ is fixed and $N \to \infty$.
\begin{figure}[!h]
        \begin{center}
	\subfigure[$N=20$]{\includegraphics[scale=0.4]{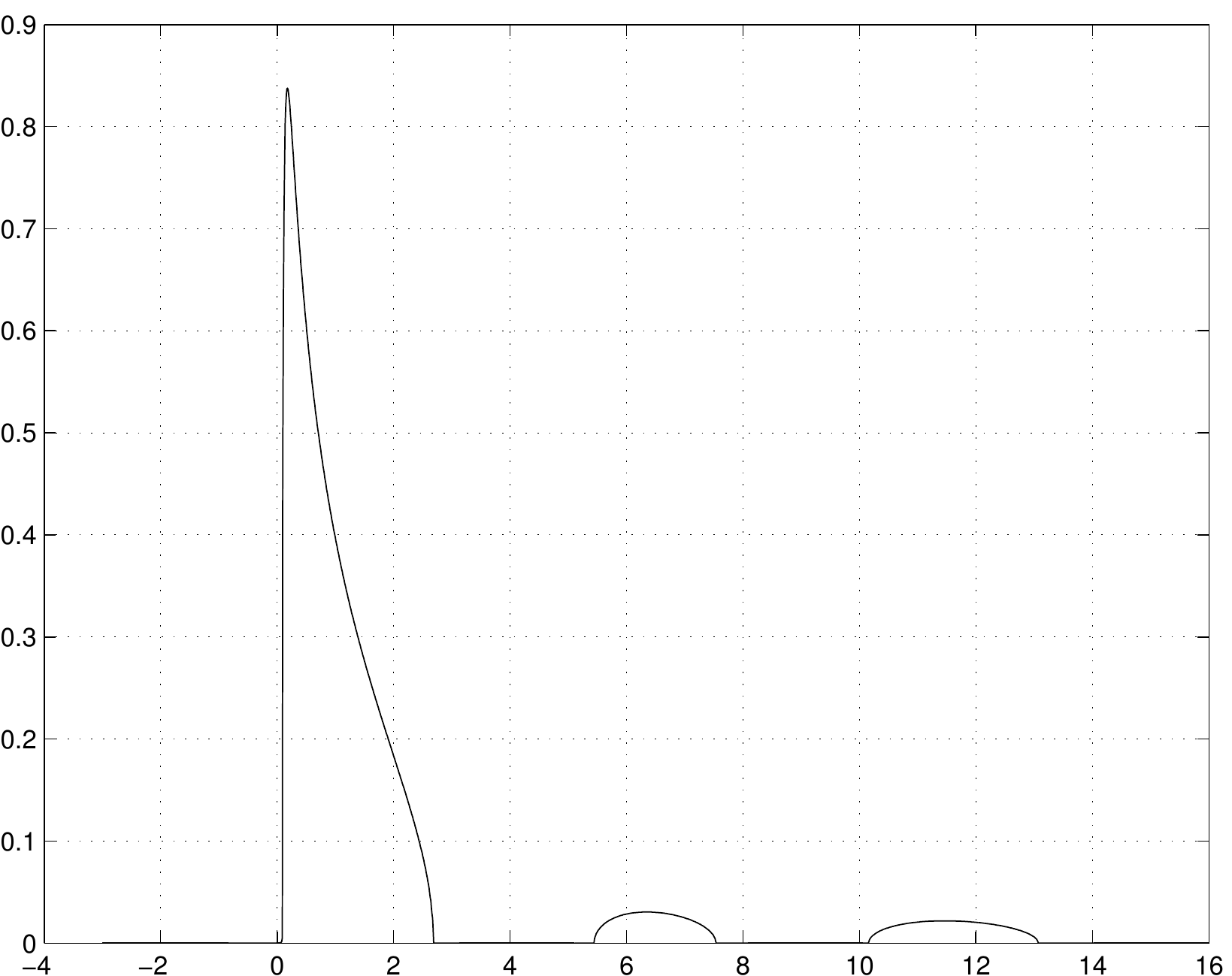}}
	\subfigure[$N=100$]{\includegraphics[scale=0.4]{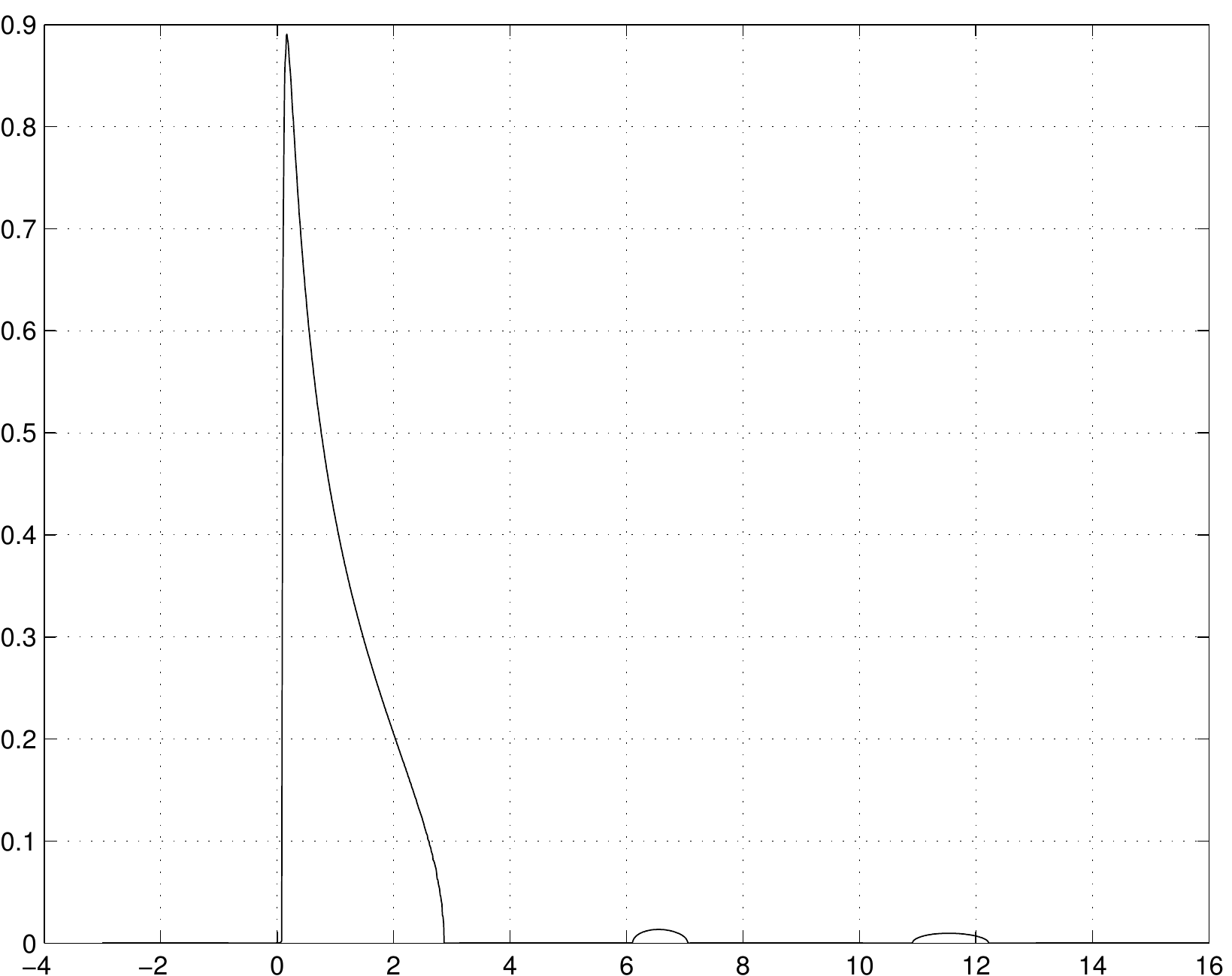}}\\
	\subfigure[$N=200$]{\includegraphics[scale=0.4]{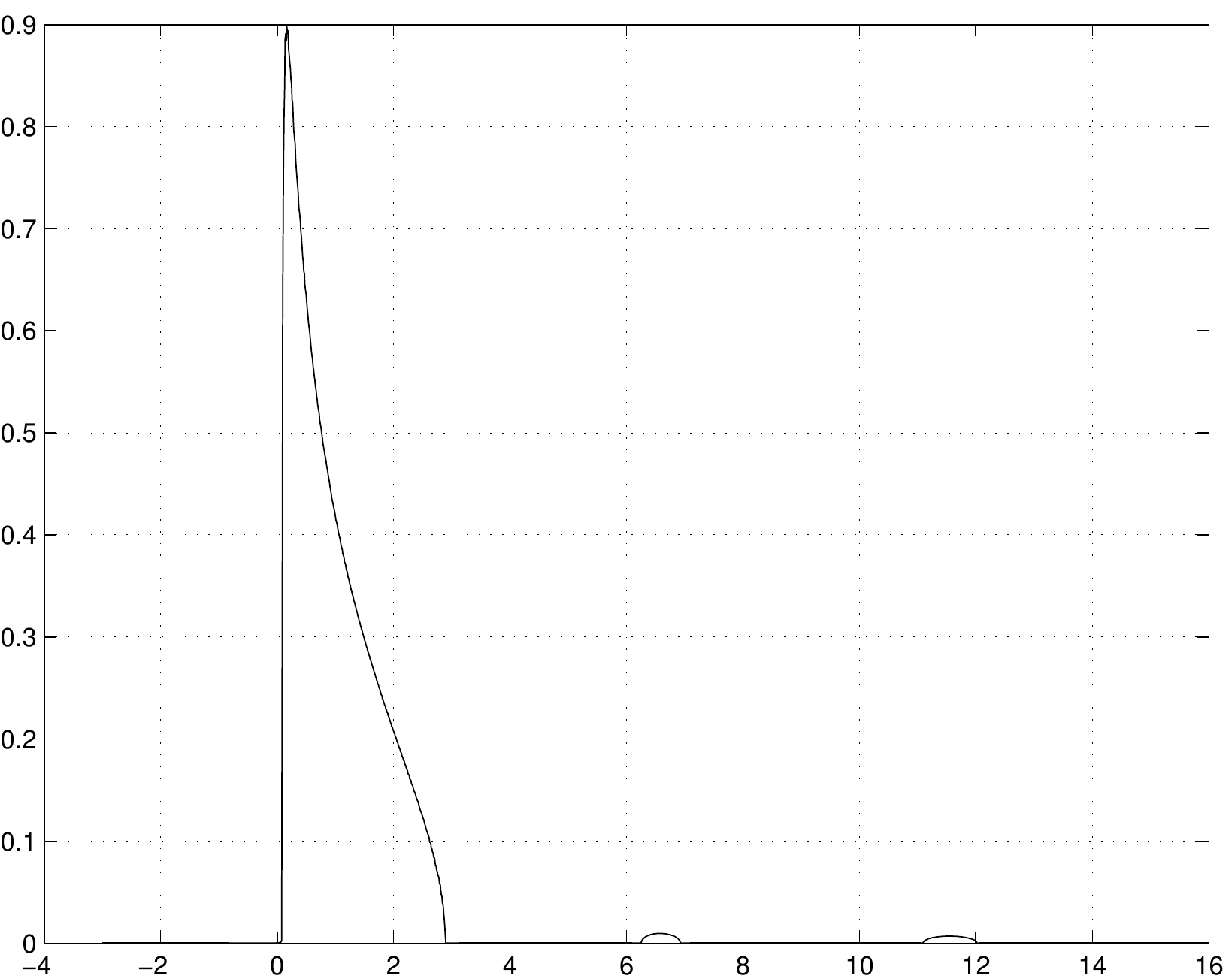}}
	\subfigure[$N=2000$]{\includegraphics[scale=0.4]{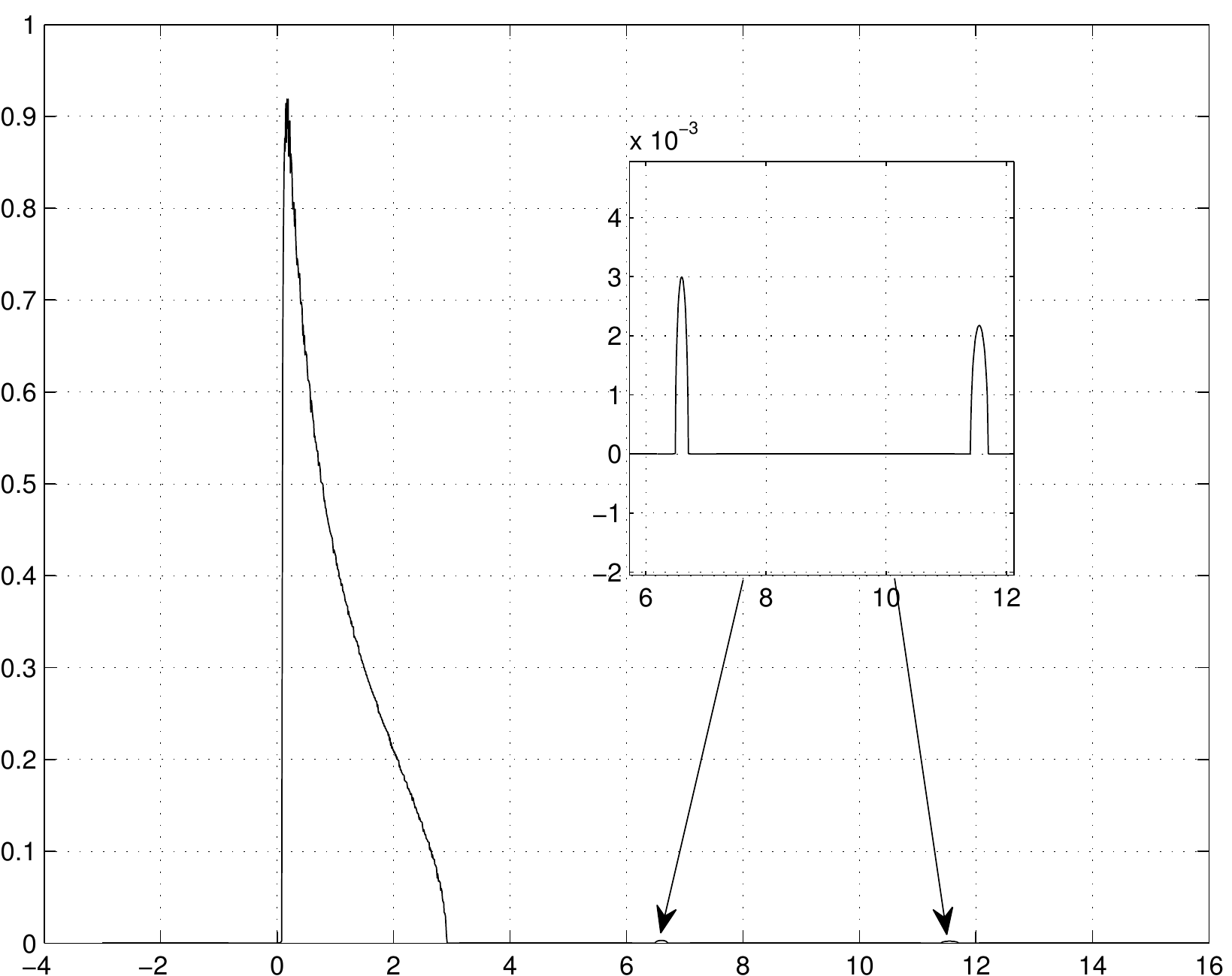}}
        \end{center}
        \caption{Density of $\mu_N$ when $K=2$, $c_N=0.5$. The two non zero eigenvalues of $\B_N\B_N^*$ are $5$ and $10$.}
        \label{fig:spikesupport}
\end{figure}
In the case when $K$ is fixed, the separation assumptions {\bf A-\ref{assumption:subspace1}} and {\bf A-\ref{assumption:subspace2}} have a useful consequence on the behaviour of the smallest non zero eigenvalue of $\B_N\B_N^*$.
\begin{lemma}
	\label{lemma:conseq_sep2}
	Assume $K$ independent of $N$. Then, the separation condition {\bf A-\ref{assumption:subspace1}} and {\bf A-\ref{assumption:subspace2}} hold iff 
	\begin{align}
		\liminf_{N \to \infty} \lambda_{K,N} > \sigma^2 \sqrt{c}.
		\label{eq:sep_cond_spike}
	\end{align}
\end{lemma}
The proof of lemma \ref{lemma:conseq_sep2} is defered in Appendix \ref{appendix:proof_lemma_conseq_sep_2}.

In the special situation where the non-zero eigenvalues of $\B_N\B_N^*$ converge to some different limits, i.e. 
\begin{align}
	\lambda_{k,N} \to_N \lambda_{k} > 0
\end{align}
for all $k=1,\ldots,K$, with $\sigma^2 \sqrt{c} < \lambda_K < \ldots < \lambda_1$, it is shown in \cite{loubaton2011almost} that the number $Q$ of clusters 
in the support of $\mu_N$ is exactly $K+1$ for $N$ large and in this case, the "noise" eigenvalue $0$ is the unique eigenvalue associated with the "noise" cluster 
$[x_{1,N}^-,x_{1,N}^+]$.
Therefore, assumptions {\bf A-\ref{assumption:subspace1}} and {\bf A-\ref{assumption:subspace2}} are ensured in this case.
It is also proved that $x_{k+1,N}^{\pm} \to \phi(\lambda_k)$, and using a refinement of \eqref{eq:separation_lambdahat}, 
\cite{loubaton2011almost} also showed that (see also Benaych \& Nadakuditi \cite{benaych2011singular}) that the $K$ largest sample eigenvalues 
$\hat{\lambda}_{1,N},\ldots,\hat{\lambda}_{K,N}$ split from the $M-K$ smallest eigenvalues $\hat{\lambda}_{K+1,N},\ldots,\hat{\lambda}_{M,N}$ and
\begin{align}
	\hat{\lambda}_{k,N} \xrightarrow[N \to \infty]{a.s} \phi(\lambda_k),
	\label{eq:lim_spike}
\end{align}
while $\hat{\lambda}_{K+1,N} \to_N \sigma^2 (1+\sqrt{c})^2$ and $\hat{\lambda}_{M,N} \to_N \sigma^2 (1-\sqrt{c})^2$ a.s.

Finally, we notice that $\Delta_N(z_1,z_2)$ defined in \eqref{def:determinant} will satisfy, when $K$ is constant
\begin{align}
	\sup_{z_1,z_2 \in \Kcal} \left|\Delta_N(z_1,z_2) - \Delta(z_1,z_2)\right| \xrightarrow[N \to \infty]{} 0,
	\notag
\end{align}
where $\Delta(z_1,z_2)$ is given by
\begin{align}
	\Delta(z_1,z_2) = 1 - \frac{\sigma^4 c}{w(z_1)w(z_2)}.
	\label{def:determinant_MP}
\end{align}
The properties given in lemma \ref{lemma:properties_determinant} are of course valid for $\Delta(z_1,z_2)$: in particular, we have
\begin{align}
	\Delta(z_1,z_2) = \frac{z_1-z_2}{w(z_1)-w(z_2)},
	\notag
\end{align}
for all $z_1 \neq z_2$ and $z_1,z_2 \not\in [\sigma^2(1-\sqrt{c})^2,\sigma^2(1+\sqrt{c})^2]$, as well as
\begin{align}
	0 < \inf_{z_1,z_2 \in \Kcal} \left|\Delta(z_1,z_2)\right|
	\leq \sup_{z_1,z_2 \in \Kcal} \left|\Delta(z_1,z_2)\right| < \infty
	\notag
\end{align}
and
\begin{align}
	\left|\Delta(z_1,z_2) - 1\right| < 1,
	\notag
\end{align}
for all $z_1,z_2 \in \Kcal$ with $\Kcal$ a compact set such that 
$\Kcal \subset \Cbb \backslash\left(\{0\} \cup [\sigma^2 (1-\sqrt{c})^2,\sigma^2 (1+\sqrt{c})^2]\right)$.

	\subsection{Contour integrals}
	\label{section:contour_integrals}

Thoughout the paper, we will deal with integrals of the form
\begin{align}
	I_N = \frac{1}{2 \pi \irm} \int_{\partial \Rcal} \Psi_N(w_N(z)) w'_N(z) \drm z 
\end{align}
where $\partial \Rcal$ is the clockwise oriented boundary of a rectangle $\Rcal$ intersecting the real axis at two points $t_2^- - \epsilon$, $t_2^+ + \epsilon$
with $\epsilon > 0$ such that $t_2^- > t_1^+ + \epsilon$, and where $\Psi_N$ is a meromorphic function with poles contained in the set 
$\{\lambda_{1,N},\ldots,\lambda_{K,N},0\}$. 
It is shown in \cite{vallet2012improved} from assumptions {\bf A-\ref{assumption:subspace1}} and {\bf A-\ref{assumption:subspace2}} that the set $w_N(\partial \Rcal)$ is a closed piecewise $\Ccal^1$ path intersecting the real axis at points $w_N(t_2^--\epsilon)$, $w_N(t_2^+ +  \epsilon)$, enclosing the non-zero eigenvalues of $\B_N\B_N^*$ with winding number $-1$ and leaving $0$ outside, for $N$ large. Therefore, for all large $N$, a change of variable and residue theorem lead to
\begin{align}
	I_N = \frac{1}{2 \pi \irm} \oint_{w_N(\partial \Rcal)} \Psi_N(w) \drm w = - \sum_{k=1}^K \Res\left(\Psi_N,\lambda_{k,N}\right),
\end{align}
where $\Res(\Psi_N,\lambda)$ is the residue of $\Psi_N$ at $\lambda$.	
Note that in the case of the spiked models (see section \ref{section:spiked_model}) where $K$ is fixed with respect to $N$, and under 
assumptions {\bf A-\ref{assumption:subspace1}} and {\bf A-\ref{assumption:subspace2}}, the previous result still holds by replacing $w_N(z)$
with $w(z)$ and $w'_N(z)$ with $w'(z)$.

\section{Noise subspace estimation}
\label{section:subspace_estimation}

In this section, we review the results of \cite{vallet2012improved}, \cite{hachem2012subspace} on the consistent subspace estimation, in the asymptotic regime where the number of antennas $M=M(N)$ is a function of the number of samples $N$ such that $c_N=\frac{M}{N} \to c \in (0,1)$ as $N \to \infty$. 

	\subsection{Consistent estimation}
	
Noise subspace estimation consists in our case in estimating the quantity
\begin{align}
	\eta_N = \d_{1,N}^* \Pibs_N \d_{2,N} = \d_{1,N}^* \left(\I - \sum_{k=1}^{K} \u_{k,N} \u_{k,N}^* \right) \d_{2,N},
	\label{eq:eta}
\end{align}
where $(\d_{1,N})$, $(\d_{2,N})$ are two sequences of deterministic vectors such that $\sup_N \|\d_{1,N}\|, \sup_N\|\d_{2,N}\| < \infty$.

We recall that the traditional estimator based on the SCM $\frac{\Y_N\Y_N^*}{N}$ is defined by
\begin{align}
	\hat{\eta}_N^{(t)} = \d_{1,N}^* \left(\I - \sum_{k=1}^{K} \hat{\u}_{k,N} \hat{\u}_{k,N}^*\right) \d_{2,N}.
	\label{eq:traditional_localization_function_estimator}
\end{align}		
It was shown that under the separation assumptions {\bf A-\ref{assumption:subspace1}} and {\bf A-\ref{assumption:subspace2}}, the quantity \eqref{eq:eta} 
can be written in terms of the following integral
\begin{align}
	\eta_N = 
	\d_{1,N}^*\left(\I - \frac{1}{2 \pi \irm}\oint_{\partial \Rcal}  \T_N(z) \frac{w'_N(z)}{1+\sigma^2 c_N m_N(z)} \drm z\right) \d_{2,N},
	\label{eq:eta_contour_integral}
\end{align}
where $\partial \Rcal$ is the clockwise oriented boundary of the rectangle
\begin{align}
	\Rcal = \left\{x + \irm y: x \in [t_2^- - \epsilon, t_2^+ + \epsilon], y \in [-\delta,\delta]\right\},
	\label{def:rectangle}
\end{align}
with $\epsilon>0$ s.t. $t_1^+ + \epsilon < t_2^-$ and $\delta > 0$.
By defining 
\begin{align}
	\hat{w}_N(z) = z (1+\sigma^2 c_N \hat{m}_N(z))^2 - \sigma^2 (1+\sigma^2 c_N \hat{m}_N(z)), 
	\label{def:what}
\end{align}
it is shown in \cite{vallet2012improved} that
\begin{align}
	\sup_{z \in \partial \Rcal} 
	\left|\d_{1,N}^* \Q_N(z) \d_{2,N} \frac{\hat{w}'_N(z)}{1+\sigma^2 c_N \hat{m}_N(z)}-\d_{1,N}^* \T_N(z) \d_{2,N} \frac{w'_N(z)}{1+\sigma^2 c_N m_N(z)}\right|
	\xrightarrow[N \to \infty]{a.s.} 0.
	\label{eq:uniform_conv_integrand}
\end{align}
This of course readily implies that 
\begin{align}
	\hat{\eta}_N - \eta_N \xrightarrow[N\to\infty]{a.s.} 0,
	\label{eq:consistency_subspace}
\end{align}
where
\begin{align}
	\hat{\eta}_N 
	= \d_{1,N}^*\left(\I - \frac{1}{2 \pi \irm}\oint_{\partial \Rcal}  \Q_N(z) \frac{\hat{w}'_N(z)}{1+\sigma^2 c_N \hat{m}_N(z)} \drm z\right) \d_{2,N}.
	\label{eq:subspace_estimator}
\end{align}
Thus $\hat{\eta}_N$ is a consistent estimator of \eqref{eq:eta}.

\begin{remark}
	\label{remark:subspace_estimator}
	The integrand in \eqref{eq:subspace_estimator} is meromorphic with poles at $\hat{\lambda}_{1,N},\ldots,\hat{\lambda}_{M,N}$ as well as at the zeros of the function
	$z \mapsto 1+\sigma^2 c_N \hat{m}_N(z)$, denoted $\hat{\omega}_{1,N},\ldots,\hat{\omega}_{M,N}$. 
	It is shown in \cite{vallet2012improved} that these zeros are the eigenvalues of the matrix 
	$\hat{\Omegabs}_N  = \hat{\Lambdabs}_N + \frac{\sigma^2 c_N}{M} \mathbf{1} \mathbf{1}^T$, 
	where $\hat{\Lambdabs}_N = \diag\left(\hat{\lambda}_{1,N},\ldots,\hat{\lambda}_{M,N}\right)$, and follow a  property  similar to \eqref{eq:separation_lambdahat}, i.e
	\begin{align}
		\hat{\omega}_{1,N},\ldots,\hat{\omega}_{M-K,N} \in [t_1^-,t_1^+]
		\quad\text{and}\quad
		\hat{\omega}_{M-K+1,N},\ldots,\hat{\omega}_{M-K,N} \in [t_2^-,t_2^+],
		\label{eq:separation_omegahat}
	\end{align}
	with probability one, for $N$ large enough.
	This ensures that the integral can be solved using residue theorem, and an explicit formula in terms of $\hat{\u}_{k,N}, \hat{\lambda}_{k,N}$ and $\hat{\omega}_{k,N}$ was provided
	in \cite{vallet2012improved} for the improved subspace estimator \eqref{eq:subspace_estimator}.
\end{remark}

\begin{remark}
	Originally, the estimator derived in \cite{vallet2012improved} was based on the representation
	\begin{align}
		\eta_N = 
		\frac{1}{2 \pi \irm}\oint_{\partial \tilde{\Rcal}}  \d_{1,N}^*\T_N(z)\d_{2,N} \frac{w'_N(z)}{1+\sigma^2 c_N m_N(z)} \drm z,
		\notag
	\end{align}
	where $\partial \tilde{\Rcal}$ is the clockwise oriented boundary of the rectangle
	\begin{align}
		\tilde{\Rcal} = \left\{x + \irm y: x \in [t_1^- - \epsilon, t_1^+ + \epsilon], y \in [-\delta,\delta]\right\},
		\notag
	\end{align}
	enclosing the noise cluster (the contour \eqref{eq:eta_contour_integral} enclosing the signal cluster).
	In that case, \eqref{eq:uniform_conv_integrand} with $\Rcal$ replaced by $\tilde{\Rcal}$ still holds, and 
	\begin{align}
		\hat{\eta}_N 
		= \frac{1}{2 \pi \irm}\oint_{\partial \Rcal}  \d_{1,N}^*\Q_N(z)\d_{2,N} \frac{\hat{w}'_N(z)}{1+\sigma^2 c_N \hat{m}_N(z)} \drm z.
		\notag
	\end{align}
	Therefore, the subspace estimators of \cite{vallet2012improved} and \eqref{eq:subspace_estimator} coincide. 
	We choose to keep the representation \eqref{eq:subspace_estimator} (with contour enclosing the signal cluster), 
	which will be simpler to analyze in the following.
\end{remark}

	When $K$ is constant, a simpler estimator of the localization function can be obtained \cite{hachem2012subspace} \cite[Sec. C]{vallet2012improved}. 
	Indeed, since 
	\begin{align}
		\hat{\eta}_N = 
		\d_{1,N}^* \left( \I - \frac{1}{2 \pi \irm}\oint_{\partial \Rcal} \Q_N(z) \frac{w'(z)}{1+\sigma^2 c m(z)} \drm z \right) \d_{2,N} + o(1),
		\notag
	\end{align}
	a straigthforward application of residue theorem leads
	\begin{align}
		\hat{\eta}_N = \hat{\eta}_N^{(s)} + o(1)
		\notag
	\end{align}
	with probability one, where
	\begin{align}		
		\hat{\eta}_N^{(s)} = \d_{1,N}^* \left( \I - \sum_{k=1}^K h\left(\hat{\lambda}_{k,N}\right) \hat{\u}_{k,N}\hat{\u}_{k,N} \right) \d_{2,N},
		\label{eq:simplified_eta}
	\end{align}
	and where $h(z)$ is given by
	\begin{align}
		h(z) = \frac{w'(z)}{1+\sigma^2 c m(z)} = \frac{w(z) (w(z)+\sigma^2 c)}{w(z)^2 - \sigma^4 c}.
		\notag
	\end{align}

% Indeed, in this case, from the results of section \ref{section:spiked_model}, we have
% \begin{align}
% 	\sup_{z \in \Rcal} \left|\hat{m}_N(z) - \check{m}(z)\right| \xrightarrow[N\to\infty]{\mathrm{a.s.}} 0
% 	\quad\text{and}\quad
% 	\sup_{z \in \Rcal} \left|\hat{w}'_N(z) - \check{w}'(z)\right| \xrightarrow[N\to\infty]{\mathrm{a.s.}} 0,
% \end{align}
% from which it can be deduced that, with probability one,
% \begin{align}
% 	\hat{\eta}_N = 
% 	\d_{1,N}^* 
% 	\left(
% 		1 - \sum_{k=1}^K \frac{h'\left(\hat{\lambda}_{M-K+k,N}\right)}{m\left(\hat{\lambda}_{M-K+k,N}\right)h\left(\hat{\lambda}_{M-K+k,N}\right)}
% 	\right)
% 	\d_{2,N} 
% 	+ o(1),
% 	\label{eq:subspace_estimator_spiked}
% \end{align}
% with $h(z) = \check{m}(z)\left( c z \check{m}(z) - (1-c) \right)$.
% Note that \eqref{eq:subspace_estimator_spiked} is well-defined from \eqref{eq:separation_lambdahat_spiked_model}.

	\subsection{CLT}
	\label{section:CLT_bilinear_form}

		\subsubsection{The main result}
		
Before stating the main result, we need to introduce some new quantities. 
For $k,\ell \in \left\{1,\ldots,M\right\}$, let 
\begin{align}
	& \vartheta_N(k,\ell) = 
	\notag\\
	&\frac{\sigma^2}{2}\left(\frac{1}{2 \pi \irm}\right)^2 \oint_{\partial \Rcal}\oint_{\partial \Rcal}
	\frac{\theta_{N}^{(k,\ell)}(z_1,z_2) w'_N(z_1)w'_N(z_2)}
	{
		\left(\lambda_{k,N} - w_N(z_1)\right)\left(\lambda_{\ell,N} - w_N(z_1)\right)\left(\lambda_{k,N} - w_N(z_2)\right)\left(\lambda_{\ell,N} - w_N(z_2)\right)
		\Delta_N(z_1,z_2)
	}
	 \drm z_1 \drm z_2,
	 \label{eq:def_vartheta}
\end{align}
with
\begin{align}
	&\theta_{N}^{(k,\ell)}(z_1,z_2) = 
	\notag\\
	&\qquad\qquad
	z_1 z_2  \left(1+\sigma^2 c_N m_N(z_1)\right)\left(1+\sigma^2 c_N m_N(z_2)\right) \tilde{v}_N(z_1,z_2)
	\notag\\
	&\qquad\qquad\qquad
	+\frac{\lambda_{k,N}\lambda_{\ell,N} v_N(z_1,z_2)}{\left(1+\sigma^2 c_N m_N(z_1)\right)\left(1+\sigma^2 c_N m_N(z_2)\right)}
	+(\lambda_{k,N}+\lambda_{\ell,N}) \left(1-u_N(z_1,z_2)\right).
	\label{eq:def_theta}
\end{align}
We define the $2 \times 2$ matrix $\Gammabs_N(k,\ell)$ by
\begin{align}
	&\Gammabs_N(k,\ell) =
	\notag\\
&	\begin{bmatrix}
		\Re\left(\eta^{(1,2)}_{k,N}\eta^{(1,2)}_{\ell,N}\right)
		+
		\frac{1}{2}\left(\eta^{(1,1)}_{k,N}\eta^{(2,2)}_{\ell,N} + \eta^{(1,1)}_{\ell,N}\eta^{(2,2)}_{k,N}\right)
		&
		-\Im\left(\eta^{(1,2)}_{k,N}\eta^{(1,2)}_{\ell,N}\right)
		\\
		-\Im\left(\eta^{(1,2)}_{k,N}\eta^{(1,2)}_{\ell,N}\right)
		&
		-\Re\left(\eta^{(1,2)}_{k,N}\eta^{(1,2)}_{\ell,N}\right)
		+
		\frac{1}{2}\left(\eta^{(1,1)}_{k,N}\eta^{(2,2)}_{l,N} + \eta^{(1,1)}_{\ell,N}\eta^{(2,2)}_{k,N}\right)
	\end{bmatrix},
	\notag
\end{align}
where $\eta^{(i,j)}_{k,N} = \d_{i,N}^* \u_{k,N}\u_{k,N}^* \d_{j,N}$, and we finally set
\begin{align}
	\Gammabs_N = \sum_{k=1}^M\sum_{l=1}^M \vartheta_{N}(k,\ell) \Gammabs_N(k,\ell).
	\label{def:Gamma}
\end{align}
The main result is the following.
\begin{theorem}
	\label{theorem:CLT_subspace}		
	Assume the separation conditions {\bf A-\ref{assumption:subspace1}} and {\bf A-\ref{assumption:subspace2}} hold. 
	Then we have 
	\begin{align}
		0 \leq \liminf_{N \to \infty} \min_{k,\ell} \vartheta_{N}(k,\ell) \leq \limsup_{N \to \infty} \max_{k,\ell} \vartheta_{N}(k,\ell) < \infty.
		\label{eq:bound_var_general}
	\end{align}	 
	Moreover, if $0 < \liminf_N \frac{K}{N}< \limsup_{N} \frac{K}{N} < c$, then 
	\begin{align}
		\liminf_{N \to \infty} \min_{k,\ell} \vartheta_{N}(k,\ell) > 0,
		\label{eq:lb_var_nonspike}
	\end{align}
	and if $K$ is independent of $N$, then
	\begin{align}
		&\vartheta_N(k,\ell)=
		\notag\\
	& \frac
		{
			\sigma^4 c_N \left(\lambda_{k,N} \lambda_{\ell,N} + (\lambda_{k,N}+\lambda_{\ell,N})\sigma^2 + \sigma^4\right)\left(\lambda_{k,N} \lambda_{\ell,N} + \sigma^4 c_N\right)
		}
		{2\left(\lambda_{k,N}^2 - \sigma^4 c_N\right)\left(\lambda_{\ell,N}^2 - \sigma^4 c_N\right)\left(\lambda_{k,N}\lambda_{\ell,N} - \sigma^4 c_N\right)}
		\left(1 - \mathbb{1}_{[K+1,M]}(k) \mathbb{1}_{[K+1,M]}(\ell)\right)
		+
		\epsilon_N(k,\ell),
		\label{eq:var_spike}
	\end{align}
	with $\max_{k,\ell} |\epsilon_N(k,\ell)| \to_N 0$.
	Finally, let $(\xi_N)$ be a deterministic bounded sequence and denote $\xibs_N=\left[\Re\left(\xi_N\right), \Im\left(\xi_N\right)\right]^T$. 
	Then,
	\begin{align}
		\Re\left(\xi_N\left(\hat{\eta}_N - \eta_N\right)\right) = \Ocal_{\Pbb}\left(\sqrt{\frac{\xibs_N^T \Gammabs_N \xibs_N}{N}}\right) + o_{\Pbb}\left(\frac{1}{\sqrt{N}}\right),
		\label{eq:tightness_real_part}
	\end{align}
	and if $\liminf_N \xibs_N^T \Gammabs_N \xibs_N > 0$, it holds that
	\begin{align}
		\sqrt{N} \frac{\Re\left(\xi_N\left(\hat{\eta}_N - \eta_N\right)\right)}{\sqrt{\xibs_N^T \Gammabs_N \xibs_N}} \xrightarrow[N \to \infty]{\Dcal} \Ncal_{\Rbb}\left(0, 1\right).	
		\label{eq:conv_distrib_eta}
	\end{align}
\end{theorem}
The proof of theorem \eqref{theorem:CLT_subspace} is defered to section \ref{section:proof_clt}.

		\subsubsection{Discussions and numerical examples}

In this section, we discuss the consequences of theorem \ref{theorem:CLT_subspace} and provide numerical examples illustrating the results.

We first remark that in the statement of theorem \eqref{theorem:CLT_subspace}, the purpose of the constraint $\liminf_N \xibs_N^T \Gammabs_N \xibs_N > 0$ is 
to ensure that the fluctuations of $\hat{\eta}_N - \eta_N$ are $\Ocal\left(N^{-1/2}\right)$. 
Indeed, there exist several situations where the fluctuations can be faster than $\Ocal\left(N^{-1/2}\right)$. For example, in the case where $K$ is independent of $N$ 
and $\d_{1,N}=\d_{2,N}=\u_{M,N}$, then we see from \eqref{eq:var_spike} that $\e_1^T \Gammabs_N \e_1=o(1)$ and thus 
\begin{align}
	\Re\left(\hat{\eta}_N - \eta_N\right) = o_{\Pbb}\left(\frac{1}{\sqrt{N}}\right).
	\notag
\end{align}
The result of theorem \eqref{theorem:CLT_subspace} can be rephrased in a more precise way, by considering the fluctuations of the random vector 
$\left[\Re\left(\hat{\eta}_N - \eta_N\right), \Im\left(\hat{\eta}_N - \eta_N\right)\right]^T$. However, in this case, we have to take into account the possible "degenerate" situations, 
when the covariance matrix $\Gammabs_N$ is asymptotically singular. Since it is difficult to state general results in this case, the next corollary focuses on one important special case where
$\d_{1,N}=\d_{2,N}$, i.e. the case of quadratic forms.
\begin{corollary}[Quadratic forms]
	\label{corollary:quad_form}
	Assume that $\d_{1,N}=\d_{2,N}=\d_N$, where $(\d_N)$ is a sequence of deterministic vectors such that $\limsup_N \|\d_N\| < \infty$. 
	Under the assumptions of theorem \ref{theorem:CLT_subspace},
	\begin{align}
		\frac{\sqrt{N} \left(\hat{\eta}_N - \eta_N\right)}
		{\sqrt{2\sum_{k=1}^M\sum_{\ell=1}^M \vartheta_{N}(k,\ell) \d_N^* \u_{k,N}\u_{k,N}^* \d_N \d_N^* \u_{\ell,N}\u_{\ell,N}^* \d_N}}
		\xrightarrow[N \to \infty]{\Dcal} \Ncal_{\Rbb}\left(0, 1\right).	
		\notag
	\end{align}
\end{corollary}
The result of corollary \ref{corollary:quad_form}  is illustrated in figure \ref{figure:quad_form_example} by comparing the empirical distribution of the quadratic form
$\sqrt{N}\left(\hat{\eta}_N-\eta_N\right)$ ($10^5$ trials), where $\d_{1,N}=\d_{2,N}=\e_M$, with the normal distribution $\Ncal_{\Rbb}\left(0,2 \vartheta_N(M,M)\right)$. 
The parameters are $M=20$, $N=40$, $\sigma=1$ and the matrix $\B_N\B_N^*$ is diagonal with non-zero eigenvalues at $5$ and $6$.
\begin{figure}[h]
	\centering
	\includegraphics[scale=0.4]{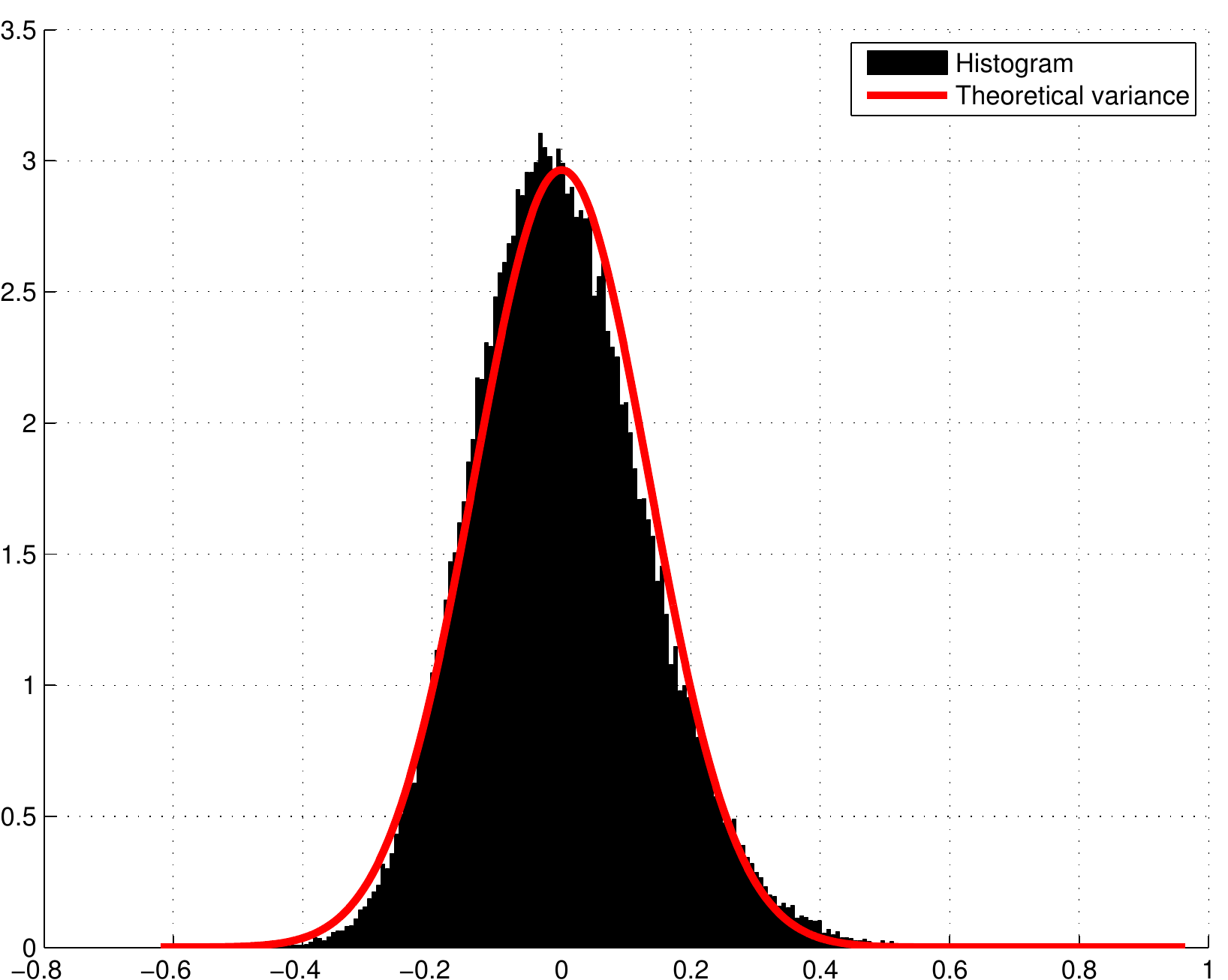}
	\caption
	{
		Empirical distribution of $\sqrt{N}\left(\hat{\eta}_N-\eta_N\right)$ (quadratic form)
	}
	\label{figure:quad_form_example}
\end{figure}

In the "non-degenerate" case, we have to ensure that $\Gammabs_N$ is asymptotically non-singular. By computing the smallest eigenvalue of $\Gammabs_N$, this is equivalent to
\begin{align}
	\liminf_{N \to \infty} 
	\left(
		\sum_{k,\ell} \left(\vartheta_N(k,\ell) \eta^{(1,1)}_{k,N}\eta^{(2,2)}_{\ell,N}\right) - \left|\sum_{k,\ell} \vartheta_N(k,\ell) \eta^{(1,2)}_{k,N}\eta^{(1,2)}_{\ell,N}\right|
	\right)
	> 0,
	\label{assumption:non_singular_cov}
\end{align}
We therefore have the following result, by using the fact that 
\begin{align}
	\Re\left(\xi_N (\hat{\eta}_N - \eta_N)\right) = 
	\begin{bmatrix}
		\Re(\xi_N) & \Im(\xi_N)
	\end{bmatrix}
	\begin{bmatrix}
		\Re\left(\hat{\eta}_N - \eta_N\right) \\ -\Im\left(\hat{\eta}_N - \eta_N\right)
	\end{bmatrix}	
\end{align}
\begin{corollary}[Non-degenerate case]
	\label{corollary:non_degenerate}
	Under the assumptions of theorem \ref{theorem:CLT_subspace} and if \eqref{assumption:non_singular_cov} holds, then
	\begin{align}
		\sqrt{N}\ 
		\Gammabs_N^{-1/2}
		\begin{bmatrix}
			\Re\left(\hat{\eta}_N - \eta_N\right)
			\\
			-\Im\left(\hat{\eta}_N - \eta_N\right)
		\end{bmatrix}
		\xrightarrow[N \to \infty]{\Dcal} \Ncal_{\Rbb^2}\left(\mathbf{0}, \I\right).	
		\notag
	\end{align}
\end{corollary}
An illustration of corollary \eqref{corollary:non_degenerate}  is given in figure \ref{figure:bil_form_example} where we have compared the empirical distribution of the bilinear form
$\sqrt{N}\Re\left(\hat{\eta}_N-\eta_N\right)$ ($10^5$ trials), where $\d_{1,N}=\e_M$ and $\d_{2,N}=\e_{M-1}$, 
with $\Ncal_{\Rbb}\left(0,\e_1^T \Gammabs_N \e_1\right)$.
The parameters are $M=20$, $N=40$, $\sigma=1$ and the matrix $\B_N\B_N^*$ is diagonal with non-zero eigenvalues at $5$ and $6$ ($K=2$).
\begin{figure}[h]
	\centering
	\includegraphics[scale=0.4]{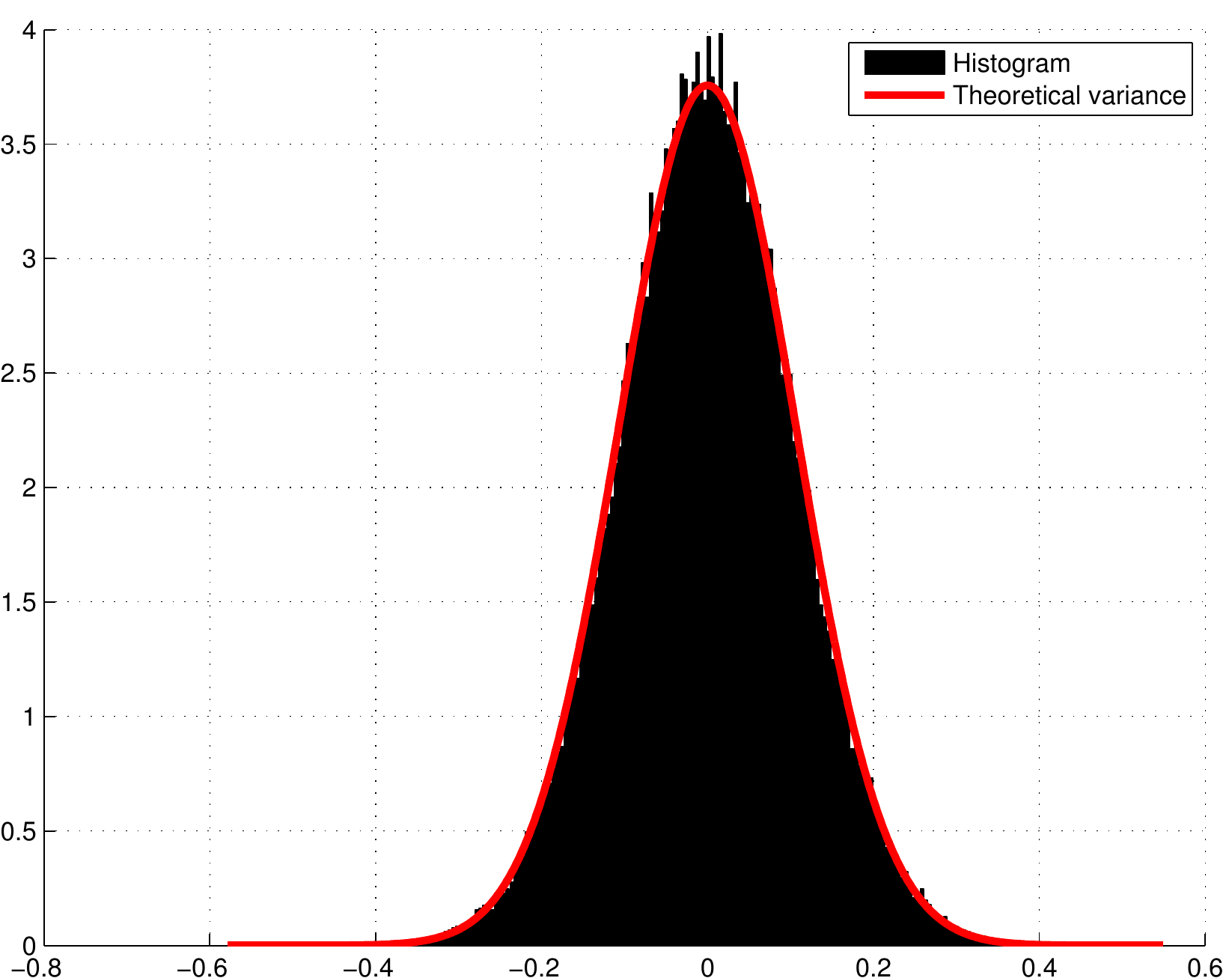}
	\caption
	{
		Empirical distribution of $\sqrt{N}\Re\left(\hat{\eta}_N-\eta_N\right)$ (bilinear form).
	}
	\label{figure:bil_form_example}
\end{figure}

\subsubsection{CLT for the traditional noise subspace estimate}

To conclude section \ref{section:CLT_bilinear_form}, we provide a CLT for the traditional noise subspace estimate, defined in \eqref{eq:traditional_localization_function_estimator}
by
\begin{align}
	\hat{\eta}_N^{(t)} = 
	\d_{1,N}^* \left(\I - \sum_{k=1}^{K} \hat{\u}_{k,N} \hat{\u}_{k,N}^*\right) \d_{2,N}.
	\notag
\end{align}	
From \eqref{eq:separation_lambdahat}, almost surely for $N$ large enough, the $K$ largest
eigenvalues $\hat{\lambda}_{1,N},\ldots,\hat{\lambda}_{K,N}$ are located inside the rectangle $\Rcal$ defined in \eqref{def:rectangle}, while the smallest $M-K$ remain
outside $\Rcal$. Therefore, for $N$ large enough, almost surely,
\begin{align}
	\hat{\eta}_N^{(t)}
	= 
	\d_{1,N}^*
	\left(
		\I - \frac{1}{2 \pi \irm}\oint_{\partial \Rcal}  \Q_N(z) \drm z
	\right) 
	\d_{2,N}.
	\notag
\end{align}
As for \eqref{eq:uniform_conv_integrand}, we have
\begin{align}
	\sup_{z \in \partial \Rcal} 
	\left| \d_{1,N}^* \left(\Q_N(z) - \T_N(z) \right) \d_{2,N} \right|
	\xrightarrow[N \to \infty]{a.s.} 0.
	\notag
\end{align}
which immediately implies that 
\begin{align}
	\hat{\eta}_N^{(t)} - \eta_N^{(t)} \xrightarrow[N\to\infty]{a.s.} 0,
	\notag	
\end{align}
where 
\begin{align}
	\eta_N^{(t)} = 
	\frac{1}{2 \pi \irm} \oint_{\partial\Rcal} \d_{1,N}^* \T_N(z) \d_{2,N} \drm z.
	\notag
\end{align}
Define, as for \eqref{eq:def_vartheta},
\begin{align}
	&\vartheta_N^{(t)}(k,\ell) = 
	\notag\\
	&
	\left(\frac{1}{2 \pi \irm}\right)^2 \oint_{\partial \Rcal}\oint_{\partial \Rcal}
	\frac
	{
		\frac{\sigma^2}{2} \theta_{N}^{(k,\ell)}(z_1,z_2) 
		\left(1+\sigma^2 c_N m_N(z_1)\right)\left(1+\sigma^2 c_N m_N(z_2)\right)}
	{
		\left(\lambda_{k,N} - w_N(z_1)\right)
		\left(\lambda_{\ell,N} - w_N(z_1)\right)
		\left(\lambda_{k,N} - w_N(z_2)\right)
		\left(\lambda_{\ell,N} - w_N(z_2)\right)
		\Delta_N(z_1,z_2)
	}
	 \drm z_1 \drm z_2,
	 \label{eq:def_vartheta_trad}
\end{align}
with $\Delta_N(z_1,z_2)$ and $\theta_{N}^{(k,\ell)}(z_1,z_2)$ defined respectively by \eqref{def:determinant} and \eqref{eq:def_theta}.
We define also
\begin{align}
	\Gammabs_N^{(t)} = 
	\sum_{k=1}^M \sum_{\ell=1}^M \vartheta_N^{(t)}(k,\ell) \Gammabs_N(k,\ell),
	\label{def:Gamma_trad}
\end{align}	
as for \eqref{def:Gamma}. Then we have the following result :
\begin{theorem}
	\label{theorem:CLT_trad_subspace}		
	Assume the separation conditions {\bf A-\ref{assumption:subspace1}} 
	and {\bf A-\ref{assumption:subspace2}} hold, and that $K$ is fixed with respect to $N$. 
	Then $\max_{k,\ell \geq K+1} \vartheta_N^{(t)}(k,\ell) \to_N 0$,
	\begin{align}
		\vartheta_N^{(t)}(k,\ell)= 
		\frac
		{
			\sigma^2\left(\lambda_{k,N} + \sigma^2\right)
			\left(\lambda_{k,N}^2 - \sigma^4 c_N\right)
		}
		{2 \lambda_{k,N}^2 \left(\lambda_{k,N} + \sigma^2 c_N\right)^2}
		+ \epsilon_{N}(k,\ell)	
	\notag
\end{align}
for $k \leq K, \ell \geq K+1$, 
with $\max_{k \leq K, \ell \geq K+1} |\epsilon_{N}(k,\ell)| \to_N 0$, and 
	\begin{align}
		\vartheta_N^{(t)}(k,\ell)=
		\frac
		{
			\sigma^4 c_N \chi_N^{(t)}(k,\ell)
		}
		{
			2 \lambda_{k,N}\lambda_{\ell,N}
			\left(\lambda_{k,N} + \sigma^2 c_N\right)^2
			\left(\lambda_{\ell,N} + \sigma^2 c_N\right)^2
			\left(\lambda_{k,N} \lambda_{\ell,N} - \sigma^4 c_N\right)
		}
		+ o(1)
		\notag
	\end{align}
	for $1 \leq k,\ell \leq K$, where $\chi_N^{(t)}(k,\ell)$ is defined by
	\begin{align}
		&\chi_{N}^{(t)}(k,\ell)
		=
		\notag\\ 
	&	\lambda_{k,N} \lambda_{\ell,N} 
		\left(
			\lambda_{k,N} \lambda_{\ell,N} + \sigma^2 (\lambda_{k,N} + \lambda_{\ell,N}) 
			+  \sigma^4
		\right)
		\left(
			(1+c_N) (\lambda_{k,N}\lambda_{\ell,N}+\sigma^4 c_N) 
			+ 2 \sigma^2 c_N (\lambda_{k,N} + \lambda_{\ell,N})
		\right)
		\notag\\
		&
		- c 
		\left(
			\lambda_{k,N}\lambda_{\ell,N}
			-\sigma^4 c_N
		\right)
		\left(
			\lambda_{k,N}\lambda_{\ell,N} 
			+ \sigma^2 (\lambda_{k,N}+\lambda_{\ell,N}) 
			+ \sigma^4 c_N
		\right)^2.
	\notag
\end{align}
	Finally, let $(\xi_N)$ be a deterministic bounded sequence and denote $\xibs_N=\left[\Re\left(\xi_N\right), \Im\left(\xi_N\right)\right]^T$. 
	Then, if $\liminf_N \xibs_N^T \Gammabs_N^{(t)} \xibs_N > 0$, it holds that
	\begin{align}
		\sqrt{N} 
		\frac
		{
			\Re\left(\xi_N\left(\hat{\eta}_N^{(t)} - \eta_N^{(t)}\right)\right)
		}
		{\sqrt{\xibs_N^T \Gammabs_N^{(t)} \xibs_N}} 
		\xrightarrow[N \to \infty]{\Dcal} \Ncal_{\Rbb}\left(0, 1\right).	
		\notag
	\end{align}
\end{theorem}
The proof of Theorem \ref{theorem:CLT_trad_subspace}, which follows step by step the proof
of Theorem \ref{theorem:CLT_subspace}, is omitted.

\section{Proof of theorem \ref{theorem:CLT_subspace}}
\label{section:proof_clt}

This section is dedicated to prove theorem \ref{theorem:CLT_subspace}. Several long computations will be defered to the appendix.

	\subsection{Regularization and confinement of the eigenvalues}
	\label{section:regularization}

To prove theorem \ref{theorem:CLT_subspace}, we will use the usual Levy's theorem and prove the convergence of the characteristic function.
Since the moments of $\hat{\eta}_N$ may not be defined, due to the poles in the integrand of \eqref{eq:subspace_estimator} (see remark \ref{remark:subspace_estimator}), 
we first use a trick from \cite{hachem2012large}, to force these poles to be away from the contour, and which does not modify the asymptotic distribution of $\hat{\eta}_N$.

Let $\varphi \in \Ccal_c^\infty(\Rbb,[0,1])$ s.t.
\begin{align}
	\varphi(\lambda) = 
	\begin{cases}
		1 & \text{for } \lambda \in [t_1^- - \frac{\epsilon}{3}, t_1^+ + \frac{\epsilon}{3}] \cup [t_2^- - \frac{\epsilon}{3}, t_2^+ + \frac{\epsilon}{3}]
		\\
		0 & \text{for } \lambda \in \Rbb \backslash \left([t_1^- - \frac{2 \epsilon}{3}, t_1^+ + \frac{2 \epsilon}{3}] \cup [t_2^- - \frac{2\epsilon}{3}, t_2^+ + \frac{2\epsilon}{3}]\right),
	\end{cases}
	\label{def:reg_phi}
\end{align}
where $\epsilon$ is given in \eqref{def:rectangle}, and define the regularization coefficient 
\begin{align}
	\chi_N = \det\ \varphi\left(\Sigmabs_N\Sigmabs_N^*\right) \det\ \varphi\left(\hat{\Omegabs}_N\right)
	\label{def:reg_chi}
\end{align}
(see remark \ref{remark:coeff_xi} for the definition of $\hat{\Omegabs}_N$).
From \eqref{eq:separation_lambdahat} and \eqref{eq:separation_omegahat}, we have $\chi_N = 1$ w.p.1 for $N$ large and thus, for all $p \in \Nbb$, we get
\begin{align}
	\hat{\eta}_N \chi_N^2 = \hat{\eta}_N + \Ocal_{\Pbb}\left(\frac{1}{N^p}\right).
	\notag
\end{align}
Therefore, to obtain a CLT for $\hat{\eta}_N$, we only need to prove a CLT for $\hat{\eta}_N \chi_N^2$.
Moreover, it is proved in \cite{hachem2012large} that
\begin{align}
	\sup_{z \in \partial \Rcal} \Ebb\left[\left|\d_{1,N}^* \left(\Q_N(z) - \T_N(z)\right) \d_{2,N}\right|^2 \chi_N\right] &= \Ocal\left(\frac{1}{N}\right),
	\label{eq:secordquad} \\
	\sup_{z \in \partial \Rcal} 
	\Ebb\left[\left|\frac{\hat{w}'_N(z)}{1+\sigma^2 c_N \hat{m}_N(z)} - \frac{w'_N(z)}{1+\sigma^2 c_N m_N(z)}\right|^2 \chi_N\right] 
	&= \Ocal\left(\frac{1}{N^2}\right).
	\label{eq:secordrem}
\end{align}
Since $\frac{\hat{w}'_N(z)}{1+\sigma^2 c_N \hat{m}_N(z)} \chi_N$ fluctuates less than the quadratic form $\d_{1,N}^* \Q_N(z)\d_{2,N} \chi_N$, we can replace it with 
$\frac{w'_N(z)}{1+\sigma^2 c_N m_N(z)}$ without modifying any asymptotic second order results. 
Indeed, it is easy to see from \eqref{eq:secordrem} that
\begin{align}
	\hat{\eta}_N \chi_N^2 
	= \d_{1,N}^* \left(\I - \frac{1}{2 \pi \irm} \oint_{\partial \Rcal} \Q_N(z) \chi_N \frac{w'_N(z)}{1+\sigma^2 c_N m_N(z)} \drm z\right) \d_{2,N}
	+ \Ocal_{\Pbb}\left(\frac{1}{N}\right),
	\notag
\end{align}
and the problem reduces finally to obtain the asymptotic distribution of 
\begin{align}
	\hat{\gamma}_N = \frac{1}{2 \pi \irm} \oint_{\partial \Rcal} \d_{1,N}^* \Q_N(z) \d_{2,N} \chi_N \frac{w'_N(z)}{1+\sigma^2 c_N m_N(z)} \drm z.
	\label{def:gammahat}
\end{align}
\begin{remark}
	If $\sqrt{N}\left(c_N - c\right) \to 0$, it can be proved that 
	\begin{align}
		\sup_{z \in \Rcal} \left|\frac{w'_N(z)}{1+\sigma^2 c_N m_N(z)} - \frac{w'(z)}{1+\sigma^2 c m(z)}\right| = o \left(\frac{1}{\sqrt{N}}\right),
		\notag
	\end{align}
	and the simplified estimator \eqref{eq:simplified_eta} $\hat{\eta}_N^{(s)}$ derived in \cite{hachem2012subspace} will have the same asymptotic
	fluctuations as $\hat{\eta}_N$.
\end{remark}
In the remainder, we denote by $\psi_N(u)$ the characteristic function defined on $\Rbb$ by
\begin{align}
	\psi_N(u) = \Ebb\left[\mathrm{exp}\left(\irm u \sqrt{N} \Re\left(\xi_N \hat{\gamma}_N \right)\right)\right],
	\notag
\end{align}
where $(\xi_N)$ is a deterministic sequence such that $\limsup_N \left|\xi_N\right| < \infty$.

Finally, we recall two useful properties from \cite[Prop. 3.3]{hachem2012large}:
\begin{align}
	\Ebb\left[\hat{\gamma}_N - \left(\d_{1,N}^*\d_{2,N}-\eta_N\right)\right] = \Ocal\left(\frac{1}{N^{3/2}}\right)
	\text{ and }
	\Ebb\left|\hat{\gamma}_N - \left(\d_{1,N}^*\d_{2,N}-\eta_N\right)\right|^2 = \Ocal\left(\frac{1}{N}\right).
	\label{eq:bias_and_variance}
\end{align}

	\subsection{The differential equation}
	\label{section:diff_equ}
	
We first prove that the characteristic function $\psi_N(u)$ satisfies the differential equation	of a Gaussian characteristic function, up to an error term.
\begin{remark}
	\label{remark:coeff_xi}
	Note that in the expression of $\psi_N(u)$, we can assume for ease of reading and without loss of generality that $\xi_N=1$ (by considering vectors $\d_{1,N} \overline{\xi_N}^{1/2}$ 
	and $\d_{2,N} \xi_N^{1/2}$ in the bilinear form).
\end{remark}
In the following, $\epsilon_N(u,z_1,z_2)$ will denote a complex generic continuous function defined on $\Rbb \times \partial \Rcal \times \partial \Rcal$, 
such that $u \mapsto \epsilon_N(u,z_1,z_2)$ is continuously differentiable, and
\begin{align}
	\limsup_{N\to\infty} \sup_{(z_1,z_2) \in \partial \Rcal \times \partial \Rcal} 
	\left\{\left|\epsilon_N(u,z_1,z_2)\right|, \left|\frac{\partial \epsilon_N(u,z_1,z_2)}{\partial u}\right|\right\} < \Prm(u),
	\notag
\end{align}
with $\Prm(u)$ a polynomial with positive coefficients. $\epsilon_N(u,z_1,z_2)$ may take different values from one line to another.
We will also keep the notation $\epsilon_N(u)$, $\epsilon_N(z_1,z_2)$, $\epsilon_N(u,z_1)$ if $\epsilon_N(u,z_1,z_2)$ does not depend on $(z_1,z_2)$, $u$ or $z_2$.

Using dominated convergence and Fubini's theorem, the derivative $\psi_N'(u)$ writes
\begin{align}
	\psi_N'(u) = 
	\frac{\irm \sqrt{N}}{2} 
	\frac{1}{2 \pi \irm} \oint_{\partial \Rcal} 
	\Ebb\left[\left(\d_{1,N}^* \Q_N(z) \d_{2,N} + \d_{2,N}^* \Q_N(z) \d_{1,N}\right)\chi_N \erm^{\irm u \sqrt{N} \Re\left(\hat{\gamma}_N\right)}\right]
	\frac{w'_N(z)}{1+\sigma^2 c_N m_N(z)} \drm z,
	\label{eq:derivative_psi}
\end{align}
so that we need to develop the term $\Ebb\left[\d_{1,N}^* \Q_N(z) \d_{2,N} \chi_N \erm^{\irm u \Re\left(\hat{\gamma}_N\right)}\right]$.
By standard computations defered to appendix \ref{appendix:gaussiancomp1}, we obtain
\begin{align}
	&\psi_N'(u) = 
	\notag\\
	&\left(
	\irm \sqrt{N} \Re\left(\d_{1,N}^*\d_{2,N} -\eta_N\right)
	- \frac{u \sigma^2}{4}
	\left(\frac{1}{2 \pi \irm}\right)^2 \oint_{\partial \Rcal} \oint_{\partial \Rcal} 
	\frac{\left(\mu_{N}(z_1,z_2)+\tilde{\mu}_N(z_1,z_2)\right) w'_N(z_1) w'_N(z_2)}
	{\left(1+\sigma^2 c_N m_N(z_1)\right)^2\left(1+\sigma^2 c_N m_N(z_2)\right)^2} \drm z_1 \drm z_2 
	\right)
	\psi_N(u)
	\notag\\
	&
	+ \frac{\epsilon_N(u)}{\sqrt{N}},
	\label{eq:charac_fun_1}
\end{align}
where the quantity $\mu_{N}(z_1,z_2)$ is given by
\begin{align}
\label{eq:def_mu}
\begin{split}
	&\mu_{N}(z_1,z_2) = 
	\\
	&
	\sum_{\substack{i,j=1 \\ i \neq j}}^2
	\Biggl[
		\d_{i,N}^* \T_N(z_1) \T_N(z_2) \d_{j,N} \Ebb\left[\d_{i,N}^* \Q_N(z_1) \B_N\B_N^*\Q_N(z_2)\d_{j,N} \chi_N \right] 
		\\
		&
		+
		\d_{i,N}^* \T_N(z_1) \B_N\B_N^* \T_N(z_2) \d_{j,N} \Ebb\left[\d_{i,N}^* \Q_N(z_1) \Q_N(z_2)\d_{j,N} \chi_N \right] 
		\\
		&
		+
		s_N(z_1,z_2) \left(1+\sigma^2 c_N m_N(z_1)\right)\left(1+\sigma^2 c_N m_N(z_2)\right)  
		\d_{i,N}^* \T_N(z_1) \T_N(z_2) \d_{j,N} \Ebb\left[\d_{i,N}^* \Q_N(z_1) \Q_N(z_2)\d_{j,N} \chi_N \right]
	\Biggr],
\end{split}
\end{align}
and $\tilde{\mu}_N(z_1,z_2)$ is given by
\begin{align}
\label{eq:def_mutilde}
\begin{split}
	&\tilde{\mu}_{N}(z_1,z_2) = 
	\\
	&
	\sum_{\substack{i,j=1 \\ i \neq j}}^2
	\Biggl[
		\d_{i,N}^* \T_N(z_1) \T_N(z_2) \d_{i,N} \Ebb\left[\d_{j,N}^* \Q_N(z_1) \B_N\B_N^*\Q_N(z_2)\d_{j,N} \chi_N \right] 
		\\
		&
		+
		\d_{i,N}^* \T_N(z_1) \B_N\B_N^* \T_N(z_2) \d_{i,N} \Ebb\left[\d_{j,N}^* \Q_N(z_1) \Q_N(z_2)\d_{j,N} \chi_N \right] 
		\\
		&
		+
		s_N(z_1,z_2) \left(1+\sigma^2 c_N m_N(z_1)\right)\left(1+\sigma^2 c_N m_N(z_2)\right)  
		\d_{i,N}^* \T_N(z_1) \T_N(z_2) \d_{i,N} \Ebb\left[\d_{j,N}^* \Q_N(z_1) \Q_N(z_2)\d_{j,N} \chi_N \right]
	\Biggr],
\end{split}
\end{align}
where $s_N(z_1,z_2)$ defined as
\begin{align}
	s_N(z_1,z_2) =  \frac{\sigma^2}{N} \Tr 
	\frac{\I - \frac{\B_N^*\T_N(z_1)\B_N}{1+\sigma^2 c_N m_N(z_1)} - \frac{\B_N^*\T_N(z_2)\B_N}{1+\sigma^2 c_N m_N(z_2)}}
	{\left(1+\sigma^2 c_N m_N(z_1)\right)\left(1+\sigma^2 c_N m_N(z_2)\right)}.
	\label{def:s}
\end{align}
It now remains to approximate $\Ebb\left[\d_{i,N}^* \Q_N(z_1) \B_N\B_N^*\Q_N(z_2)\d_{j,N} \chi_N \right]$ and $\Ebb\left[\d_{i,N}^* \Q_N(z_1) \Q_N(z_2)\d_{j,N} \chi_N \right]$,
and we introduce for that purpose the following quantity
\begin{align}
	r_N(z_1,z_2) = \frac{\sigma^2}{N} \Tr \frac{\T_N(z_1) \B_N\B_N^* \T_N(z_2)\B_N\B_N^*}{\left((1+\sigma^2 c_N m_N(z_1)\right)^2\left((1+\sigma^2 c_N m_N(z_2)\right)^2}.
	\label{def:r}
\end{align}
Recall moreover the definitions of $u_{N}(z_1,z_2)$, $v_N(z_1,z_2)$ and $\tilde{v}_N(z_1,z_2)$ given by \eqref{def:u} and \eqref{def:v_vtilde}.
\begin{proposition}
	\label{proposition:bil_form}
	For all $z_1,z_2 \in \partial\Rcal$, we have $s_N(z_1,z_2) + r_N(z_1,z_2) = z_1 z_2 \tilde{v}_N(z_1,z_2)$,
	\begin{align}
			&\Ebb\left[\d_{1,N}^* \Q_N(z_1) \Q_N(z_2) \d_{2,N} \chi_N\right] = 
			\notag\\
			&
			\frac{1-u_N(z_1,z_2)}{\Delta_N(z_1,z_2)}  \d_{1,N}^* \T_N(z_1) \T_N(z_2) \d_{2,N} 
			+ \frac{v_N(z_1,z_2)}{\Delta_N(z_1,z_2)} \frac{\d_{1,N}^* \T_N(z_1) \B_N\B_N^*\T_N(z_2) \d_{2,N}}{\left(1+\sigma^2 c_N m_N(z_1)\right)\left(1+\sigma^2 c_N m_N(z_2)\right)}
			+ \frac{\epsilon_N(z_1,z_2)}{N},
			\label{eq:approx_doubleQ1}
	\end{align}
	and
	\begin{align}
			&\Ebb\left[\d_{1,N}^* \Q_N(z_1) \B_N\B_N^* \Q_N(z_2) \d_{2,N} \chi_N\right] = 
			\notag\\
			&\frac{\left(1+\sigma^2 c_N m_N(z_1)\right)\left(1+\sigma^2 c_N m_N(z_1)\right)\left(u_N(z_1,z_2)s_N(z_1,z_2)+r_N(z_1,z_2)\right)}{\Delta_N(z_1,z_2)}
			 \d_{1,N}^* \T_N(z_1)\T_N(z_2)\d_{2,N}
			 \notag\\
			&+\frac{1-s_N(z_1,z_2)v_N(z_1,z_2)-u_N(z_1,z_2)}{\Delta_N(z_1,z_2)}
			\d_{1,N}^* \T_N(z_1) \B_N\B_N^*\T_N(z_2) \d_{2,N}
			+ \frac{\epsilon_N(z_1,z_2)}{N}.
			\label{eq:approx_doubleQ2}
	\end{align}
\end{proposition}
Note that the inverse of $\Delta_N(z_1,z_2)$ is well defined thanks to \eqref{eq:bound_det} in lemma \ref{lemma:properties_determinant}.
The proof of proposition \ref{proposition:bil_form} is given in appendix \ref{appendix:bil_form}.
Using the expression of $\Ebb\left[\d_{1,N}^*\Q_N(z_1)\B_N\B_N^*\Q_N(z_2)\d_{2,N}\chi_N\right]$ and $\Ebb\left[\d_{1,N}^*\Q_N(z_1)\Q_N(z_2)\d_{2,N}\chi_N\right]$ in 
proposition \ref{proposition:bil_form} and the fact that $s_N(z_1,z_2) + r_N(z_1,z_2) = z_1 z_2 \tilde{v}_N(z_1,z_2)$, we further obtain
\begin{align}
\begin{split}	
	&\mu_N(z_1,z_2)=
	\\
	&
	\sum_{\substack{i,j=1 \\ i \neq j}}^2
	\Biggl(
		\frac{z_1 z_2 \left(1+\sigma^2 c_N m_N(z_1)\right)\left(1+\sigma^2 c_N m_N(z_2)\right) \tilde{v}_N(z_1,z_2)}{\Delta_N(z_1,z_2)}
		\left(\d_{i,N}^* \T_N(z_1)\T_N(z_2)\d_{j,N}\right)^2
		\\
		&
		+\frac{v_N(z_1,z_2) \left(\d_{i,N}^* \T_N(z_1)\B_N\B_N^*\T_N(z_2)\d_{j,N}\right)^2}
		{\left(1+\sigma^2 c_N m_N(z_1)\right)\left(1+\sigma^2 c_N m_N(z_2)\right)\Delta_N(z_1,z_2)}
		\\
		&
		+\frac{2\left(1-u_N(z_1,z_2)\right)}{\Delta_N(z_1,z_2)}\d_{i,N}^* \T_N(z_1)\T_N(z_2)\d_{j,N}\d_{i,N}^* \T_N(z_1)\B_N\B_N^*\T_N(z_2)\d_{j,N}
	\Biggr) +\frac{\epsilon_N(z_1,z_2)}{N},
\end{split}
\notag
\end{align}
and
\begin{equation}
\begin{split}	
	&\tilde{\mu}_N(z_1,z_2)=
	\\
	&
	\sum_{\substack{i,j=1 \\ i \neq j}}^2
	\Biggl(
		\frac{z_1 z_2 \left(1+\sigma^2 c_N m_N(z_1)\right)\left(1+\sigma^2 c_N m_N(z_2)\right)\tilde{v}_N(z_1,z_2)}{\Delta_N(z_1,z_2)}
		\d_{i,N}^* \T_N(z_1)\T_N(z_2)\d_{i,N}\d_{j,N}^* \T_N(z_1)\T_N(z_2)\d_{j,N}
		\\
		&
		+\frac{v_N(z_1,z_2) \d_{i,N}^* \T_N(z_1)\B_N\B_N^*\T_N(z_2)\d_{i,N} \d_{j,N}^* \T_N(z_1)\B_N\B_N^*\T_N(z_2)\d_{j,N}}
		{\left(1+\sigma^2 c_N m_N(z_1)\right)\left(1+\sigma^2 c_N m_N(z_2)\right)\Delta_N(z_1,z_2)}
		\\
		&
		+\frac{2\left(1-u_N(z_1,z_2)\right)}{\Delta_N(z_1,z_2)}\d_{i,N}^* \T_N(z_1)\T_N(z_2)\d_{i,N}\d_{j,N}^* \T_N(z_1)\B_N\B_N^*\T_N(z_2)\d_{j,N}
	\Biggr) + \frac{\epsilon_N(z_1,z_2)}{N},	
\end{split}
\notag
\end{equation}
Going back to \eqref{eq:derivative_psi} and introducing again the deterministic sequence $(\xi_N)$ (see remark \ref{remark:coeff_xi}), we finally obtain
\begin{align}
	&\psi'_N(u) =  
	\notag\\
	&
	\left(
		\irm \sqrt{N} \Re\left(\xi_N (\d_{1,N}^*\d_{2,N} - \eta_N)\right) 
		- u \sum_{k=1}^M \sum_{\ell=1}^M \vartheta_N(k,\ell) 
		\left(
			\Re\left(\xi_N^2 \eta_{k,N}^{(1,2)} \eta_{\ell,N}^{(1,2)}\right) 
			+ 
			\frac{|\xi_N|^2}{2}\left(\eta_{k,N}^{(1,1)}\eta_{\ell,N}^{(2,2)} + \eta_{k,N}^{(2,2)}\eta_{\ell,N}^{(1,1)}\right)
		\right)
	\right) \psi_N(u) 
	\notag\\
	&+ \frac{\epsilon_N(u)}{\sqrt{N}},
	\label{eq:diff_eq}
\end{align}
where we recall that $\eta_{k,N}^{(i,j)} = \d_{i,N}^* \u_{k,N}\u_{k,N}^* \d_{j,N}$ and where $\vartheta_N(k,\ell)$ is given by
\begin{align}
	&\vartheta_N(k,\ell) =
	\notag\\
	& 
	\frac{\sigma^2}{2}\left(\frac{1}{2 \pi \irm}\right)^2 \oint_{\partial \Rcal}\oint_{\partial \Rcal}
	\frac{\theta_{N}^{(k,\ell)}(z_1,z_2) w'_N(z_1)w'_N(z_2)}
	{
		\left(\lambda_{k,N} - w_N(z_1)\right)\left(\lambda_{\ell,N} - w_N(z_1)\right)\left(\lambda_{k,N} - w_N(z_2)\right)\left(\lambda_{\ell,N} - w_N(z_2)\right)
		\Delta_N(z_1,z_2)
	}
	 \drm z_1 \drm z_2,
	 %\label{eq:def_vartheta}
\end{align}
with
\begin{align}
	&\theta_{N}^{(k,\ell)}(z_1,z_2) = 
	\notag\\
	&
	z_1 z_2  \left(1+\sigma^2 c_N m_N(z_1)\right)\left(1+\sigma^2 c_N m_N(z_2)\right) \tilde{v}_N(z_1,z_2)
	\notag\\
	&
	+\frac{\lambda_{k,N}\lambda_{\ell,N} v_N(z_1,z_2)}{\left(1+\sigma^2 c_N m_N(z_1)\right)\left(1+\sigma^2 c_N m_N(z_2)\right)}
	+(\lambda_{k,N}+\lambda_{\ell,N}) \left(1-u_N(z_1,z_2)\right).
	%\label{eq:def_theta}
\end{align}
We can check easily that $\vartheta_{N}(k,\ell) \in \Rbb$. 
By letting $\xibs_N=[\Re(\xi_N),\Im(\xi_N)]^T$, we obtain
\begin{align}
	\psi'_N(u) =  
	\left(
		\irm \sqrt{N} \Re\left(\xi_N (\d_{1,N}^*\d_{2,N}-\eta_N)\right) 
		- u \xibs_N^T \Gammabs_N \xibs_N
	\right) \psi_N(u) 
	+ \frac{\epsilon_N(u)}{\sqrt{N}},
	\label{eq:diff_eq2}
\end{align}
where 
\begin{align}
	\Gammabs_N = \sum_{k=1}^M \sum_{\ell=1}^M \vartheta_N(k,\ell) \Gammabs_N(k,\ell).
	\notag
\end{align}	
with the $2 \times 2 $ matrix $\Gammabs_N(k,\ell)$ given by
\begin{align}
	&\Gammabs_N(k,\ell)=
	\notag\\
	&	
	\begin{bmatrix}
		\Re\left(\eta_{k,N}^{(1,2)} \eta_{\ell,N}^{(1,2)}\right) 
		+ \frac{1}{2}\left(\eta_{k,N}^{(1,1)}\eta_{\ell,N}^{(2,2)} + \eta_{k,N}^{(2,2)}\eta_{\ell,N}^{(1,1)}\right)
		&
		-\Im\left(\eta_{k,N}^{(1,2)} \eta_{\ell,N}^{(1,2)}\right).
		\\
		-\Im\left(\eta_{k,N}^{(1,2)} \eta_{\ell,N}^{(1,2)}\right).
		&
		-\Re\left(\eta_{k,N}^{(1,2)} \eta_{\ell,N}^{(1,2)}\right) + \frac{1}{2}\left(\eta_{k,N}^{(1,1)}\eta_{\ell,N}^{(2,2)} 
		+ \eta_{k,N}^{(2,2)}\eta_{\ell,N}^{(1,1)}\right)
	\end{bmatrix}.
	\label{eq:def_gamma_kl}
\end{align}
From the trivial inequality $\left|\Re(z_1 z_2)\right| \leq \frac{1}{2}\left(|z_1|^2 + |z_2|^2\right)$ for $z_1,z_2\in \Cbb$, we have
\begin{align}
	\left|\Re\left(\xi_N^2\eta_{k,N}^{(1,2)} \eta_{\ell,N}^{(1,2)}\right)\right| 
	\leq
	\frac{|\xi_N|^2}{2}\left(\eta_{k,N}^{(1,1)}\eta_{\ell,N}^{(2,2)} + \eta_{k,N}^{(2,2)}\eta_{\ell,N}^{(1,1)}\right),
	\notag
\end{align}
which implies
\begin{align}
	\xibs_N^T \Gammabs_N(k,\ell) \xibs_N
	=
	\Re\left(\xi_N^2 \eta_{k,N}^{(1,2)} \eta_{\ell,N}^{(1,2)}\right) + \frac{|\xi_N|^2}{2}\left(\eta_{k,N}^{(1,1)}\eta_{\ell,N}^{(2,2)} + \eta_{k,N}^{(2,2)}\eta_{\ell,N}^{(1,1)}\right)
	\geq 0,
\end{align}
in other words that $\Gammabs_N(k,\ell)$ is non-negative definite.

		\subsection{Asymptotics of $\vartheta_{N}(k,\ell)$}
		\label{section:asymptotic_variance}

The purpose of this section is to prove \eqref{eq:bound_var_general}, \eqref{eq:lb_var_nonspike} and \eqref{eq:var_spike} for the coefficients 
$\vartheta_N(k,\ell)$.
	
Using the bounds \eqref{eq:bound_norm_T_Ttilde}, \eqref{eq:bound_m_half} and the fact that $\left|m'_N(z)\right| \leq \drm(z,\supp(\mu_N))^{-2}$, 
it is easily shown that $\theta^{(k,\ell)}_N(z_1,z_2)$, defined in \eqref{eq:def_theta}, satisfies
\begin{align}
	\limsup_{N \to \infty} \max_{k,\ell} \sup_{z_1,z_2 \in \partial \Rcal} \left|\theta^{(k,\ell)}_N(z_1,z_2) w'_N(z_1) w'_N(z_2)\right| < \infty.
	\notag
\end{align}
Moreover, from \eqref{eq:bound_dist_w_lambda} and \eqref{eq:bound_det}, we also have
\begin{align}
	\limsup_{N \to \infty} \max_{k,\ell} \sup_{z_1,z_2 \in \partial \Rcal}
	\left|
		\frac{\Delta_N(z_1,z_2)^{-1}}
		{\left(\lambda_{k,N} - w_N(z_1)\right)\left(\lambda_{\ell,N} - w_N(z_1)\right)\left(\lambda_{k,N} - w_N(z_2)\right)\left(\lambda_{\ell,N} - w_N(z_2)\right)}
	\right|
	< \infty.
	\notag
\end{align}
Therefore, these bounds readily imply
\begin{align}
	\limsup_{N \to \infty} \max_{k,\ell} \vartheta_N(k,\ell) < \infty.
	\label{eq:upper_bound_vartheta}
\end{align}
We now express the integrand of $\vartheta_{N}(k,\ell)$ as a series of functions which are separable and symetric in $z_1,z_2$ 
(i.e. a function $g(z_1,z_2)$ is symetric separable if it can be written as $g(z_1,z_2)=\tilde{g}(z_1)\tilde{g}(z_2)$).
Notice that, except for $\Delta_N(z_1,z_2)$, all the functions appearing in the integrand in the definition of $\vartheta_N(k,l)$ are trivially sums of separable functions from their very definition.
From lemma \ref{lemma:properties_determinant}, we have
\begin{align}
	\left|\frac{\Delta_N(z_1,z_2)}{\left(1-u_N(z_1,z_2)\right)^2} - 1\right| < 1,
	\notag
\end{align}
and by writing
\begin{align}
	\Delta_N(z_1,z_2)^{-1} 
	= 
	\frac{1}{\left(1 - u_N(z_1,z_2)\right)^2 \left(1 - \frac{z_1 z_2 v_N(z_1,z_2) \tilde{v}_N(z_1,z_2)}{\left(1 - u_N(z_1,z_2)\right)^2}\right)},
	\notag
\end{align}
we obtain
\begin{align}	
	\Delta_N(z_1,z_2)^{-1} =
	\sum_{k \in \Nbb} \frac{\left(z_1 z_2 v_N(z_1,z_2) \tilde{v}_N(z_1,z_2)\right)^{k}}{\left(1 - u_N(z_1,z_2)\right)^{2k+2}}.
	\label{eq:det_serie_pre}
\end{align}
Using this time the fact that $\left|u_N(z_1,z_2)\right| < 1$ (see lemma \ref{lemma:properties_determinant}), we can further write
\begin{align}	
	\Delta_N(z_1,z_2)^{-1} =
	\sum_{k \in \Nbb} \ \sum_{l_1,\ldots,l_{2k+2} \in \Nbb} \left(z_1 z_2 v_N(z_1,z_2) \tilde{v}_N(z_1,z_2)\right)^{k} u_N(z_1,z_2)^{l_1+\ldots+l_{2k+2}}.
	\label{eq:det_serie}
\end{align}
Since the functions $u_N$, $v_N$ and $\tilde{v}_N$ are continuous on the compact set $\partial \Rcal \times \partial \Rcal$, the bound previously derived shows that 
the series of functions defining \eqref{eq:det_serie} is uniformly convergent on $\partial \Rcal \times \partial \Rcal$.
Consequently, we can rewrite the coefficients $\vartheta_{N}(k,\ell)$, defined in \eqref{eq:def_vartheta}, as
\begin{align}
	&\vartheta_{N}(k,\ell) = 
	\notag\\
	& 
	\sum_{n \in \Nbb} \ \sum_{l_1,\ldots,l_{2n+2} \in \Nbb}
	\left(\frac{1}{2 \pi \irm}\right)^2 \oint_{\partial \Rcal}\oint_{\partial \Rcal}
	\frac{\sigma^2\theta_N^{(k,\ell)}(z_1,z_2) \left(z_1 z_2 v_N(z_1,z_2) \tilde{v}_N(z_1,z_2)\right)^{n} u_N(z_1,z_2)^{l_1+\ldots+l_{2n+2}} w'_N(z_1) w'_N(z_2)}
	{
		2\left(\lambda_{k,N} - w_N(z_1)\right)\left(\lambda_{\ell,N} - w_N(z_1)\right)\left(\lambda_{k,N} - w_N(z_2)\right)\left(\lambda_{\ell,N} - w_N(z_2)\right)
	}
	 \drm z_1 \drm z_2,
	 \label{eq:vartheta_serie}
\end{align}
where $\theta_N^{(k,l)}(z_1,z_2)$ is defined in \eqref{eq:def_theta}.
In other words, we have written $\vartheta_{N}(k,\ell)$ as a convergent series of integrals of symetric separable functions. 
Consequently, $\vartheta_{N}(k,\ell)$ can be written as a series of squared single integrals, i.e. there exists a sequence of continuous functions $\left(g_N^{(p)}\right)_{p \in \Nbb}$ 
defined on $\partial \Rcal$ such that
\begin{align}
	\vartheta_{N}(k,\ell) = \sum_{p \in \Nbb} \left(\frac{1}{2 \pi \irm}\oint_{\partial \Rcal} g_{N}^{(p)}(z)\drm z\right)^2
	\notag
\end{align}
implying that $\vartheta_{N}(k,\ell) \geq 0$. This proves \eqref{eq:bound_var_general}.

To prove \eqref{eq:lb_var_nonspike}, we rely on the series expansion \eqref{eq:vartheta_serie} introduced above.
Using only one of the three terms in the definition of $\theta_N^{(k,\ell)}(z_1,z_2)$ (see \eqref{eq:def_theta}), and by only considering $n=0$ in the sum of the series in \eqref{eq:vartheta_serie},
we obtain the following lower-bound
\begin{align}
	&\vartheta_{N}(k,\ell) 	
	\geq
	\notag\\
	& 
	\left(\frac{1}{2 \pi \irm}\right)^2 \oint_{\partial \Rcal}\oint_{\partial \Rcal}
	\frac{\sigma^2\lambda_{k,N}\lambda_{\ell,N} v_N(z_1,z_2)\left(1+\sigma^2 c_N m_N(z_1)\right)^{-1}\left(1+\sigma^2 c_N m_N(z_2)\right)^{-1} w'_N(z_1) w'_N(z_2)}
	{
		2\left(\lambda_{k,N} - w_N(z_1)\right)\left(\lambda_{\ell,N} - w_N(z_1)\right)\left(\lambda_{k,N} - w_N(z_2)\right)\left(\lambda_{\ell,N} - w_N(z_2)\right)
	}
	 \drm z_1 \drm z_2,
\end{align}
From the definition of $v_N(z_1,z_2)$ (see \eqref{def:v_vtilde}), we have
\begin{align}
	\frac{v_N(z_1,z_2)}{\left(1+\sigma^2 c_N m_N(z_1)\right)\left(1+\sigma^2 c_N m_N(z_2)\right)} =
	\frac{\sigma^2}{N} \sum_{m=1}^M \frac{1}{\left(\lambda_{m,N}-w_N(z_1)\right)\left(\lambda_{m,N}-w_N(z_2)\right)},
	\notag
\end{align}
and a usual change of variable gives
\begin{align}
	\vartheta_{N}(k,\ell) 
	\geq 
	\lambda_{k,N}\lambda_{\ell,N} \frac{\sigma^4}{2N} \sum_{m=1}^{M}
	\left(
		\frac{1}{2 \pi \irm} \oint_{w_N(\partial\Rcal)}
		\frac{\drm w}{\left(\lambda_{m,N}-w\right)\left(\lambda_{k,N} - w\right)\left(\lambda_{\ell,N} - w\right)}
	\right)^2.
	\notag
\end{align}
Residue's theorem thus implies that for $k,\ell \leq K$,
\begin{align}
	\vartheta_{N}(k,\ell) \geq \frac{M-K}{N} \frac{\sigma^4}{2\lambda_{k,N}\lambda_{\ell,N}}.
	\label{eq:lb_var_nonspike1}
\end{align}
In the same way, for $k,\ell \geq K+1$, we obtain
\begin{align}
	\vartheta_{N}(k,\ell) 
	&\geq 
	\left(\frac{1}{2 \pi \irm}\right)^2 \oint_{\partial \Rcal}\oint_{\partial \Rcal}
	\frac{\sigma^2 z_1 z_2 \left(1+\sigma^2 c_N m_N(z_1)\right)\left(1+\sigma^2 c_N m_N(z_2)\right) \tilde{v}_N(z_1,z_2) w'_N(z_1) w'_N(z_2)}
	{
		2\left(\lambda_{k,N} - w_N(z_1)\right)\left(\lambda_{\ell,N} - w_N(z_1)\right)\left(\lambda_{k,N} - w_N(z_2)\right)\left(\lambda_{\ell,N} - w_N(z_2)\right)
	}
	 \drm z_1 \drm z_2
	\notag\\
	&=
	 \frac{\sigma^4}{2N} \sum_{m = 1}^K \frac{1}{\lambda_{m,N}^2}.
	 \label{eq:lb_var_nonspike2}
\end{align}
For $k \leq K$ and $\ell \geq K+1$, we have
\begin{align}
	\vartheta_{N}(k,\ell) 
	\geq 
	\left(\frac{1}{2 \pi \irm}\right)^2 \oint_{\partial \Rcal}\oint_{\partial \Rcal}
	\frac{\sigma^2\lambda_{k,N}\left(1-u_N(z_1,z_2)\right) \Delta_N(z_1,z_2)^{-1} w'_N(z_1) w'_N(z_2)}
	{
		2 w_N(z_1) w_N(z_2) \left(\lambda_{k,N} - w_N(z_1)\right)\left(\lambda_{k,N} - w_N(z_2)\right)
	}
	 \drm z_1 \drm z_2.
	 \notag
\end{align}
Using \eqref{eq:det_serie_pre} and performing a serie expansion of $\left(1-u_N(z_1,z_2)\right)^{-1}$, we obtain as well
\begin{align}
	\vartheta_{N}(k,\ell) 
	\geq 
	\left(\frac{1}{2 \pi \irm}\right)^2 \oint_{\partial \Rcal}\oint_{\partial \Rcal}
	\frac{\sigma^2\lambda_{k,N} w'_N(z_1) w'_N(z_2)}
	{
		2w_N(z_1) w_N(z_2) \left(\lambda_{k,N} - w_N(z_1)\right)\left(\lambda_{k,N} - w_N(z_2)\right)
	}
	 \drm z_1 \drm z_2
	 =
	 \frac{\sigma^2}{2\lambda_{k,N}}.
	 \label{eq:lb_var_nonspike3}
\end{align}
\eqref{eq:lb_var_nonspike} will follow from \eqref{eq:lb_var_nonspike1}, \eqref{eq:lb_var_nonspike2} and \eqref{eq:lb_var_nonspike3} and lemma \ref{lemma:conseq_sep} 
in section \ref{section:spiked_model}.

	\subsection{Asymptotics of $\vartheta_{N}(k,l)$}
		\label{section:asymptotic_variance_K_fixed}
		
We now show \eqref{eq:var_spike}, by assuming that $K$ is independent of $N$.
In this case, using the results of section \ref{section:spiked_model}, it is not difficult to show that
\begin{align}
	\max_{k,l} \left|\vartheta_{N}(k,\ell) - \check{\vartheta}_N(k,\ell)\right| \xrightarrow[N \to \infty]{} 0,
	\label{eq:conv_var_spike}
\end{align}
where
\begin{align}
	\check{\vartheta}_N(k,\ell) =
	\left(\frac{1}{2 \pi \irm}\right)^2 \oint_{\partial \Rcal}\oint_{\partial \Rcal} 
	\frac
	{
		\frac{\sigma^2}{2}
		\left(\sigma^2 + \frac{\sigma^2 c_N \lambda_{k,N}\lambda_{\ell,N}}{w(z_1)w(z_2)} + (\lambda_{k,N}+\lambda_{\ell,N})\right)
		\Delta(z_1,z_2)^{-1}
		w'(z_1)w'(z_2) 
	}
	{
		\left(\lambda_{k,N} - w(z_1)\right) \left(\lambda_{\ell,N} - w(z_1)\right) 
		\left(\lambda_{k,N}-  w(z_2)\right) \left(\lambda_{\ell,N}- w(z_2)\right)
	}
	\drm z_1 \drm z_2
\end{align}
with $m(z), w(z)$ and $\Delta(z_1,z_2)$ defined in \eqref{eq:def_m_check_N}, \eqref{eq:def_w_check_N} and \eqref{def:determinant_MP}.
From \eqref{eq:det_serie}, it is clear that the following serie expansion also holds,
\begin{align}
	\Delta(z_1,z_2)^{-1} 
	= \left(1 - \frac{\sigma^4 c_N}{w(z_1) w(z_2)}\right)^{-1}
	= \sum_{n \in \Nbb} \left(\frac{\sigma^4 c_N}{w(z_1) w(z_2)}\right)^n,
	\notag
\end{align}
uniformly on $\partial\Rcal \times \partial\Rcal$, and thus
\begin{align}
	&\check{\vartheta}_N(k,\ell) =
	\notag\\
	&
	\sum_{n \in \Nbb}
	\left(\frac{1}{2 \pi \irm}\right)^2 \oint_{\partial \Rcal}\oint_{\partial \Rcal} 
	\frac
	{
		\frac{\sigma^2}{2}
		\left(\sigma^2 + \frac{\sigma^2 c_N \lambda_{k,N}\lambda_{\ell,N}}{w(z_1) w(z_2)} + (\lambda_{k,N}+\lambda_{\ell,N})\right)
		\left(\sigma^4 c_N\right)^n w(z_1)^{-n} w(z_2)^{-n} w'(z_1) w'(z_2)
	}
	{
		\left(\lambda_{k,N} - w(z_1)\right) \left(\lambda_{\ell,N} - w(z_1)\right) 
		\left(\lambda_{k,N}- w(z_2)\right) \left(\lambda_{\ell,N}- w(z_2)\right)
	}
	\drm z_1 \drm z_2.
	 \notag
\end{align}
By expressing the previous expression as square of single integrals, we obtain
\begin{align}
	&\check{\vartheta}_N(k,\ell) =
	\notag\\
	&\qquad
	\frac{\sigma^2}{2}\left(\sigma^2 + (\lambda_{k,N}+\lambda_{\ell,N})\right)
	\sum_{n \in \Nbb}
	\left(\sigma^4 c_N\right)^n
	\left[
		\frac{1}{2 \pi \irm} \oint_{w(\partial \Rcal)}
		\frac
		{
			1
		}
		{
			w^{n}\left(\lambda_{k,N} - w\right) \left(\lambda_{\ell,N} - w\right) 
		}
		\drm w
	\right]^2
	\notag\\
	&\qquad\qquad
	+
	\frac{\sigma^4 c_N}{2} \lambda_{k,N}\lambda_{\ell,N}
	\sum_{n \in \Nbb}
	\left(\sigma^4 c_N\right)^n
	\left[
		\frac{1}{2 \pi \irm} \oint_{w(\partial \Rcal)}
		\frac
		{
			 1
		}
		{
			w^{n+1}\left(\lambda_{k,N} - w\right) \left(\lambda_{\ell,N} - w\right) 
		}
		\drm w
	\right]^2
	.
	 \notag
\end{align}
Using classical residue computation, we eventually end up with \eqref{eq:var_spike},
\begin{align}
	\check{\vartheta}_N(k,\ell) =
		\frac
		{
			\sigma^4 c_N \left(\lambda_{k,N} \lambda_{\ell,N} 
			+ (\lambda_{k,N}+\lambda_{\ell,N})\sigma^2 
			+ \sigma^4\right)\left(\lambda_{k,N} \lambda_{\ell,N} 
			+ \sigma^4 c_N\right)
		}
		{2\left(\lambda_{k,N}^2 - \sigma^4 c_N\right)\left(\lambda_{\ell,N}^2 - \sigma^4 c_N\right)\left(\lambda_{k,N}\lambda_{\ell,N} - \sigma^4 c_N\right)}
		\left(1 - \mathbb{1}_{[K+1,M]}(k) \mathbb{1}_{[K+1,M]}(\ell)\right).
	 \label{eq:var_spike2}
\end{align}
We can easily show that 
\begin{align}
	\liminf_{N \to \infty} \min_{(k,\ell) \not\in \left\{K+1,\ldots,M\right\}^2}  \check{\vartheta}_N(k,\ell) > 0.
	\notag
\end{align}
and we also have the boundedness
\begin{align}
	\limsup_{N \to \infty} \max_{k,\ell} \check{\vartheta}_N(k,\ell) < \infty,
	\notag
\end{align}
which is ensured in the general case ($K$ not necessarily fixed) by \eqref{eq:upper_bound_vartheta}, but which also comes from lemma \ref{lemma:conseq_sep2} in section 
\ref{section:spiked_model}.

		\subsection{Solution to the differential equation}
		\label{section:equ_diff_sol}
	
Recall that the differential equation \eqref{eq:diff_eq2}
\begin{align}
	\psi'_N(u) =  \left(\irm \sqrt{N} \Re\left(\xi_N \left(\d_{1,N}^*\d_{2,N} - \eta_N\right)\right) - u \xibs_N^T \Gammabs_N \xibs_N) \right) \psi_N(u) + \frac{\epsilon_N(u)}{\sqrt{N}},
	\label{eq:diff_eq3}
\end{align}	
where $(\xi_N)$ is any deterministic sequence such that $\limsup_N |\xi_N|<\infty$ and $\xibs_N=[\Re(\xi_N),\Im(\xi_N)]^T$ and $u \mapsto \epsilon_N(u)$ 
a generic continuously differentiable function such that
\begin{align}
	\limsup_{N\to \infty}\left\{\left|\epsilon_N(u)\right|, \left|\epsilon'_N(u)\right|\right\} \leq \Prm(u),
	\notag
\end{align}
with $\Prm$ is a polynomial independent of $N$ with positive coefficient. 
By differentiating \eqref{eq:diff_eq3} with respect to $u$ and using \eqref{eq:bias_and_variance}, one can check that
\begin{align}
	\Ebb\left|\sqrt{N} \Re\left(\xi_N \left(\hat{\gamma}_N - \left(\d_{1,N}^*\d_{2,N} - \eta_N\right)\right)\right)\right|^2 = \xibs_N^T \Gammabs_N \xibs_N + \Ocal\left(\frac{1}{\sqrt{N}}\right),
	\notag
\end{align}
where $\hat{\gamma}_N$ is defined in \eqref{def:gammahat}. This implies
\begin{align}
	\Re\left(\xi_N \left(\hat{\gamma}_N - \left(\d_{1,N}^*\d_{2,N} - \eta_N\right)\right)\right) =
	\begin{cases}
		\Ocal_{\Pbb}\left(\frac{1}{\sqrt{N}}\right) & \text{if } \liminf_{N} \xibs_N^T \Gammabs_N \xibs_N > 0
		\\
		o_{\Pbb}\left(\frac{1}{\sqrt{N}}\right) & \text{otherwise}
	\end{cases}.
	\notag
\end{align}
which shows \eqref{eq:tightness_real_part}.
By assuming that $\liminf_N \xibs_N^T \Gammabs_N \xibs_N > 0$, we can replace $\xi_N$ by $\frac{\xi_N}{\sqrt{\xibs_N^T \Gammabs_N \xibs_N}}$ in \eqref{eq:diff_eq2} 
without modifying the boundedness property of $\epsilon_N(u)$, and thus
\begin{align}
	\psi'_N(u) =  
	\left(\irm \frac{\sqrt{N} \Re\left(\xi_N \left(\d_{1,N}^*\d_{2,N} - \eta_N\right)\right)}{\sqrt{\xibs_N^T \Gammabs_N \xibs_N}} - u\right) \psi_N(u) + \frac{\epsilon_N(u)}{\sqrt{N}},
	\label{eq:diff_eq4}
\end{align}
\eqref{eq:diff_eq4} being a classical nonhomogeneous linear differential equation of the first order, we easily obtain that
\begin{align}
	\Ebb\left[\mathrm{exp}\left(\irm u \frac{\sqrt{N} \Re\left(\xi_N \left(\hat{\gamma}_N - \left(\d_{1,N}^*\d_{2,N} - \eta_N\right)\right)\right)}{\sqrt{\xibs_N^T \Gammabs_N \xibs_N}}\right)\right] =  
	\mathrm{exp}\left(-u^2/2\right) + o(1),
	\notag
\end{align}
which proves \eqref{eq:conv_distrib_eta}.

\section{Appendix}
\label{section:appendix}

	\subsection{Proof of formula \eqref{eq:charac_fun_1} and proposition \ref{proposition:bil_form}}
	
In this section, we prove formula \eqref{eq:charac_fun_1} and proposition \ref{proposition:bil_form} respectively. For that purpose, we use two tools, an integration by part formula and
a Poincar\'e inequality for Gaussian variables, which are well-known in the field of random matrix theory since the work of Pastur \cite{pastur2005simple}.
We first note that every function $f: \Cbb \mapsto \Cbb$ can be written as $f(z)=\tilde{f}\left(\Re(z),\Im(z)\right)$.
If $\tilde{f} \in \Ccal^1(\Rbb^2,\Cbb)$, we define the usual differential operators
\begin{align}
	\frac{\partial f(z)}{\partial z} = 
	\frac{1}{2}
	\left(
		\frac{\partial \tilde{f}(x,y)}{\partial x} 
		- \irm\frac{\partial \tilde{f}(x,y)}{\partial y} 
	\right)
	\quad\text{and}\quad
	\frac{\partial f(z)}{\partial \overline{z}} = 
	\frac{1}{2} 
	\left(
		\frac{\partial \tilde{f}(x,y)}{\partial x} 
		+ \irm\frac{\partial \tilde{f}(x,y)}{\partial y} 
	\right).
	\notag
\end{align}
In this context, we say that $f$ is continuously differentiable if $\tilde{f}$ is continuously differentiable.
The following lemma gives the integration by part formula and the Poincar\'e inequality.
\begin{lemma}
	\label{lemma:IPP_Poincare}
	Let $z_1 = x_1 + \irm y_1, \ldots, z_n= x_n + \irm y_n$ be $n$ i.i.d. $\Ncal_{\Cbb}(0,\rho^2)$ variables and let $f$ a continuously differentiable function defined on $\Cbb^n$ 
	with polynomially bounded partial derivatives. Then, if $\z = (z_1, \ldots, z_n)^{T}$, it holds that
	\begin{align}
		\Ebb\left[z_k f(\z)\right] = \rho^2 \Ebb\left[\frac{\partial f(\z)}{\partial \overline{z_k}}\right]
		\quad\text{and}\quad
		\Ebb\left[\overline{z_k} f(\z)\right] = \rho^2 \Ebb\left[\frac{\partial f(\z)}{\partial z_k}\right].
		\notag
	\end{align}
	Moreover,
        \begin{align}
		\Vbb\left[f(\z)\right] \leq \rho^2 \sum_{k=1}^n 
		\left(
			\Ebb \left|\frac{\partial f(\z)}{\partial z_k}\right|^{2} 
			+ \Ebb \left|\frac{\partial f(\z)}{\partial \overline{z_k}}\right|^{2}
		\right).
		\label{eq:ineq_poin}
        \end{align}
\end{lemma}
Hereafter and in all the remainder of this appendix, $\epsilon_N(u,z_1,z_2)$ will denote a generic continuous function on $\Rbb \times \partial \Rcal \times \partial \Rcal$ such that 
$u \mapsto \epsilon_N(u,z_1,z_2)$ is continuously differentiable and
\begin{align}
	\limsup_{N \to \infty} \sup_{(z_1,z_2) \in \partial \Rcal \times \partial \Rcal} 
	\left\{\left|\epsilon_N(u,z_1,z_2)\right|, \left|\frac{\partial \epsilon_N(u,z_1,z_2)}{\partial u}\right|\right\} < \Prm(u),
	\notag
\end{align}
with $\Prm(u)$ a polynomial with positive coefficients. $\epsilon_N(u,z_1,z_2)$ may take different values from one line to another.
We will also keep the notation $\epsilon_N(z_1,z_2)$ and $\epsilon_N(u,z_1)$ if $\epsilon_N(u,z_1,z_2)$ does not depend on $u$ or $z_2$.

We recall the quantity $\hat{\gamma}_N$, which is the regularized estimator defined in \eqref{def:gammahat} by
\begin{align}
	\hat{\gamma}_N = \frac{1}{2 \pi \irm} \oint_{\partial \Rcal} \d_{1,N}^* \Q_N(z) \d_{2,N} \chi_N \frac{w'_N(z)}{1+\sigma^2 c_N m_N(z)} \drm z.
	\notag
\end{align}
Using lemma \ref{lemma:IPP_Poincare}, it is not difficult to obtain, as in \cite[Lem. 5.7]{hachem2012large}, the following useful properties.
\begin{corollary}
        \label{corollary:tool1}
	Let $(h_N)_{N   \geq 1}$ be a sequence of continuously differentiable functions defined on $\Cbb^{M(M+N)}$ with polynomially bounded partial derivatives satisfying the condition 
	\begin{align}
	        \limsup_{N \to \infty} \sup_{z \in \partial \Rcal} \left| h_N\left(\Vec\left(\Q_N(z)\right),\Vec(\Sigmabs_N)\right) \chi_N \right| < \infty. 
	\end{align}
	Then, for all $k \in \Nbb^*$, we have
	\begin{align}
		&\Ebb\left[h_N\left(\Vec\left(\Q_N(z)\right),\Vec(\Sigmabs_N)\right) \chi_N \erm^{\irm u \sqrt{N}\Re\left(\hat{\gamma}_N\right)}\right]
		\notag\\
		=
		&\Ebb\left[h_N\left(\Vec\left(\Q_N(z)\right),\Vec(\Sigmabs_N)\right)\chi_N^k \erm^{\irm u \sqrt{N}\Re\left(\hat{\gamma}_N\right)}\right]
		+
		\frac{\epsilon_{N}(u,z)}{N^p}.
		\label{eq:IPPder1}
	\end{align}
	for all $p \in \Nbb$. Moreover, 	
	\begin{align}
		&\Ebb\left[ W_{i,j} h_N\left(\Vec\left(\Q_N(z)\right),\Vec(\Sigmabs_N)\right) \chi_N \erm^{\irm u \sqrt{N}\Re\left(\hat{\gamma}_N\right)}\right]
		=
		\notag\\
		&\qquad\qquad
		\frac{\sigma^2}{N} 
		\Ebb
		\left[
			\frac{\partial}{\partial \overline{W}_{i,j}}
			\left\{h_N\left(\Vec\left(\Q_N(z)\right),\Vec(\Sigmabs_N)\right)\erm^{\irm u \sqrt{N}\Re\left(\hat{\gamma}_N\right)}\right\} 
			\chi_N 
		\right]
		+
		\frac{\epsilon_{N}(u,z)}{N^p},
		\label{eq:IPPder2}
	\end{align}
	and
	\begin{align}
		&\Ebb\left[ \overline{W}_{i,j} h_N\left(\Vec\left(\Q_N(z)\right),\Vec(\Sigmabs_N)\right) \chi_N \erm^{\irm u \sqrt{N}\Re\left(\hat{\gamma}_N\right)}\right]
		=
		\notag\\
		&\qquad\qquad
		\frac{\sigma^2}{N} 
		\Ebb
		\left[
			\frac{\partial}{\partial W_{i,j}}
			\left\{h_N\left(\Vec\left(\Q_N(z)\right),\Vec(\Sigmabs_N)\right)\erm^{\irm u \sqrt{N}\Re\left(\hat{\gamma}_N\right)}\right\} 
			\chi_N 
		\right]
		+
		\frac{\epsilon_{N}(u,z)}{N^p}.
		\label{eq:IPPder3}
	\end{align}
	Finally, we have
	\begin{align}
		\Ebb\left[h_N\left(\Vec\left(\Q_N(z)\right),\Vec(\Sigmabs_N)\right) D_{i,j} \erm^{\irm u \sqrt{N}\Re\left(\hat{\gamma}_N\right)}\right] = \frac{\epsilon_{N}(u,z)}{N^p}
		\notag
	\end{align}
	and
	\begin{align}
		\Ebb\left[h_N\left(\Vec\left(\Q_N(z)\right),\Vec(\Sigmabs_N)\right) D_{i,j}^* \erm^{\irm u \sqrt{N}\Re\left(\hat{\gamma}_N\right)}\right] = \frac{\epsilon_{N}(u,z)}{N^p},
		\notag
	\end{align}
	for all $p \in \Nbb$, with $D_{i,j} = \frac{\partial}{\partial W_{i,j}^*} \left\{\chi_N\right\}$.
\end{corollary}
We now introduce the matrix $\R_N(z)$ given by
\begin{align}
	\R_N(z) = 
	\left(\frac{\B_N\B_N^*}{1+\sigma^2 c_N \Ebb\left[\hat{m}_N(z_1)\chi_N\right]} - z \left(1+\sigma^2 c_N \Ebb\left[\hat{m}_N(z_1)\chi_N\right]\right) + \sigma^2 (1-c_N)\right)^{-1}.
	\notag
\end{align}
Matrix $\R_N(z)$ is similar to $\T_N(z)$ where we just replaced $m_N(z)$ by $\Ebb[\hat{m}_N(z) \chi_N]$. Since $\Ebb[\hat{m}_N(z) \chi_N] - m_N(z) \to_N 0$, it is of course
expected that $\R_N(z)$ will be close to $\T_N(z)$ asymptotically. This result is given by the following lemma.
\begin{corollary}[$\text{\cite[Lem. 3.10, 5.5 \& 5.6]{hachem2012large}}$]
	\label{corollary:tool2}
	Let $\M_N(z)$ a sequence of deterministic matrices of size $M \times M$ such that 
	\begin{align}
		\limsup_{N \to \infty} \sup_{z \in \partial \Rcal} \left\|\M_N(z)\right\| < \infty.
		\notag
	\end{align}
	Then we have
	\begin{align}
		\limsup_{N \to \infty} \sup_{z \in \Rcal} \left|1 + \sigma^2 c_N \Ebb[\hat{m}_N(z) \chi_N]\right|^{-1} < \infty
		\quad\text{and}\quad
		\limsup_{N \to \infty} \sup_{z \in \Rcal} \left\|\R_N(z)\right\| < \infty.
		\notag
	\end{align}
	Moreover, $\Ebb\left[\hat{m}_N(z) \chi_N\right] - m_N(z) = \frac{\epsilon_N(z)}{N^2}$ and
	\begin{align}
		\d_{1,N}^* (\Ebb[\Q_N(z) \chi_N]-\R_N(z)) \M_N(z) \d_{2,N} = \frac{\epsilon_N(z)}{N^{3/2}}
		\quad\text{and}\quad
		\d_{1,N}^* (\R_N(z)-\T_N(z)) \M_N(z) \d_{2,N} = \frac{\epsilon_N(z)}{N^{3/2}}.
		\notag
	\end{align}
\end{corollary}
We now give a result on the variance of certain expressions involving the resolvent, whose proof is a standard application of the Poincar\'e inequality 
(see e.g. \cite[Lem. 5.8]{hachem2012large}, \cite[Lem. 10]{vallet2012improved}), and is therefore omitted .
\begin{corollary}
	\label{corollary:tool3}
	Let $\M_N(z_1,z_2)$ a sequence of deterministic matrices of size $M \times M$ such that 
	\begin{align}
		\limsup_{N \to \infty} \sup_{(z_1,z_2) \in \partial \Rcal^2} \left\|\M_N(z_1,z_2)\right\| < \infty.
		\notag
	\end{align}
	Let $\Prm$ a $5$-variate polynomial function independent of $N$ such that the matrix
	\begin{align}
		\Xibs_N(z_1,z_2) = \Prm\left(\Q_N(z_1), \Q_N(z_2), \Sigmabs_N, \Sigmabs_N^*, \M_N(z_1,z_2)\right) \chi_N.
		\notag
	\end{align}
	is properly defined. Then, it holds that
	\begin{align}
		\limsup_{N \to \infty} \sup_{(z_1,z_2) \in \partial \Rcal^2} \Vbb\left[\Tr \Xibs_N(z_1,z_2) \right] < \infty
		\quad\text{and}\quad
		\limsup_{N \to \infty} \sup_{(z_1,z_2) \in \partial \Rcal^2} \Vbb\left[\sqrt{N}\d_{1,N}^* \Xibs_N(z_1,z_2) \d_{2,N}\right] < \infty.
		\notag
	\end{align}
	Moreover, it also holds that $\Vbb\left[\d_{1,N}^* \Xibs_N(z_1,z_2) \d_{2,N}  \erm^{\irm u \sqrt{N}\Re\left(\hat{\gamma}_N\right)}\right]$ 
	is a term behaving as $\epsilon_N(u,z_1,z_2)$.
\end{corollary}

	\subsubsection{Proof of formula \eqref{eq:charac_fun_1}}
	\label{appendix:gaussiancomp1}

By expressing the derivative of $\erm^{\irm u \sqrt{N}\Re\left(\hat{\gamma}_N\right)}$ w.r.t. $W_{i,j}$, we obtain
\begin{align}
	&\frac{\partial}{\partial W_{i,j}} \left\{\erm^{\irm u \sqrt{N}\Re\left(\hat{\gamma}_N\right)}\right\}
	=
	\notag\\
	& \frac{\irm u}{2} \frac{1}{2 \pi \irm} \oint_{\partial \Rcal} 
	\left(- \d_{1,N}^* \Q_N(z) \e_i \e_j^* \Sigmabs_N^* \Q_N(z) \d_{2,N} \chi_N + \d_{1,N}^* \Q_N(z) \d_{2,N} D_{i,j}^*\right)
	\frac{w'_N(z)}{1+\sigma^2 c_N m_N(z)} \erm^{\irm u \sqrt{N}\Re\left(\hat{\gamma}_N\right)} \drm z
	\notag\\
	& + \frac{\irm u}{2} \frac{1}{2 \pi \irm} \oint_{\partial \Rcal} 
	\left(- \d_{2,N}^* \Q_N(z) \e_i \e_j^* \Sigmabs_N^* \Q_N(z) \d_{1,N} \chi_N + \d_{2,N}^* \Q_N(z) \d_{1,N} D_{i,j}^*\right)
	\frac{w'_N(z)}{1+\sigma^2 c_N m_N(z)} \erm^{\irm u \sqrt{N}\Re\left(\hat{\gamma}_N\right)} \drm z,
	\label{eq:der_exp1}
\end{align}
where $D_{i,j}$ is defined in corollary \ref{corollary:tool1}.
The derivative with respect to $\overline{W}_{i,j}$ is computed in the same way.
To develop $\Ebb\left[\d_{1,N}^* \Q_N(z) \d_{2,N} \chi_N \erm^{\irm u \sqrt{N}\Re\left(\hat{\gamma}_N\right)}\right]$, we start with the classical resolvent identity 
$\Q_N(z)\Sigmabs_N\Sigmabs_N^* = \I + z \Q_N(z)$.
By applying corollaries \ref{corollary:tool1} and \ref{corollary:tool2} several times, long but straightforward computations lead to
\begin{align}
	&\Ebb\left[\left[\Q_N(z_1)\Sigmabs_N\Sigmabs_N^*\right]_{i,j} \chi_N \erm^{\irm u \sqrt{N}\Re\left(\hat{\gamma}_N\right)}\right]
	=
	\notag\\
	&
	\alpha_{N}^{(i,j)}(u,z_1)
	+ \frac{\irm u \sigma^2}{2 \sqrt{N}} \frac{1}{2 \pi \irm} \oint_{\partial \Rcal} 
 	\frac{\left(\beta^{(i,j)}_N(u,z_1,z_2) + \tilde{\beta}^{(i,j)}_N(u,z_1,z_2)\right) w'_N(z_2)}
 	{\left(1+\sigma^2 c_N \Ebb\left[\hat{m}_N(z_1) \chi_N\right]\right)\left(1+\sigma^2 c_N m_N(z_2)\right)} \drm z_2
	\notag\\
	&	
	- \left[\Deltabs_N(u, z_1)\right]_{i,j} + \frac{\epsilon_N(u,z_1)}{N^p},
	\label{eq:exp_Q_sig_sig}
\end{align}
for all $p \in  \Nbb$, where $\alpha_{N}^{(i,j)}(u,z_1)$ is given by
\begin{align}
	&\alpha_{N}^{(i,j)}(u,z_1) = 
	\notag\\
	&\qquad\qquad
	\frac{\Ebb\left[\left[\Q_N(z_1)\B_N\B_N^*\right]_{i,j} \chi_N \erm^{\irm u \sqrt{N}\Re\left(\hat{\gamma}_N\right)}\right]}{1+\sigma^2 c_N \Ebb\left[\hat{m}_N(z_1) \chi_N\right]}
	+
	\frac{\sigma^2 \Ebb\left[\left[\Q_N(z)\right]_{i,j} \chi_N \erm^{\irm u \sqrt{N}\Re\left(\hat{\gamma}_N\right)}\right]}{{1+\sigma^2 c_N \Ebb\left[\hat{m}_N(z_1) \chi_N\right]}}
	\notag\\
	&\qquad\qquad\qquad\qquad
	-
	\frac{\Ebb\left[\left[\Q_N(z_1)\right]_{i,j}\erm^{\irm u \Re\left(\hat{\gamma}_N\right)}\right]}{\left(1+\sigma^2 c_N \Ebb\left[\hat{m}_N(z_1)\chi_N\right]\right)^2} 
	\frac{\sigma^2}{N} \Tr \B_N^* \Ebb\left[\Q_N(z_1)\chi_N\right] \B_N
	,
	\notag
\end{align}
$\beta_{N}^{(i,j)}(u,z_1,z_2)$ and $\tilde{\beta}_{N}^{(i,j)}(u,z_1,z_2)$ respectively by
\begin{align}
	&\beta_{N}^{(i,j)}(u,z_1,z_2)=
	\notag\\
	&\qquad
	\Ebb\left[\d_{1,N}^* \Q_N(z_2) \e_j \e_i^*\Q_N(z_1) \B_N\Sigmabs_N^*\Q_N(z_2)\d_{2,N}\chi_N^2\erm^{\irm u \sqrt{N}\Re\left(\hat{\gamma}_N\right)}\right]
	\notag\\
	&\qquad\qquad
	+
	\Ebb\left[\d_{2,N}^* \Q_N(z_2) \e_j \e_i^*\Q_N(z_1) \B_N\Sigmabs_N^*\Q_N(z_2)\d_{1,N}\chi_N^2\erm^{\irm u \sqrt{N}\Re\left(\hat{\gamma}_N\right)}\right]
	\notag
\end{align}
and
\begin{align}
	&\tilde{\beta}_{N}^{(i,j)}(u,z_1,z_2)=
	\notag\\
	&\qquad
	\Ebb\left[\d_{1,N}^* \Q_N(z_2) \Sigmabs_N\Sigmabs_N^*\e_j \e_i^*\Q_N(z_1) \Q_N(z_2)\d_{2,N}\chi_N^2\erm^{\irm u \sqrt{N}\Re\left(\hat{\gamma}_N\right)}\right]
	\notag\\
	&\qquad\qquad
	+
	\Ebb\left[\d_{2,N}^* \Q_N(z_2)  \Sigmabs_N\Sigmabs_N^* \e_j \e_i^*\Q_N(z_1)\Q_N(z_2)\d_{1,N}\chi_N^2\erm^{\irm u \sqrt{N}\Re\left(\hat{\gamma}_N\right)}\right]
	\notag
\end{align}
and where finally the matrix $\Deltabs_N(u,z_1)$ is given by
\begin{align}
	&\Deltabs_{N}(u,z_1) =
	\notag\\
	&
	-\frac{\Ebb\left[\Q_N(z_1)\chi_N \erm^{\irm u \sqrt{N}\Re\left(\hat{\gamma}_N\right)}\right]}
	{\left(1+\sigma^2 c_N \Ebb\left[\hat{m}_N(z_1)\chi_N\right]\right)^2}
	\Ebb
	\left[
	        \left(\frac{\sigma^2}{N}\Tr \Q_N(z_1)\chi_N - \Ebb\left[\frac{\sigma^2}{N}\Tr \Q_N(z_1)\chi_N\right]\right)
	        \frac{\sigma^2}{N}\Tr \Sigmabs_N^*\Q_N(z_1)\B_N\chi_N
	\right]
	\notag\\
	&
	+
	\frac{1}{1+\sigma^2 c_N \Ebb\left[\hat{m}_N(z_1)\chi_N\right]}
	\Ebb
	\left[
	        \left(\frac{\sigma^2}{N}\Tr\Sigmabs_N^*\Q_N(z_1)\B_N \chi_N - \Ebb\left[\frac{\sigma^2}{N}\Tr\Sigmabs_N^*\Q_N(z_1)\B_N \chi_N\right]\right)           
	        \Q_N(z)\chi_N \erm^{\irm u \sqrt{N}\Re\left(\hat{\gamma}_N\right)}
	\right]
	\notag\\
	&
	+ \frac{1}{1+\sigma^2 c_N \Ebb\left[\hat{m}_N(z_1)\chi_N\right]}
	\Ebb
	\left[
	        \left(\frac{\sigma^2}{N} \Tr\Q_N(z_1)\chi_N - \Ebb\left[\frac{\sigma^2}{N} \Tr\Q_N(z_1)\chi_N\right]\right) 
	        \Q_N(z_1)\Sigmabs_N\Sigmabs_N^*\chi_N \erm^{\irm u \sqrt{N}\Re\left(\hat{\gamma}_N\right)}
	\right].
	\label{eq:defDeltaz}
\end{align}
From \eqref{eq:exp_Q_sig_sig} and the resolvent identity, we obtain
\begin{align}
	&\Biggl[
		\Ebb\left[\Q_N(z_1)\chi_N \erm^{\irm u \sqrt{N}\Re\left(\hat{\gamma}_N\right)} \right] 
		\left(\frac{\B_N\B_N^*}{1+\sigma^2 c_N \Ebb\left[\hat{m}_N(z_1)\chi_N\right]} - z\left(1+\sigma^2 c_N \tilde{\tau}_N(z_1)\right)\I\right)
	\Biggr]_{i,j}
	=
	\notag\\
	&\qquad
	\Ebb\left[[\I]_{i,j}\chi_N \erm^{\irm u \sqrt{N}\Re\left(\hat{\gamma}_N\right)}\right] 
	- \frac{\irm u \sigma^2}{2 \sqrt{N}} \frac{1}{2 \pi \irm} \oint_{\partial \Rcal} 
 	\frac{\left(\beta^{(i,j)}_N(u,z_1,z_2) + \tilde{\beta}^{(i,j)}_N(u,z_1,z_2)\right) w'_N(z_2)}
 	{\left(1+\sigma^2 c_N \Ebb\left[\hat{m}_N(z_1) \chi_N\right]\right)\left(1+\sigma^2 c_N m_N(z_2)\right)} \drm z_2
	\notag\\
	&\qquad\qquad
	+ [\Deltabs_N(u,z_1)]_{i,j} + \frac{\epsilon_N(u,z_1)}{N^p},
	\notag
\end{align}
where
\begin{align}
	\tilde{\tau}_N(z_1) = 
	- \frac{\sigma}{z\left(1+\sigma^2 c_N \Ebb\left[\hat{m}_N(z_1)\chi_N\right]\right)}
	\left(1 - \frac{1}{N} \Tr\left(\frac{\B_N^*\Ebb\left[\Q_N(z_1)\right]\B_N^*}{1+\sigma^2 c_N \Ebb\left[\hat{m}_N(z_1)\chi_N\right]}\right)\right)
	\notag
\end{align}
Straightforward algebra gives 
\begin{align}
	\tilde{\tau}_N(z_1) =  \Ebb\left[\left(\hat{m}_N(z_1) - \frac{(1-1/c_N)}{z_1}\right) \chi_N\right] + \frac{1}{z_1}\frac{1}{M} \Tr \Deltabs_N(0,z_1) + \frac{\epsilon_N(z_1)}{N^p},
	\notag
\end{align}
and finally we get
\begin{align}
	&\Ebb\left[\d_{1,N}^*\Q_N(z_1)\d_{2,N}\chi_N \erm^{\irm u \sqrt{N}\Re\left(\hat{\gamma}_N\right)}\right] =
	\notag\\
	&\qquad
	\d_{1,N}^*\R_{N}(z_1) \d_{2,N} \Ebb\left[\chi_N \erm^{\irm u \sqrt{N}\Re\left(\hat{\gamma}_N\right)}\right]
	- \frac{\irm u \sigma^2}{2 \sqrt{N}} \frac{1}{2 \pi \irm} \oint_{\partial \Rcal} 
 	\frac{\left(\beta_N(u,z_1,z_2) + \tilde{\beta}_N(u,z_1,z_2)\right) w'_N(z_2)}
 	{\left(1+\sigma^2 c_N \Ebb\left[\hat{m}_N(z_1) \chi_N\right]\right)\left(1+\sigma^2 c_N m_N(z_2)\right)} \drm z_2
 	\notag\\
 	&\qquad\qquad
 	+ \d_{1,N}^*\Deltabs_{N}(u,z_1) \R_{N}(z_2)\d_{2,N}  	
 	+ \Ebb\left[\d_{1,N}^*\Q_N(z_1)\d_{2,N}\erm^{\irm u \sqrt{N}\Re\left(\hat{\gamma}_N\right)} \right] \frac{\sigma^2}{N} \Tr \Deltabs_N(0,z_1)
 	+ \frac{\epsilon_N(u,z_1)}{N^p}.	
	\label{eq:mastereq}
\end{align}
where $\beta_N(u,z_1,z_2)$ is defined by
\begin{align}
	&\beta_N(u,z_1,z_2) = 
	\notag\\
	&\qquad\qquad
	\Ebb\left[\d_{1,N}^* \Q_N(z_2) \R_N(z_1)\d_{2,N} \d_{1,N}^*\Q_N(z_1) \B_N\Sigmabs_N^*\Q_N(z_2)\d_{2,N}\chi_N^2\erm^{\irm u \sqrt{N}\Re\left(\hat{\gamma}_N\right)}\right]
	\notag\\
	&\qquad\qquad\qquad\qquad
	+ \Ebb\left[\d_{1,N}^* \Q_N(z_2) \Sigmabs_N\Sigmabs_N^* \R_N(z_1) \d_{2,N} \d_{1,N}^*\Q_N(z_1) \Q_N(z_2)\d_{2,N}\chi_N^2\erm^{\irm u \sqrt{N}\Re\left(\hat{\gamma}_N\right)}\right],
	\notag
\end{align}
and $\tilde{\beta}_N(u,z_1,z_2)$ by
\begin{align}
	&\tilde{\beta}_N(u,z_1,z_2) = 
	\notag\\
	&\qquad\qquad
	\Ebb\left[\d_{2,N}^* \Q_N(z_2) \R_N(z_1)\d_{2,N} \d_{1,N}^*\Q_N(z_1) \B_N\Sigmabs_N^*\Q_N(z_2)\d_{1,N}\chi_N^2\erm^{\irm u \sqrt{N}\Re\left(\hat{\gamma}_N\right)}\right]
	\notag\\
	&\qquad\qquad\qquad\qquad
	+ \Ebb\left[\d_{2,N}^* \Q_N(z_2)  \Sigmabs_N\Sigmabs_N^* \R_N(z_1)\d_{2,N} \d_{1,N}^*\Q_N(z_1)\Q_N(z_2)\d_{1,N}\chi_N^2\erm^{\irm u \sqrt{N}\Re\left(\hat{\gamma}_N\right)}\right].
	\notag
\end{align}
Using corollaries \ref{corollary:tool2} and \ref{corollary:tool3} (in conjonction with Cauchy-Schwarz inequality), 
and the fact that $\Ebb[\chi_N \erm^{\irm u \Re\left(\hat{\gamma}_N\right)}]=\psi_N(u) + \Ocal\left(N^{-p}\right)$ 
(by dominated convergence theorem, see section \ref{section:regularization}), it is straightforward to show that
\begin{align}
	&\Ebb\left[\d_{1,N}^*\Q_N(z_1)\d_{2,N}\chi_N \erm^{\irm u \sqrt{N}\Re\left(\hat{\gamma}_N\right)}\right] =
	\notag\\
	&\qquad\qquad
	\d_{1,N}^*\T_{N}(z_1) \d_{2,N} \psi_N(u)
	- \frac{\irm u \sigma^2}{2 \sqrt{N}} \frac{1}{2 \pi \irm} \oint_{\partial \Rcal} 
 	\frac{\left(\beta_N(u,z_1,z_2) + \tilde{\beta}_N(u,z_1,z_2)\right) w'_N(z_2)}
 	{\left(1+\sigma^2 c_N m_N(z_1)\right)\left(1+\sigma^2 c_N m_N(z_2)\right)} \drm z_2
 	+ \frac{\epsilon_N(u,z_1)}{N}.	
	\label{eq:mastereq2}
\end{align}
The next step consists in decorrelating the different terms inside the expressions of $\beta_N$ and $\tilde{\beta}_N$. 
We have
\begin{align}
	&\Ebb\left[\d_{1,N}^* \Q_N(z_2) \R_N(z_1)\d_{2,N} \d_{1,N}^*\Q_N(z_1) \B_N\Sigmabs_N^*\Q_N(z_2)\d_{2,N}\chi_N^2\erm^{\irm u \sqrt{N}\Re\left(\hat{\gamma}_N\right)}\right]
	=
	\notag\\
	&
	\Ebb\left[\d_{1,N}^* \Q_N(z_2) \R_N(z_1)\d_{2,N}\chi_N\right] \Ebb\left[\d_{1,N}^*\Q_N(z_1) \B_N\Sigmabs_N^*\Q_N(z_2)\d_{2,N}\chi_N\right] \psi_N(u)
	\notag\\
	&
	+
	\Ebb
	\left[
		\d_{1,N}^* \left(\Q_N(z_2)-\Ebb\left[(\Q_N(z_2)\right]\right) \R_N(z_1)\d_{2,N} 
		\d_{1,N}^*\Q_N(z_1) \B_N\Sigmabs_N^*\Q_N(z_2)\d_{2,N}\chi_N^2\erm^{\irm u \sqrt{N}\Re\left(\hat{\gamma}_N\right)}
	\right]
	\notag\\
	&
	+
	\Ebb\left[\d_{1,N}^* \Q_N(z_2) \R_N(z_1)\d_{2,N}\chi_N\right]
	\Ebb
	\left[
	\d_{1,N}^* \left(\Q_N(z_1)\B_N\Sigmabs_N^*\Q_N(z_2)-\Ebb\left[\Q_N(z_1)\B_N\Sigmabs_N^*\Q_N(z_2)\right]\right)
		\d_{2,N}\chi_N^2\erm^{\irm u \sqrt{N}\Re\left(\hat{\gamma}_N\right)}
	\right],
	\notag
\end{align}
and using again corollaries \ref{corollary:tool2} and \ref{corollary:tool3}, we end up with
\begin{align}
	&\Ebb\left[\d_{1,N}^* \Q_N(z_2) \R_N(z_1)\d_{2,N} \d_{1,N}^*\Q_N(z_1) \B_N\Sigmabs_N^*\Q_N(z_2)\d_{2,N}\chi_N^2\erm^{\irm u \sqrt{N}\Re\left(\hat{\gamma}_N\right)}\right]
	=
	\notag\\
	&
	\Ebb\left[\d_{1,N}^* \Q_N(z_2) \R_N(z_1)\d_{2,N}\chi_N\right] \Ebb\left[\d_{1,N}^*\Q_N(z_1) \B_N\Sigmabs_N^*\Q_N(z_2)\d_{2,N}\chi_N\right] \psi_N(u)
	+
	\frac{\epsilon_N(u,z_1,z_2)}{\sqrt{N}}.
	\notag
\end{align}
In the same way, 
\begin{align}
	&\Ebb\left[\d_{1,N}^* \Q_N(z_2) \Sigmabs_N\Sigmabs_N^* \R_N(z_1) \d_{2,N} \d_{1,N}^*\Q_N(z_1) \Q_N(z_2)\d_{2,N}\chi_N^2\erm^{\irm u \sqrt{N}\Re\left(\hat{\gamma}_N\right)}\right]
	=
	\notag\\
	&\qquad\qquad\qquad\qquad
	\Ebb\left[\d_{1,N}^* \Q_N(z_2) \Sigmabs_N\Sigmabs_N^* \R_N(z_1) \d_{2,N} \chi_N\right] \Ebb\left[\d_{1,N}^*\Q_N(z_1) \Q_N(z_2)\d_{2,N}\chi_N\right] \psi_N(u)
	+
	\frac{\epsilon_N(u,z_1,z_2)}{\sqrt{N}}.
	\notag
\end{align}
A standard application of corollaries \ref{corollary:tool1} \ref{corollary:tool2} and \ref{corollary:tool3} leads to 
\begin{align}
	&\Ebb\left[\d_{1,N}^*\Q_N(z_1) \B_N\Sigmabs_N^*\Q_N(z_2)\d_{2,N}\chi_N\right] = 
	\notag \\
	&
	\frac{\Ebb\left[\d_{1,N}^*\Q_N(z_1) \B_N\B_N^*\Q_N(z_2)\d_{2,N}\chi_N\right]}{1+\sigma^2 c_N  \Ebb\left[\hat{m}_N(z_2)\right]}
	- \frac{\Ebb\left[\d_{1,N}^* \Q_N(z_1)\Q_N(z_2)\d_{2,N}\right] \frac{\sigma^2}{N} \Tr \B_N^*\Ebb\left[\Q_N(z_1) \chi_N \right]\B_N}
	{\left(1+\sigma^2 c_N  \Ebb\left[\hat{m}_N(z_1)\right]\right)\left(1+\sigma^2 c_N  \Ebb\left[\hat{m}_N(z_2)\right]\right)}
	\notag\\
	&	
	 + \frac{\epsilon_N(u,z_1,z_2)}{N}
	\notag
\end{align}
and
\begin{align}
	&\Ebb\left[\d_{1,N}^* \Q_N(z_2) \Sigmabs_N\Sigmabs_N^* \R_N(z_1) \d_{2,N} \chi_N\right]
	=
	\notag\\
	&
	\frac{\Ebb\left[\d_{1,N} \Q_N(z_2) \B_N\B_N^* \R_N(z_1) \d_{2,N}\right]}{1+\sigma^2 c_N  \Ebb\left[\hat{m}_N(z_2)\right]}
	\notag\\
	&
	- 
	\left(
		\frac
		{
			\frac{\sigma^2}{N} \Tr \B_N^*\Ebb\left[\Q_N(z_2)\right]\B_N
		}
		{\left(1+\sigma^2 c_N  \Ebb\left[\hat{m}_N(z_2)\right]\right)^2}
		+
		\frac{\sigma^2}{1+\sigma^2 c_N  \Ebb\left[\hat{m}_N(z_2)\right]}
	\right) \Ebb[\d_{1,N}^* \Q_N(z_2) \R_N(z_1) \d_{2,N}]
	+ \frac{\epsilon_N(u,z_1,z_2)}{N^{3/2}}.
	\label{eq:useful_bil_form}
\end{align}
Inserting the previous estimates into the expressions of $\beta_N(u,z_1,z_2)$ and $\tilde{\beta}_N(u,z_1,z_2)$, and replacing $\Ebb[\hat{m}_N(z)]$ by $m_N(z)$ as well as 
$\Ebb[\Q_N(z)]$ and $\R_N(z)$ by $\T_N(z)$ thanks to corollary \ref{corollary:tool2}, we finally obtain \eqref{eq:charac_fun_1}.

	\subsubsection{Proof of proposition \ref{proposition:bil_form}}
	\label{appendix:bil_form}

Since the proof of proposition \ref{proposition:bil_form} uses the same technic as in the proof of formula \eqref{eq:charac_fun_1} (see appendix \ref{appendix:gaussiancomp1}), we will only
provide the main lines of the computations.

Let $\M_N(z_1,z_2)$ be a $M \times M$ deterministic matrix  s.t. 
\begin{align}
	\limsup_{N \to \infty} \sup_{(z_1,z_2) \in \partial \Rcal \times \partial \Rcal} \|\M_N(z_1,z_2)\| <\infty.
	\notag
\end{align}
We will also use the generic notation $\E_N(z_1,z_2)$ for $M \times M$ matrices such that
\begin{align}
	\limsup_{N \to \infty} \sup_{(z_1,z_2) \in \partial \Rcal \times \partial \Rcal} \left|\Tr \E_N(z_1,z_2)\right| < \infty,
	\notag
\end{align}
i.e. such that $\Tr \E_N(z_1,z_2)$ behaves as $\epsilon_N(z_1,z_2)$. The value of $\E_N(z_1,z_2)$ may change from one line to another.

Starting from the matrix $\Ebb\left[\Q_N(z_1)\M_N(z_1,z_2)\Q_N(z_2) \chi_N\right]$, a repeated use of corollaries \ref{corollary:tool1} \ref{corollary:tool2} and \ref{corollary:tool3} 
together with the decorrelation trick, in the same way as in appendix \ref{appendix:gaussiancomp1}, leads to
\begin{align}
	&\Ebb
	\left[
		\frac{\Q_N(z_1)\M_N(z_1,z_2)\Q_N(z_2)}{z_2\left(1+\sigma^2 c_N \Ebb\left[\hat{m}_N(z_2)\right]\right)} 
		\left(
			z_2\left(1+\sigma^2 c_N \Ebb\left[\hat{m}_N(z_2)\right]\right)\I - \B_N\B_N^* 
			- \sigma^2\left(\I - \frac{\frac{\sigma^2}{N} \Tr \B_N^* \Ebb\left[\Q_N(z_2)\right]\B_N}{1+\sigma^2 c_N \Ebb\left[\hat{m}_N(z_2)\right]}\right)
		\right)
		\chi_N
	\right] = 
	\notag\\
	&\qquad
	-\frac{1}{z_2} \Ebb\left[\left[\Q_N(z_1) \M_N(z_1)\right] \chi_N\right]
	- \frac{\Ebb\left[\Q_N(z_1)\Sigmabs_N\Sigmabs_N^* \chi_N\right]}{z_2\left(1+\sigma^2 c_N \Ebb\left[\hat{m}_N(z_2)\right]\right)}
	\Ebb\left[\frac{\sigma^2}{N} \Tr \Q_N(z_1)\M_N(z_1,z_2)\Q_N(z_2) \chi_N\right]
	\notag\\
	&\qquad\qquad\qquad
	- \frac{\Ebb\left[\Q_N(z_1) \chi_N\right]}{z_2\left(1+\sigma^2 c_N \Ebb\left[\hat{m}_N(z_1)\right]\right)\left(1+\sigma^2 c_N \Ebb\left[\hat{m}_N(z_2)\right]\right)}
	\Ebb\left[\frac{\sigma^2}{N} \Tr \Q_N(z_1)\M_N(z_1,z_2)\Q_N(z_2)\B_N\B_N^* \chi_N\right]
	\notag\\
	&\qquad\qquad\qquad\qquad
	+ \frac{\Ebb\left[\Q_N(z_1) \chi_N\right]\frac{\sigma^2}{N} \Tr \B_N^* \Ebb\left[\Q_N(z_2)\right]\B_N}
	{z_2\left(1+\sigma^2 c_N \Ebb\left[\hat{m}_N(z_1)\right]\right)\left(1+\sigma^2 c_N \Ebb\left[\hat{m}_N(z_2)\right]\right)^2}
	\Ebb\left[\frac{\sigma^2}{N} \Tr \Q_N(z_1)\M_N(z_1,z_2)\Q_N(z_2)\chi_N\right]
	\notag\\
	&\qquad\qquad\qquad\qquad\qquad
	+ \frac{\E_N(z_1,z_2)}{N^2}.
	\notag
\end{align}
By introducing the matrix $\R_N(z_2)$ and as in appendix \ref{appendix:gaussiancomp1}, we obtain
\begin{align}
	&\Ebb\left[\Q_N(z_1)\M_N(z_1,z_2)\Q_N(z_2)\right] = 
	\notag\\
	&\qquad
	\Ebb\left[\left[\Q_N(z_1) \M_N(z_1) \R_N(z_2) \chi_N\right] \chi_N\right]
	+ \frac{\Ebb\left[\Q_N(z_1)\Sigmabs_N\Sigmabs_N^* \R_N(z_2)\chi_N\right]}{1+\sigma^2 c_N \Ebb\left[\hat{m}_N(z_2)\right]}
	\Ebb\left[\frac{\sigma^2}{N} \Tr \Q_N(z_1)\M_N(z_1,z_2)\Q_N(z_2) \chi_N\right]
	\notag\\
	&\qquad\qquad\qquad
	+ \frac{\Ebb\left[\Q_N(z_1) \chi_N\right]}{\left(1+\sigma^2 c_N \Ebb\left[\hat{m}_N(z_1)\right]\right)\left(1+\sigma^2 c_N \Ebb\left[\hat{m}_N(z_2)\right]\right)}
	\Ebb\left[\frac{\sigma^2}{N} \Tr \Q_N(z_1)\M_N(z_1,z_2)\Q_N(z_2)\B_N\B_N^* \chi_N\right]
	\notag\\
	&\qquad\qquad\qquad\qquad
	- \frac{\Ebb\left[\Q_N(z_1) \chi_N\right]\frac{\sigma^2}{N} \Tr \B_N^* \Ebb\left[\Q_N(z_2)\right]\B_N}
	{\left(1+\sigma^2 c_N \Ebb\left[\hat{m}_N(z_1)\right]\right)\left(1+\sigma^2 c_N \Ebb\left[\hat{m}_N(z_2)\right]\right)^2}
	\Ebb\left[\frac{\sigma^2}{N} \Tr \Q_N(z_1)\M_N(z_1,z_2)\Q_N(z_2)\chi_N\right]
	\notag\\
	&\qquad\qquad\qquad\qquad\qquad
	+ \frac{\E_N(z_1,z_2)}{N^2},
	\notag
\end{align}
By taking the trace in the previous expression, and using corollaries \ref{corollary:tool1} \ref{corollary:tool2} and \ref{corollary:tool3}, we end up with the following 
$2 \times 2$ linear system
\begin{align}
	&
	\begin{bmatrix}
		\Ebb\left[\frac{\sigma^2}{N} \Tr \Q_N(z_1) \M_N(z_1,z_2) \Q_N(z_2) \chi_N\right]
		\\
		\Ebb
		\left[
			\frac{\sigma^2}{N} \Tr \frac{\Q_N(z_1) \M_N(z_1,z_2) \Q_N(z_2) \B_N\B_N^*}
			{\left(1+\sigma^2 c_N m_N(z_1)\right)\left(1+\sigma^2 c_N m_N(z_2)\right)} \chi_N
		\right]
	\end{bmatrix}
	=
	\notag\\
	&\qquad\qquad
	\begin{bmatrix}
		v_N(z_1,z_2) s_N(z_1,z_2) + u_N(z_1,z_2) & v_N(z_1,z_2)
		\\
		\\
		u_N(z_1,z_2) s_N(z_1,z_2) + r_N(z_1,z_2) & u_N(z_1,z_2)
	\end{bmatrix}
	\begin{bmatrix}
		\Ebb\left[\frac{\sigma^2}{N} \Tr \Q_N(z_1) \M_N(z_1,z_2) \Q_N(z_2) \chi_N\right]
		\\
		\Ebb
		\left[
			\frac{\sigma^2}{N} \Tr \frac{\Q_N(z_1) \M_N(z_1,z_2) \Q_N(z_2)\B_N\B_N^*}
			{\left(1+\sigma^2 c_N m_N(z_1)\right)\left(1+\sigma^2 c_N m_N(z_2)\right)} \chi_N
		\right]
	\end{bmatrix}
	\notag\\
	&\qquad\qquad\qquad\qquad
	+
	\begin{bmatrix}
		\frac{\sigma^2}{N} \Tr \T_N(z_1) \M_N(z_1,z_2) \T_N(z_2)
		\\
		\frac{\sigma^2}{N} \Tr \T_N(z_1) \M_N(z_1,z_2) \T_N(z_2)\B_N\B_N^*
	\end{bmatrix}
	+
	\frac{1}{N^2}
	\begin{bmatrix}
		\epsilon_{N}(z_1,z_2)
		\\
		\epsilon_{N}(z_1,z_2)
	\end{bmatrix}
	,
	\label{eq:lin_sys}
\end{align}
where $u_N(z_1,z_2), v_N(z_1,z_2), r_N(z_1,z_2)$ and $s_N(z_1,z_2)$ are respectively defined in \eqref{def:u}, \eqref{def:v_vtilde}, \eqref{def:r} and \eqref{def:s}.
The determinant of the previous system is given by
\begin{align}
	\Delta_N(z_1,z_2) = \left(1-u_N(z_1,z_2)\right)^2 - v_N(z_1,z_2)\left(s_N(z_1,z_2) + r_N(z_1,z_2)\right).
	\notag
\end{align}
By relating $\T_N(z)$ with $\tilde{\T}_N(z)$, we obtain the equality
\begin{align}
	\frac{\B_N^*\T_N(z)\B_N}{1+\sigma^2 c_N m_N(z)} = \I + z\left(1+\sigma^2 c_N m_N(z)\right) \tilde{\T}_N(z).
	\label{eq:link_T_Ttilde}
\end{align}
Inserting relation \eqref{eq:link_T_Ttilde} in the expressions of $r_N(z_1,z_2)$ and $s_N(z_1,z_2)$, we obtain respectively
\begin{align}
	s_N(z_1,z_2) = 
		\frac{-\sigma^2}{\left(1+\sigma^2 c_N m_N(z_1)\right)\left(1+\sigma^2 c_N m_N(z_2)\right)}
		- z_1 \frac{\sigma^2}{N} \Tr \frac{\T_N(z_1)}{1+\sigma^2 c_N m_N(z_2)}
		- z_2 \frac{\sigma^2}{N} \Tr \frac{\T_N(z_2)}{1+\sigma^2 c_N m_N(z_1)}
	\notag
\end{align}
and
\begin{align}
	&r_N(z_1,z_2) = 
	\notag\\
	&\quad
		z_1 z_2 \tilde{v}_N(z_1,z_2)
		-\frac{-\sigma^2}{\left(1+\sigma^2 c_N m_N(z_1)\right)\left(1+\sigma^2 c_N m_N(z_2)\right)}
		+ z_1 \frac{\sigma^2}{N} \Tr \frac{\T_N(z_1)}{1+\sigma^2 c_N m_N(z_2)}
		+ z_2 \frac{\sigma^2}{N} \Tr \frac{\T_N(z_2)}{1+\sigma^2 c_N m_N(z_1)},
		\notag
\end{align}
where $\tilde{v}_N(z_1,z_2)$ is defined by \eqref{def:v_vtilde}.
The determinant thus writes 
\begin{align}
	\Delta_N(z_1,z_2) = \left(1-u_N(z_1,z_2)\right)^2 - z_1 z_2 v_N(z_1,z_2)\tilde{v}_N(z_1,z_2).
	\notag
\end{align}
Using lemma \ref{lemma:properties_determinant}, we can finally solve the system \eqref{eq:lin_sys} to obtain
\begin{align}
	&\Ebb\left[\frac{\sigma^2}{N} \Tr \Q_N(z_1) \M_N(z_1,z_2) \Q_N(z_2) \chi_N\right] =
	\notag\\
	&\qquad
	\frac
	{
		\left(1 - u_N(z_1,z_2)\right)\frac{\sigma^2}{N}\Tr \T_N(z_1)\M_N(z_1,z_2)\T_N(z_2) 
		+ v_N(z_1,z_2) \frac{\sigma^2}{N} \Tr \frac{\T_N(z_1) \M_N(z_1,z_2) \T_N(z_2) \B_N\B_N^*}{\left(1+\sigma^2 c_N m_N(z_1)\right) \left(1+\sigma^2 c_N m_N(z_1)\right)} 
	}
	{\Delta_N(z_1,z_2)}
	+ \frac{\epsilon_N(z_1,z_2)}{N^2}
	\notag
\end{align}
and
\begin{align}
	&\Ebb\left[\frac{\sigma^2}{N} \Tr \frac{\Q_N(z_1) \M_N(z_1,z_2) \Q_N(z_2) \B_N\B_N^*}{\left(1+\sigma^2 c_N m_N(z_1)\right)\left(1+\sigma^2 c_N m_N(z_2)\right)} \chi_N\right] =
	\notag\\
	&\qquad\frac
	{
		\left(u_N(z_1,z_2) s_N(z_1,z_2) + r_N(z_1,z_2)\right) \frac{\sigma^2}{N} \Tr \T_N(z_1)\M_N(z_1,z_2)\T_N(z_2)			
	}
	{\Delta_N(z_1,z_2)}
	\notag\\
	&\qquad + 
	\frac
	{
		\left(1 - v_N(z_1,z_2) s_N(z_1,z_2) - u_N(z_1,z_2)\right) 
		\frac{\sigma^2}{N} \Tr \frac{\T_N(z_1)\M_N(z_1,z_2)\T_N(z_2) \B_N\B_N^*}{\left(1+\sigma^2 c_N m_N(z_1)\right)\left(1+\sigma^2 c_N m_N(z_2)\right)}
	}
	{\Delta_N(z_1,z_2)}
	+ \frac{\epsilon_N(z_1,z_2)}{N^2}.
	\notag
\end{align}
The approximations \eqref{eq:approx_doubleQ1} and \eqref{eq:approx_doubleQ2} will unfold by choosing $\M_N(z_1,z_2)=\d_{1,N}\d_{2,N}^*$, whcih concludes 
the proof of proposition \ref{proposition:bil_form}.

	\subsection{Proof of lemma \ref{lemma:properties_determinant}}
	\label{appendix:determinant}

We recall here that $\Delta_N(z_1,z_2)$ is defined \eqref{def:determinant} by
\begin{align}
	\Delta_N(z_1,z_2) = \left(1-u_N(z_1,z_2)\right)^2 - z_1 z_2 v_N(z_1,z_2) \tilde{v}_N(z_1,z_2),
	\notag
\end{align}
with $u_N(z_1,z_2)$, $v_N(z_1,z_2)$ and $\tilde{v}_N(z_1,z_2)$ given by \eqref{def:u} and \eqref{def:v_vtilde}.
We also recall the following bounds from \cite{hachem2012subspace} :
\begin{align}
	\sup_{N \to \infty} \sup_{z \in \Kcal} \left|u_N(z,z^*)\right| < 1
	\quad\text{and}\quad
	\inf_{N \to \infty} \inf_{z \in  \Kcal} \left|\Delta_N(z,z^*)\right| > 0.
	\label{eq:bound_det_old}
\end{align}
For $z_1,z_2 \in \partial \Kcal$, Cauchy-Schwarz inequality and \eqref{eq:bound_det_old} gives
\begin{align}
	\left|u_N(z_1,z_2)\right| \leq \left|u_N(z_1,z_1^*)\right|^{1/2} \left|u_N(z_2,z_2^*)\right|^{1/2},
	\label{eq:bound_det2}
\end{align}
and thus
\begin{align}
	\limsup_{N \to \infty} \sup_{z_1,z_2 \in \Kcal} \left|u_N(z_1,z_2)\right| < 1,
	\notag
\end{align}
which shows \eqref{eq:bound_u}.
In the same way, since
\begin{align}
	\left|\Delta_N(z_1,z_2)\right| \geq \left(1 - \left|u_N(z_1,z_2)\right|\right)^2 - |z_1||z_2|\left|v_N(z_1,z_2)\right|\left|\tilde{v}_N(z_1,z_2)\right|,
	\notag
\end{align}
a straightforward application of Cauchy-Schwarz inequality yields
\begin{align}
	&\left|\Delta_N(z_1,z_2)\right| 
	\geq 
	\notag\\
	&\quad
	\left(1 - \left|u_N(z_1,z_1^*)\right|^{1/2}\left|u_N(z_2,z_2^*)\right|^{1/2}\right)^2 
	- |z_1||z_2|\left|v_N(z_1,z_1^*)\right|^{1/2}\left|v_N(z_2,z_2^*)\right|^{1/2}\left|\tilde{v}_N(z_1,z_1^*)\right|^{1/2}\left|\tilde{v}_N(z_1,z_1^*)\right|^{1/2}.
	\notag
\end{align}
Using the inequality 
\begin{align}
	\left(1-\sqrt{x_1 x_2}\right)^2 \geq \sqrt{\left(1-x_1\right)^2 - s_1 t_1}\sqrt{\left(1-x_2\right)^2 - s_2 t_2},
	\notag
\end{align}
valid for $x_i \in [0,1]$ and $(1-x_i)^2 \geq s_i t_i$ ($i=1,2$), we finally obtain
\begin{align}
	\left|\Delta_N(z_1,z_2)\right|  \geq \left|\Delta_N(z_1,z_1^*)\right|^{1/2} \left|\Delta_N(z_2,z_2^*)\right|^{1/2},
\end{align}
which readily implies
\begin{align}
	\inf_{N \to \infty} \inf_{z_1,z_2 \in \Kcal} \left|\Delta_N(z_1,z_2)\right| > 0.
	\notag
\end{align}
Moreover, from \eqref{eq:bound_norm_T_Ttilde}, and the definition of $u_N(z_1,z_2)$, $v_N(z_1,z_2)$ and $\tilde{v}_N(z_1,z_2)$, we also see that
\begin{align}
	\left|\Delta_N(z_1,z_2)\right| \leq 
	\Prm\left(\frac{1}{\drm\left(z_1,\supp(\mu_N) \cup \{0\}\right)}\right)\Qrm\left(\frac{1}{\drm\left(z_2,\supp(\mu_N) \cup \{0\}\right)}\right),
	\label{eq:upper_bound_det}
\end{align}
where $\Prm, \Qrm$ are two polynomials independent of $N,z_1,z_2$ with positive coefficients, and we thus deduce
\begin{align}
	\sup_{N \to \infty} \sup_{z_1,z_2 \in \Kcal} \left|\Delta_N(z_1,z_2)\right| < \infty,
	\notag
\end{align}
which shows \eqref{eq:bound_det}.
We now prove \eqref{eq:link_Delta_w}. By straightforward computations, it is easily shown that $\Delta_N(z_1,z_2)$ is the determinant of the following $2 \times 2$ linear system
\begin{align}
	\begin{bmatrix}
		\sigma c_N \left(m_N(z_1) - m_N(z_2)\right)
		\\
		\sigma\left(z_1 \tilde{m}_N(z_1) - z_2 \tilde{m}_N(z_2)\right)
	\end{bmatrix}
	=
	\begin{bmatrix}
		u_N(z_1,z_2) & v_N(z_1,z_2)
		\\
		z_1 z_2 \tilde{v}_N(z_1,z_2) & u_N(z_1,z_2)
	\end{bmatrix}
	\begin{bmatrix}
		\sigma c_N \left(m_N(z_1) - m_N(z_2)\right)
		\\
		\sigma\left(z_1 \tilde{m}_N(z_1) - z_2 \tilde{m}_N(z_2)\right)
	\end{bmatrix}
	+
	\frac{z_1 - z_2}{\sigma} 
	\begin{bmatrix}
		v_N(z_1,z_2)
		\\
		u_N(z_1,z_2)
	\end{bmatrix},
	\notag
\end{align}
Since $\Delta_N(z_1,z_2) \neq 0$ for $z_1,z_2 \in \Kcal$, solving the previous linear system gives
\begin{align}
	\sigma c_N \left(m_N(z_1) - m_N(z_1)\right) = \frac{z_1 - z_2}{\Delta_N(z_1,z_2)} \frac{v_N(z_1,z_2)}{\sigma},
	\notag
\end{align}
and it is easy to show that
\begin{align}
	\sigma c_N \left(m_N(z_1) - m_N(z_1)\right) = \left(w_N(z_1) - w_N(z_2)\right) \frac{v_N(z_1,z_2)}{\sigma}.
	\notag
\end{align}
Thus we obtain the relation
\begin{align}
	w_N(z_1) - w_N(z_2) = \frac{z_1 - z_2}{\Delta_N(z_1,z_2)}
	\notag
\end{align}
which shows \eqref{eq:link_Delta_w}.
Finally, to prove \eqref{eq:bound_det_serie}, we write, using Cauchy-Schwarz inequality,
\begin{align}
	\left|\Delta_N(z_1,z_2) -  \left(1-u_N(z_1,z_2)\right)^2\right|
	&=
	\left|z_1\right| \left|z_2\right| \left|v_N(z_1,z_2)\tilde{v}_N(z_1,z_2)\right|
	\notag\\
	&\leq
	\left|z_1\right| \left|z_2\right| 
	v_N(z_1,z_1^*)^{1/2} \tilde{v}_N(z_2,z_2^*)^{1/2} 
	v_N(z_1,z_1^*)^{1/2} \tilde{v}_N(z_2,z_2^*)^{1/2},
	\notag
\end{align}
and since $\left(1-u_N(z,z^*)\right)^2 > |z|^2 v_N(z,z^*)\tilde{v}_N(z,z^*)$ for $z \in \Kcal$, we have
\begin{align}
	\left|\Delta_N(z_1,z_2) -  \left(1-u_N(z_1,z_2)\right)^2\right| < \left(1-u_N(z_1,z_1^*)\right) \left(1-u_N(z_2,z_2^*)\right).
	\notag
\end{align}
As $|1 - u_N(z_1,z_2)| \geq 1 - |u_N(z_1,z_1^*)|^{1/2} |u_N(z_2,z_2^*)|^{1/2}$ and thanks to the inequality $\sqrt{(1-a)(1-b)} \leq 1 - \sqrt{ab}$ valid for $a,b \in [0,1]$, 
we finally obtain
\begin{align}
	\left|\Delta_N(z_1,z_2) -  \left(1-u_N(z_1,z_2)\right)^2\right| < \left|1-u_N(z_1,z_2)\right|^2, 
	\notag
\end{align}
or equivalently
\begin{align}
	\left|\frac{\Delta_N(z_1,z_2)}{\left(1-u_N(z_1,z_2)\right)^2} - 1\right| < 1.
	\notag
\end{align}

	\subsection{Proof of lemma \ref{lemma:conseq_sep2}}
	\label{appendix:proof_lemma_conseq_sep_2}

	Assume that the separation condition {\bf A-\ref{assumption:subspace1}} and {\bf A-\ref{assumption:subspace2}} hold, and let $t \in (t_1^+,t_2^-)$.
	Since $\mu_N$ converges to the Marchenko-Pastur distribution, $w_N(t) > 0 \to w(t) \geq 0$ and we deduce that $t > \sigma^2\left(1+\sqrt{c}\right)^2$.
	From \eqref{eq:canonical_eq_mp} and the behaviour of $\phi$, we have $w(t) > \sigma^2 \sqrt{c}$, and finally 
	assumption {\bf A-\ref{assumption:subspace2}} implies that $\liminf_N \lambda_{K,N} \geq w(t)$,
	which proves \eqref{eq:sep_cond_spike}.
	
	Now, assume that \eqref{eq:sep_cond_spike} holds and let $\epsilon > 0$ such that 
	\begin{align}
		\liminf_{N \to \infty} \lambda_{K,N} > \sigma^2 \sqrt{c} + \epsilon.
	\end{align}
	For any compact $\Kcal \subset \left(-\infty, 0\right)\cup\left(0, \sigma^2 \sqrt{c}+\epsilon\right)$, we have
	\begin{align}
		\sup_{w \in \Kcal} \left|f_N(w) + \frac{1}{w}\right| \xrightarrow[N \to \infty]{} 0,
		\quad
		\sup_{w \in \Kcal} \left|\phi_N(w) - \phi(w)\right| \xrightarrow[N \to \infty]{} 0,
		\quad\text{and}\quad
		\sup_{w \in \Kcal} \left|\phi'_N(w) - \phi'(w)\right| \xrightarrow[N \to \infty]{} 0
		\label{eq:conv_unif}
	\end{align}
	Since $\phi(w)$ has a unique maximum $\sigma^2 (1-\sqrt{c})^2$ at point $w=-\sigma^2 \sqrt{c}$ on the interval $(-\infty, 0)$, 
	$\phi_N$ will also admit a positive maximum in this interval for all large $N$, and thus
	\begin{align}
		x_{1,N}^-= \sigma^2 (1-\sqrt{c})^2 + o(1) \quad\text{and}\quad w_N(x_{1,N}^-) = -\sigma^2 \sqrt{c} + o(1).
		\label{eq:convergence_w1-}
	\end{align}
	In the same way, $\phi$ has a unique positive minimum $\sigma^2 (1+\sqrt{c})^2$ at $\sigma^2 \sqrt{c}$ on the interval 
	$\left(0, \sigma^2 \sqrt{c}+\epsilon\right)$, and thus $\phi_N$ will also admit a positive minimum on this interval, 
	at the point $w_N(x_{1,N}^+)$, for $N$ large enough, and
	\begin{align}
		x_{1,N}^+ = \sigma^2 (1+\sqrt{c})^2 + o(1) \quad\text{and}\quad w_N(x_{1,N}^+) = \sigma^2 \sqrt{c} + o(1).
		\label{eq:convergence_w1+}
	\end{align}
	Therefore, we can find $t_1^-$ such that $\liminf_N x_{1,N}^- > t_1^- > 0$.
	Moreover, if $\Kcal'$ is a compact included in $\left(\sigma^2 \sqrt{c}, \sigma^2 \sqrt{c}+\epsilon\right)$, \eqref{eq:conv_unif} also implies that for $N$ 
	large enough,
	\begin{align}
		\inf_{w\in \Kcal'} 1 - \sigma^2 c_N f_N(w) > 0\quad\text{and}\quad\inf_{w \in \Kcal'} \phi'_N(w) > 0.
		\label{eq:cond_phi_der}
	\end{align}
	which proves that $\phi_N(\Kcal') \subset \Rbb \backslash \supp(\mu_N)$ from \cite[Lemma 6]{vallet2012improved}.
	This shows that it necessarily exists a local maximum $x_{2,N}^-$ of $\phi_N$ with preimage $w_N(x_{2,N}^-) > \sup \Kcal'$.
	By fixing two points $t_2^- > t_1^+ > \sigma^2 (1+\sqrt{c})^2$ such that $w(t_1^+), w(t_2^-) \in \mathrm{Int}\left(\Kcal'\right)$, 
	we easily conclude that 
	\begin{align}
		\limsup_{N\to\infty} x_{1,N}^+ < t_1^+ \quad\text{and}\quad \liminf_{N \to \infty} x_{2,N}^- > t_2^-, 
	\end{align}
	which proves {\bf A-1}.
	By definition of $t_2^-$, $w(t_2^-) < \sigma^2 \sqrt{c} + \epsilon$, which implies of course {\bf A-2}.

\bibliographystyle{plain}
\bibliography{mainbib} 

\end{document}